\pdfoutput=1

\documentclass[a4paper,DIV=12]{scrartcl}

\usepackage{ifthen}
\ifthenelse{\isundefined{\development}}{
}{
  \overfullrule=1mm
}

\usepackage[T1]{fontenc}
\usepackage[utf8]{inputenc}

\usepackage{amsthm}

\usepackage{microtype}
\usepackage{etoolbox}

\usepackage{libertine}
\usepackage[libertine]{newtxmath}
\usepackage{amsmath}

\usepackage[numbers]{natbib}
\usepackage{doi}

\usepackage{authblk}

\usepackage{bussproofs} % consider replacing by ebproof
\usepackage{empheq}
\usepackage{proof}
\usepackage[f]{esvect}
\usepackage{mathtools}
\usepackage{hhline}
\usepackage{tikz}
\usetikzlibrary{positioning}
\usepackage{tikz-cd}
\usepackage{tikz-qtree}
\usepackage{listings}
\usepackage{comment}
\usepackage{soul}
\usepackage{cancel}
\usepackage{calc}
\usepackage[inline]{enumitem}
\usepackage{color}

\newenvironment{lemmarep}{\begin{lemma}}{\end{lemma}}
\newenvironment{theoremrep}{\begin{theorem}}{\end{theorem}}
\newenvironment{propositionrep}{\begin{proposition}}{\end{proposition}}
\newenvironment{appendixproof}{\begin{proof}}{\end{proof}}

\newenvironment{toappendix}{}{}

\usepackage{local-macros}
\usepackage{references}
\usepackage{thm}

\ifthenelse{\isundefined{\development}}{
  \usepackage[disable]{todonotes}
}{
  \usepackage[textwidth=1.8cm,textsize=footnotesize,colorinlistoftodos]{todonotes}
  \setlength{\marginparwidth}{1.8cm}
  \usepackage[normalem]{ulem}
  \usepackage[notref,notcite]{showkeys}
  \renewcommand*{\showkeyslabelformat}[1]{%
    \fbox{\parbox[t]{1.6cm}{\raggedright\normalfont\tiny{\url{#1}}}}}
  \usepackage{gitinfo}
}

\usepackage{hyperref}

\newcommand{\knote}[2][]{\todo[inline, color=blue!10   ,#1]{#2}}
\newcommand{\hnote}[2][]{\todo[inline, color=yellow!10 ,#1]{#2}}

\specialcomment{longproof}{\begin{proof}}{\end{proof}}
\specialcomment{shortproof}{\begin{proof}}{\end{proof}}
\ifthenelse{\isundefined{\showlongproofs}}{
  \excludecomment{longproof}
}{
  \excludecomment{shortproof}
}

\begin{document}

\title{Coinduction in Uniform}
\subtitle{Foundations for Corecursive Proof Search with Horn Clauses}

\newcommand{\email}[1]{%
  \textit{\href{mailto:#1}{#1}}%
}
\newcommand{\orcid}[1]{%
  \textit{\href{https://orcid.org/#1}{#1}}%
}

\author{Henning Basold\thanks{This work is supported by the European Research
    Council (ERC) under the EU’s Horizon 2020 programme~(CoVeCe, grant
    agreement No. 678157)}}
\affil{CNRS, ENS Lyon
  (email \email{henning.basold@ens-lyon.fr}, orcid \orcid{0000-0001-7610-8331})}

\author{Ekaterina Komendantskaya\thanks{This work is supported by EPSRC research
    grant EP/N014758/1.}}
\affil{Heriot-Watt University
  (email \email{ek19@hw.ac.uk})}

\author{Yue Li}
\affil{Heriot-Watt University
  (email \email{yl55@hw.ac.uk}, orcid \orcid{0000-0003-0545-0732})}

\date{November 2018}

%% Keywords
%% comma separated list
% \keywords{Horn Clause Logic, Coinduction, Uniform Proofs, Intuitionsitic First-Order Sequent Calculus, Coalgebra}  %% \keywords is optional

\maketitle

\begin{abstract}
  We establish proof-theoretic, constructive and coalgebraic foundations for
  proof search in coinductive Horn clause theories.
  Operational semantics of coinductive Horn clause resolution is cast in
  terms of \emph{coinductive uniform proofs}; its constructive content is
  exposed via soundness relative to an intuitionistic first-order logic with
  recursion controlled by the later modality;
  and soundness of both proof systems is proven relative to a novel
  coalgebraic description of complete Herbrand models.
 % New heuristics for automated coinduction hypothesis discovery are presented
 % as an application of the theoretical results.
 % These heuristics are more powerful than all known ones.
\end{abstract}

\section{Introduction}
\label{sec:intro}

%Let us illustrate the need for more \emph{generic} and \emph{constructive} corecursive proof search methods by considering

\emph{Horn clause logic} is a Turing complete and constructive fragment of
first-order logic, that plays a central role in
verification~\cite{BjornerGMR15},
% description of the behaviour of, e.g., internet servers~\cite{Dav01},
automated theorem proving~\cite{Llo87,NadathurM99,MN12} and type inference.
Examples of the latter can be traced from the Hindley-Milner type inference
algorithm~\cite{Milner78,SulzmannS08}, to more recent uses of Horn clauses in
Haskell type classes~\cite{Lammel:2005,BottuKSOW17} and in refinement
types~\cite{HashimotoU15,BurnOR17}.
Its popularity can be attributed to well-understood fixed point semantics and
an efficient semi-decidable resolution procedure for automated proof search.

According to the standard fixed point semantics~\cite{vEmdK76,Llo87}, given a set
$P$ of Horn clauses, the
\emph{least Herbrand model} for $P$ is the  set of
all (finite) ground atomic formulae \emph{inductively entailed} by $P$.
For example,
%\begin{example}\label{ex:nat}
the two clauses below define the set of natural numbers in the least Herbrand
model.
\begin{align*}
  % \SwapAboveDisplaySkip
  \clNatZ &:   \nat \, 0  \\
  \clNatS &: \all{x}  \nat \, x  \impl  \nat \, (s\, x)
  %\kappa_{\stream 0} &: \all{x} \nat \, x \, \land \, \stream \, y \,  \impl \stream \, (\scons \, x \, y)
  % \label{eq:horn clause: nats}
\end{align*}
\noindent
Formally, the least Herbrand model  for the above two clauses is the set of ground atomic
formulae obtained by taking a (forward) closure of the above two clauses.
%seen as the inference rules $\frac{}{\nat \, 0}$, $\frac{\nat \, x}{\nat \, (s \, x)}$.
The model for $\nat$ is given by
$\natModel = \set{
  \nat \, 0, \, \nat\, (s \, 0),\, \nat \, (s\, (s \, 0)), \dotsc
}$.

% \begin{equation*}
%   \natModel = \set{
%     \nat \, 0, \, \nat\, (s \, 0),\, \nat \, (s\, (s \, 0)), \dotsc
%   }
% \end{equation*}

% \end{example}

We can also view Horn clauses coinductively.
The \emph{greatest complete Herbrand model} for a
set $P$ of Horn clauses is the largest set of finite and infinite
ground atomic formulae \emph{coinductively entailed} by $P$. For example, the
greatest complete Herbrand model for the above two clauses is the set
\begin{equation*}
  \natModel^{\infty} =
  \natModel \cup \set{\nat \, (s\, (s \, (\dotsb)))},
\end{equation*}
obtained by taking a backward closure of the
above two inference rules on the set of all finite and infinite ground atomic formulae.
%It will contain all
%of the formulae in $\natModel$ together
%with  $\nat \, (s\, (s \, (...)))$ representing the first limit
%ordinal.
The \emph{greatest Herbrand model}  is the largest set of \emph{finite}
ground atomic formulae \emph{coinductively entailed} by $P$.
In our example, it would be given by $\natModel$ already.
Finally, one can also consider the \emph{least complete Hebrand model}, which
interprets entailment inductively but over potentially infinite terms.
In the case of $\nat$, this interpretation does not differ from $\natModel$.
However, finite paths in coinductive structures like transition systems,
for example, require such semantics.

The need for coinductive semantics of Horn clauses arises in several scenarios:
the Horn clause theory may explicitely define a coinductive data structure or
a coinductive relation.
However, it may also happen that a Horn clause theory, which is not explicitly
intended as coinductive, nevertheless gives rise to infinite inference by resolution and has
an interesting coinductive  model.
%and operational
%behaviour.
This commonly happens in type inference.
We will illustrate all these cases by means of examples.

\paragraph{Horn clause theories as coinductive data type
  declarations}
The following clause defines, together with $\clNatZ$ and $\clNatS$, the
type of streams over natural numbers.
% The following three Horn clauses define predicates $nat$ and $stream$, and thereby the
% sets of natural numbers and infinite streams of natural numbers: % (we use $\kappa$s to label clauses):
\begin{equation*}
  \kappa_{\stream} :
  \all{x y} \nat \, x \, \land \, \stream \, y \,  \impl \stream \, (\scons \, x \, y)
\end{equation*}
This Horn clause does not have a meaningful inductive, i.e. least fixed point, model.
%This is because no finite term satisfies the property $\stream$.
 The greatest Herbrand model of the clauses is given by
 \begin{equation*}
   \streamModel =  \natModel^{\infty} \cup
   \setDef{
     \stream (\scons \, x_0 \, (\scons \, x_1 \, \dotsb))
   }{\nat \, x_0, \nat \, x_1, \dotsc \in \natModel^{\infty}}
\end{equation*}
% $ \stream (\scons\, 0, \ldots), \stream (\scons \, 0, \scons\, (s \, 0), \ldots), \ldots\}$.
% The first two clauses ofcourse resemble the standard inductive definition of
% the type natural numbers, and the third clause -- a coinductive definition of the type of streams.

In trying to prove, for example, the goal $(\stream \, x)$, a goal-directed
proof search may try to find a substitution for $x$ that
will make  $(\stream \, x)$ valid relative to the coinductive model of this set
of clauses.
This search by resolution may proceed by means of an infinite reduction
$\underline{\stream \, x}
\stackrel{\kappa_{\stream}:[\sconsFunc{y}{x'}/x]}{\sldRed} \nat\, y \land \stream \, x'
\stackrel{\clNatZ:[0/y]}{\sldRed}  \underline{\stream \, x'}
\stackrel{\kappa_{\stream}:[\sconsFunc{y'}{x''}/x']}{ \sldRed} \cdots$,
thereby generating a stream $Z$ of zeros  via composition of the computed
substitutions: $Z = (\scons \, 0 \, x')[\scons \, 0 \, x''/x'] \dotsm$.
Above, we annotated each resolution step with the label of the clause it resolves
against and the computed substitution.
A method to compute an answer for this infinite sequence of reductions was given
by Gupta~et~al.~\cite{GuptaBMSM07} and Simon~et~al.~\cite{SimonBMG07}:
the underlined loop gives rise to the circular unifier $x = \scons\, 0\, x$ that
corresponds to the infinite term $Z$. It is proven that, if a loop and a corresponding
circular unifier are detected, they provide an answer that is sound relative to the greatest complete Herbrand model of the clauses.
This approach is known under the name of CoLP.

%% We can explain this example in a journal version better.

% \paragraph{Horn clause theories for verification of potentially infinite processes}
% The following clauses define verification properties for concurrent behaviour of
% internet servers~\cite{Dav01}, where we use the notation $[\_ | \_]$ to denote
% the list or stream constructor:
% \begin{alignat*}{2}
% %  \SwapAboveDisplaySkip
%   \kappa_{\mathrm{res}}&: \all{x \, \inP \, l} \resource \, \inP  \, l
%   & & \impl \resource \, [\get(x)| \inP] \,  [x|l] \\
%   \kappa_{\mathrm{sig}}&: \all{x \, \inP} \signal \, (\novalue \, (\get \, x))
%   & & \impl  \resource \, [\get(x)| \inP] \, []
% \end{alignat*}

% The predicate $\resource$ is supposed to access the values of the incoming
% stream (by using the function ``$\get$''), as long as this is possible,
% and terminate if no value can be obtained.
% Since $\resource$ shall only terminate if the input stream in its first
% argument cannot provide values anymore, the above theory has to be interpreted
% coinductively, e.g. via the CoLP approach. %
% %\todo{Do we want to explain this more? Example resolution?} -- added a reference that it can be done via CoLP as we explained one para above

% We now consider two examples when coinductive properties of Horn clauses arise in type inference. %, and seemingly by accident:

\paragraph{Horn Clause Theories in Type Inference}
Below clauses give the typing rules of the simply typed $\lambda$-calculus, and may be used for  type inference or type checking:
\begin{align*}
%   \SwapAboveDisplaySkip
    \kappa_{t1} &:
    \all{x \, \Gamma  \, a} \varP \, x \land \findP \, \Gamma \, x \, a\
    \impl \typedP \, \Gamma \, x  \, a \\
    \kappa_{t2}&:
    \all{x \, \Gamma \, a \, m \, b} \typedP \, [x : a | \Gamma] \,  m \, b
    \impl \typedP \, \Gamma\, (\lambda \, x \,  m)\, (a \rightarrow b) \\
    \kappa_{t3}&:
    \all{\Gamma \, a \, m \, n \, b}  \typedP \, \Gamma \, m\, (a \rightarrow b) \land
    \typedP\, \Gamma\,  n\, a
    \impl \typedP \, \Gamma \, (\appP\, m\,  n)  \, b
  \end{align*}

It is well known that the $Y$-combinator is not typable in the simply-typed
$\lambda$-calculus and, in particular, self-application $\lam{x} x \, x$ is not
typable either.
However, by switching off the occurs-check in Prolog or by allowing circular
unifiers in CoLP~\cite{GuptaBMSM07,SimonBMG07}, we can
resolve the goal ``$\typedP \, [] \,  (\lambda \, x \, (\appP\, x\,  x)) \, a$''
and would compute the circular substitution:
$a = b \rightarrow c, b = b \rightarrow c$ suggesting that an infinite, or
circular, type may be able to type this $\lambda$-term.
A similar trick would provide a typing for the $Y$-combinator.
%This kind of coinductive computational behaviour has been observed before,
%but is not well understood.
Thus, a coinductive interpretation of the above Horn clauses yields a theory
of infinite types, while an inductive interpretation corresponds to the standard
type system of the simply typed $\lambda$-calculus.

\paragraph{Horn Clause Theories in Type Class Inference}
 Haskell type class inference does not require circular unifiers but
may require a cyclic
%an infinite\todo{Better: \emph{coinductive} res. inf.?}
resolution inference~\cite{Lammel:2005,FKS15}.
Consider, for example, the following mutually defined data structures in
Haskell.
\lstset{language=Haskell}
\begin{lstlisting}
data OddList  a  =   OCons a (EvenList  a)
data EvenList a  =   Nil | ECons a (OddList a)
\end{lstlisting}
This type declaration gives rise to the following equality class instance
declarations, where we leave the, here irrelevant, body out.
\begin{lstlisting}
instance(Eq a, Eq (EvenList a)) => Eq (OddList a) where
instance(Eq a, Eq (OddList  a)) => Eq (EvenList a) where
\end{lstlisting}
The above two type class instance declarations have the shape of Horn clauses.
Since the two declarations mutually refer to each other, an instance inference
for, e.g., \lstinline{Eq (OddList Int)} will give rise to
an infinite resolution that alternates between the subgoals
\lstinline{Eq (OddList Int)} and \lstinline{Eq (EvenList Int)}.
The solution is to terminate the computation as soon as the cycle is detected~\cite{Lammel:2005}, and this method has been
shown sound relative to the greatest Herbrand models in~\cite{FKH16}. %, i.e. the greatest models that comprise finite atomic formulae only.
%capture this infinite computation by a fixpoint term $t$ and
%prove that $t$ inhabits the type \lstinline{Eq (OddList Int)},
%see~\cite{FKH16,FKS15,Lammel:2005}.\todo{This is wrong! We don't need
%  fp terms here.}
%Here, we
%The cycle signifies that there is an infinite proof for
%a property of \emph{finite objects}, i.e. of
%a finite term \lstinline{OddList Int}.
%in this case finite lists of even or odd
%length.
We will demonstrate this later in the proof systems proposed in this paper.

The diversity of these coinductive examples in the existing literature shows
that there is a practical demand for coinductive methods in Horn clause logic, but it also shows that
no unifying proof-theoretic approach exists to allow for a generic use of these methods. This causes several problems.

\textbf{Problem 1. The existing proof-theoretic coinductive interpretations of
  cycle and loop detection are unclear, incomplete and not uniform.}
\begin{table}[bt]
  {\small{
      % \centering
      \begin{tabular}{|p{2.65cm} ||  c|c|c|}
        \hline
        Horn clauses
        & $\gamma_1: \all{x} p \, x \impl p \, x$
        & $\gamma_2: \all{x} p (f\, x) \impl p \, x$
        & $\gamma_3: \all{x} p \, x \impl p(f\, x)$
        \\ \hline
        Greatest  Herbrand model:
        & $ \{p \, a\}$
        & \parbox[t]{4.5cm}{
          $\{p(a), p(f \, a), p(f(f \, a)), \dotsc \}$}
        & $\emptyset$
        \\ \hline
        Greatest complete  Herbrand model:
        & $ \{ p \, a\}$
        & \parbox[t]{4cm}{
          $ \{p(a), p(f \, a), p(f(f \, a), \dotsc,$
          \hspace*{1ex}$p(f(f\ldots)\}$
        }
        & $\{p(f(f\ldots)\}$
        \\  \hline
        CoLP substitution for query $p \, a$
        & $\textit{id}$
        & fails
        & fails
        \\ \hline
        CoLP substitution for query $p \, x$
        & $\textit{id}$
        & $x = f \, x $
        & $x = f \, x $
        \\ \hline
      \end{tabular}
      \caption{\textbf{Examples of greatest (complete) Herbrand models for Horn
          clauses $\gamma_1$, $\gamma_2$, $\gamma_3$.}
        The signatures are $\set{a}$ for the clause $\gamma_1$ and $\set{a, f}$
        for the others.
        % The signature for $\gamma_1$ consists of an arbitrary constant symbol
        % $a$, while the other clauses are given for the signature containing
        % $a$ and $f$.
      }
        % to the signature, in order to
        % have ground instances of formulae in the models.}
      \label{tab:herbrand-overview}
    }
  }
%  \vspace*{-1cm}
\end{table}

To see this, consider \tabRef*{herbrand-overview}, which exemplifies three kinds of
circular phenomena in Horn clauses:
The clause $\gamma_1$ is the easiest case.
Its coinductive models are given by the finite set $\{ p \, a\}$.
On the other extreme is the clause $\gamma_3$ that, just like $\kappa_{\stream}$, admits only an infinite formula in its coinductive model.
%(however, the model contains just one such formula).
The intermediate case is $\gamma_2$, which could be interpreted by an infinite set of finite formulae in its greatest Herbrand model, or may admit an infinite formula in its greatest complete Herbrand model. Examples like $\gamma_1$ appear in Haskell type class resolution~\cite{Lammel:2005}, and examples like $\gamma_2$ in its experimental extensions~\cite{FKS15}.
%shows, admits, while examples like $\gamma_1$ were given in~\cite{Lammel:2005}. % ($\gamma_1$ being a special case).
Cycle detection would only cover computations for $\gamma_1$, whereas $\gamma_2$, $\gamma_3$ require some form of loop detection%
\footnote{We follow the standard terminology of~\cite{Terese}
  and say that two formulae $F$ and $G$ form a cycle if $F = G$, and
  a loop if $F[\theta] = G[\theta]$ for some (possibly circular)
  unifier $\theta$.}.
However, CoLP's loop detection gives confusing results here. It correctly fails to infer $p\, a$ from $\gamma_3$ (no unifier for subgoals $p \, a$ and $p \, (f \, a)$ exists), but incorrectly fails to infer $p \, a$ from $\gamma_2$ (also failing to unify $p \, a$ and $p \, (f \, a)$). The latter failure is misleading bearing in mind that  $p\,  a$ is in fact in the coinductive model of $\gamma_2$.
Vice versa, if we interpret the CoLP answer $x = f \, x$ as a declaration
of an infinite term $(f \, f \, \ldots)$ in the model, then CoLP's answer
for $\gamma_3$ and $p \, x$ is exactly correct, however the same answer is badly incomplete for
the query involving  $p \, x$ and $\gamma_2$, because $\gamma_2$  in fact admits
other, finite, formulae in its models.
And in some applications, e.g. in Haskell type class inference, a finite formula would be the only acceptable answer for any query to $\gamma_2$.

This set of examples shows that loop detection is too coarse a tool to give an operational semantics to a diversity of coinductive models.
%Is it possible to have a more uniform, complete and precise coinductive operational semantics for Horn clauses?

\textbf{Problem 2. Constructive interpretation of coinductive proofs in Horn clause logic is unclear.}
Horn clause logic is known to be a constructive fragment of FOL.
Some applications of Horn clauses rely on this property in a crucial way.
For example, inference in Haskell type class resolution is constructive: when a certain formula $F$ is inferred,
the Haskell compiler in fact constructs a proof term that inhabits $F$ seen as type.
In our earlier example \lstinline{Eq (OddList Int)} of the Haskell type classes,
Haskell in fact  captures the cycle by a fixpoint term $t$ and
proves that $t$ inhabits the type  \lstinline{Eq (OddList Int)}.
Although we know from~\cite{FKH16} that these computations
are sound relative to greatest Herbrand models of Horn clauses, the results of~\cite{FKH16} do not extend to Horn clauses like $\gamma_3$ or  $\kappa_{\stream}$, or generally to Horn clauses modelled by the greatest \emph{complete} Herbrand models.
% Does this draw a boundary between constructive and non-constructive coinductive proofs in Horn clause logic?
This shows that there is not just a need for coinductive proofs in Horn clause
logic, but \emph{constructive} coinductive proofs.

\textbf{Problem 3. Incompleteness of circular unification for irregular coinductive data structures.}
\tabRef{herbrand-overview} already showed some issues with incompleteness of circular unification.
A more famous consequence of it is the failure of circular unification to capture irregular terms.
This is illustrated by the following Horn clause, which defines the
infinite stream of successive natural numbers.
\begin{equation*}
  \kappa_{\fromP}: \all{x \, y}  \fromP \, (s\, x) \, y
  \impl \fromP \, x \, (\scons \, x \, y)
  % \label{eq:horn clause: from}
\end{equation*}
The reductions for $\fromP \, 0 \, y$ consist only of irregular
(non-unifiable) formulae:
\begin{equation*}
%\label{eq:from-trace}
\fromPred{0}{y}
\stackrel{ \kappa_{\fromP}: [\sconsFunc{0}{y'}/y]}{\sldRed} \fromPred{(s\, 0)}{y'}
\stackrel{ \kappa_{\fromP}:[\sconsFunc{(s\, 0)}{y''}/y']}{\sldRed} \cdots
\end{equation*}
%Just as in the previous example,
The composition of the computed substitutions would suggest as answer an
infinite term that is given by
$\fromP \, 0 \, (\sconsFunc{0}{(\sconsFunc{(s \, 0)}{\ldots})})$.
However, circular unification no longer helps to compute this answer, and CoLP fails.
Thus, there is a need for
 more general operational semantics that allows irregular coinductive structures.

\subsection*{A New Theory of Coinductive Proof Search in Horn Clause Logic}

In this paper, we aim to give a principled and \emph{general} theory that resolves the three problems above.
%underlies and explains both
%coinductive resolution and circular unification for Horn clauses.
This theory establishes a \emph{constructive} foundation for coinductive resolution and allows us to
give proof-theoretic characterisations of the approaches that have been proposed
throughout the literature.

%\paragraph{Technical Contributions and Outline}

To solve Problem 1, we follow the footsteps of the \emph{uniform proofs} by
Miller et al.~\cite{MN12,Miller91:UniformProofs},
who gave a general proof-theoretic account of resolution in first-order
Horn clause logic  (\emph{fohc}) and three extensions:
first-order hereditary Harrop clauses (\emph{fohh}),
higher-order Horn clauses (\emph{hohc}),
and higher-order hereditary Harrop clauses (\emph{hohh}).
%showed its connections to intuitionistic logic.
%After introducing in \secRef*{terms-formulae} the basic terms and formulae,
%we
  In \secRef*{uniform-proofs}, we extend uniform proofs with a general coinduction
proof principle. % and fixed point terms at the formula level.
The resulting framework is called \emph{coinductive uniform proofs (CUP)}.
%and it gives a
%coherent and general framework for constructing proofs for all cases of coinduction described above.
We show how the coinductive extensions of the four logics of
Miller et al., which we name  \cofohc{}, \cofohh{}, \cohohc{}
and \cohohh{}, give a precise proof-theoretic characterisation to the different kinds of coinduction described in the literature.
For example, coinductive proofs involving the clauses $\gamma_1$ and $\gamma_2$ belong to  \cofohc{} and \cofohh{}, respectively.
However, proofs involving clauses like $\gamma_3$ or $\kappa_{\stream}$ require
in addition fixed point terms to express infinite data.
These extentions are denoted by
$\cofohc_{\fix}$, $\cofohh_{\fix}$, $\cohohc_{\fix}$ and $\cohohh_{\fix}$.

%\knote{should the below snippet got into CUP section?}
%An important novel contribution that we make in \secRef*{uniform-proofs} is the
%classification of the logical strength of coinductive proofs in Horn clause logic along the lines
%of uniform proofs.
%This classification uses four well-known and studied logics:
%first-order Horn clauses (\emph{fohc}),
%first-order hereditary Harrop clauses (\emph{fohh}),
%higher-order Horn clauses (\emph{hohc}),
%and higher-order hereditary Harrop clauses (\emph{hohh}).
%In coinductive uniform proofs, each of these logics comes with a coinduction
%principle that admits formulae of varying strength as coinduction hypothesis,
%depending on the logic.
%The resulting logics are denoted by \cofohc{}, \cofohh{}, \cohohc{}
%and \cohohh{}.
%Furthermore, the logical strength can be further varied by allowing
%fixed points terms.

\secRef*{uniform-proofs} shows that this yields the cube in \figRef*{cup-cube},
where the arrows show the increase in logical strength.
\begin{figure}[bt]
\begin{equation*}
  \begin{tikzcd}[row sep=small, column sep=small]
    & \cohohc_{\fix}   \arrow{rr}{}
    & & \cohohh_{\fix}
    \\
    \cohohc           \arrow{ur}{}
    & & \cohohh       \arrow{ur}{}
    \\
    & \cofohc_{\fix}   \arrow{rr}{} \arrow{uu}{}
    & & \cofohh_{\fix}              \arrow{uu}{}
                                   \arrow[from=lllu, to=lu, crossing over]
    \\
    \cofohc           \arrow{rr}{} \arrow{uu}{}   \arrow{ur}{}
    & & \cofohh       \arrow[crossing over]{uu}{} \arrow{ur}{}
  \end{tikzcd}
\end{equation*}
\vspace*{-1em}
\caption{Cube of logics covered by CUP}
\label{fig:cup-cube}
% \vspace*{-2em}
\end{figure}
%All currently known methods for corecursive proof-search in Horn clause logic
%are described by this cube.
%In particular,
The invariant search for regular infinite objects done in CoLP is
fully described by the logic $\cofohc_{\fix}$, including proofs for clauses
like $\gamma_3$ and $\kappa_{\stream}$.
An important consequence is that CUP is complete for
$\gamma_1$, $\gamma_2$, and $\gamma_3$,
e.g.  $p \, a$ is provable from $\gamma_2$ in CUP, but not in CoLP.
%Haskell's coinductive type class inference given in~\cite{Lammel:2005} falls
%under the remit of \cofohc, while the extension of~\cite{Lammel:2005}
%in~\cite{FKS15,FKH16} requires \cofohh.

In tackling Problem~3, we will find that the irregular proofs, such as those for
$\kappa_{\fromP}$,  %``$\fromP \, 0 \, (\fromFun \, 0)$''
can be given in
$\cohohh_{\fix}$.
The stream of successive numbers can be defined as a higher-order fixed point term
$\fromFun = \fix[f] \lam{x} \scons \, x \, (f \, (s \, x))$, and the proposition
$\all{x} \fromP \, x \, (\fromFun \, x)$ is provable in $\cohohh_{\fix}$.
This requires the use of higher-order syntax,
fixed point terms and the goals of universal shape, which become available in the syntax of Hereditary Harrop logic.
%with our proposed coinductive proof
%principle.

% \knote{end of snippet}

In order to solve Problem 2 and to expose the constructive nature of the
resulting proof systems, we present
in \secRef*{loeb-translation} a coinductive extension of first-order
intuitionistic logic and its sequent calculus.
This extension  ($\iFOLm$) is based on the so-called later modality (or Löb modality)
known from
provability logic~\cite{Beklemishev1999:ParameterFreeInduction,%
  Smorynski1985:SelfReference},
type theory~\cite{Nakano00:ModalityRec,Appel07:ModalTypeSystem}
and domain theory~\cite{Birkedal11:GuardedDomainTheory-LICS}.
However, our way of using the later modality to control recursion in
first-order proofs is new and builds on~\cite{B18,Basold18:BreakingTheLoop}.
In the same section we also show that
CUP is
% coinductive uniform proofs are
sound
relative to  $\iFOLm$, which gives us a handle on the constructive
content of CUP.
This yields, among other consequences, a constructive interpretation of CoLP
proofs.

\secRef{soundness} is dedicated to showing soundness of both coinductive proof
systems relative to %the well known coinductive models of Horn clause theories,
%also known as
\emph{complete Herbrand models}~\cite{Llo87}.
The construction of these models is carried out by using coalgebras and
category theory.
This frees us from having to use topological methods and will simplify future
extensions of the theory to, e.g., encompass typed logic programming.
It also makes it possible to give original and constructive proofs of soundness
for both CUP and $\iFOLm$ in \secRef{soundness}.
% Finally, in \secRef*{heuristics} we give a generalisation of currently known
% heuristics for coinductive proof search and theory exploration.
% %We expose the fact that the search for
% %suitable coinductive hypotheses in a proof
% %corresponds to structural properties of the complete Herbrand models of the Horn clause theory in question.
% %This allows us to propose a new and more general method for discovering
% %coinduction hypotheses for irregular infinite formulae.
% For example,  a proof search for  $\exist{x} \fromP \, 0 \, x$ requires
% the discovery of and a coinductive proof for a more general lemma:
% $\all{x} \fromP \, x\, (s_{\mathrm{fr}} \, x)$, where $\fromFun$ is the
% fixed point term that we defined above.
% Our new heuristic is able discover such lemmas automatically.
We finish the paper with discussion of related and future work.

\subsection*{Originality of the contribution}
The results of this paper give a comprehensive characterisation of
coinductive Horn clause theories from the point of view of proof search
(by expressing coinductive proof search and resolution as coinductive
uniform proofs),
constructive proof theory (via a translation into an intuitionistic sequent
calculus), and
coalgebraic semantics (via coinductive Herbrand models
and constructive soundness results).
Several of the presented results have never appeared before:
the coinductive extension of uniform proofs;
characterisation of coinductive properties of Horn clause theories in
higher-order logic with and without fixed point operators;
coalgebraic and fibrational view on complete Herbrand models; and
soundness of an intuitionistic logic with later modality relative to complete
Herbrand models.
%; and
%novel heuristic for discovering coinductive hypotheses describing infinite
%irregular terms.

%More precisely,

%\begin{itemize}
%\item we propose a new proof-theoretic framework for proof search and proof
%  construction for coinductive theories expressed in Horn clause logic.
%  This framework unifies and generalises all the previous coinductive methods
%  known for Horn clauses.
%\item we use this novel framework to extend the power of coinductive proof
%  search to deal with irregular corecursive inference.
%\item we prove soundness of this framework relative to a first-order
%  intuitionistic logic, thereby establishing the constructive nature of the
%  framework.
%\item we establish soundness of this logic and the framework with respect to
%  term models à la Herbrand by using category theoretical methods (coalgebras
%  and fibrations).
%  This enables the introduction of typing into resolution and logic programming in the long run.
%\end{itemize}

% \begin{toappendix}
%   \input{content/Appendices/notation}
% \end{toappendix}
\section{Preliminaries: Terms and Formulae}
\label{sec:terms-formulae}

In this section, we set up notation and terminology for the rest of the paper.
Most of it is standard, and blends together the notation used
in~\cite{MN12} and~\cite{Barendregt:LambdaCalcTypes}.

% \begin{itemize}
%\item We define (higher-order and fixed point) $\lambda$-terms and formulae
%  required for fohc, fohh, hohc and hohh that also involve disjunction etc.
%  to be able to refer to Miller and Nadathur
%\item then we say that our goal formulae and the uniform proofs only involve
%  conjunction etc. (``coinductive logic connectives'')
%\item properly split between terms and formulae, as we will not use any mix
%\end{itemize}

\begin{definition}
  \label{def:types}
  % Let $\BaseT$ be a non-empty set of \emph{base types}.
  We define the sets $\Types$ of \emph{types} and
  $\PropT$ of \emph{proposition types}
  by the following grammars, where $\baseT$ and $\propT$
  are the \emph{base type} and \emph{base proposition type}.
  \begin{equation*}
    \Types \ni \sigma, \tau \coloncolonequals
    \baseT
    \mid \sigma \to \tau
    \qquad \qquad
    \PropT \ni \rho \coloncolonequals
    \propT
    \mid \sigma \to \rho, \quad \sigma \in \Types
  \end{equation*}
  We adapt the usual convention that $\to$ binds to the right.
\end{definition}

\begin{definition}
  \label{def:signature-terms}
  A \emph{term signature} $\TSig$ is a set of pairs $c : \tau$, where
  $\tau \in \Types$, and a \emph{predicate signature} is a set
  $\PSig$ of pairs $p : \rho$ with $\rho \in \PropT$.
  The elements in $\TSig$ and $\PSig$ are called \emph{term symbols}
  and \emph{predicate symbols}, respectively.
  Given term and predicate signatures $\TSig$ and $\PSig$, we refer to
  the pair $(\TSig, \PSig)$ as \emph{signature}.
  Let $\Vars$ be a countable set of variables, the elements of which we denote
  by $x, y, \dotsc$
  We call a finite list $\Gamma$ of pairs $x : \tau$ of variables and types a
  \emph{context}.
  The set  $\FixTerms{\TSig}$ of \emph{(well-typed) terms} over $\TSig$
  is the collection of all $M$ with $\typed{M}{\tau}$ for some context
  $\Gamma$ and type $\tau \in \Types$, where $\typed{M}{\tau}$ is defined
  inductively in \figRef*{terms}.
  A term is called \emph{closed} if $\typed[]{M}{\tau}$, otherwise it is called
  \emph{open}.
  Finally, we let $\Terms{\TSig}$ denote the set of all terms $M$ that
  do not involve $\fix$.
  \begin{figure}[bt]
    \begin{spreadlines}{7pt}
      \begin{empheq}[box=\fbox]{gather*}
        \AxiomC{$c : \tau \in \TSig$}
        \UnaryInfC{$\typed{c}{\tau}$}
        \DisplayProof
        \quad
        \AxiomC{$x : \tau \in \Gamma$}
        \UnaryInfC{$\typed{x}{\tau}$}
        \DisplayProof
        \quad
        \AxiomC{$\typed{M}{\sigma \to \tau}$}
        \AxiomC{$\typed{N}{\sigma}$}
        \BinaryInfC{$\typed{M \> N}{\tau}$}
        \DisplayProof
        \\
        \AxiomC{$\typed[\Gamma, x : \sigma]{M}{\tau}$}
        \UnaryInfC{$\typed{\lam{x} M}{\sigma \to \tau}$}
        \DisplayProof
        \quad
        \AxiomC{$\typed[\Gamma, x : \tau]{M}{\tau}$}
        \UnaryInfC{$\typed{\fix[x] M}{\tau}$}
        \DisplayProof
      \end{empheq}
    \end{spreadlines}
    \vspace*{-1em}
    \caption{Well-Formed Terms}
    \label{fig:terms}
    \vspace*{-1em}
  \end{figure}
\end{definition}

\begin{definition}
  \label{def:formulae}
  Let $(\TSig, \PSig)$ be a signature.
  We say that $\varphi$ is a \emph{(first-order) formula} in context $\Gamma$,
  if $\validForm{\varphi}$ is inductively derivable from the rules
  in \figRef*{formulae}.
  \begin{figure}[bt]
    \begin{spreadlines}{7pt}
      \begin{empheq}[box=\fbox]{gather*}
        \def\defaultHypSeparation{\hskip .05in}
        % \AxiomC{}
        % \UnaryInfC{$\validForm{\bot}$}
        % \bottomAlignProof
        % \DisplayProof
        % \quad
        \AxiomC{$(p : \tau_1 \to \dotsm \to \tau_n \to \propT) \in \PSig$}
        \AxiomC{$\typed{M_1}{\tau_1}$}
        \AxiomC{$\dotsm$}
        \AxiomC{$\typed{M_n}{\tau_n}$}
        \QuaternaryInfC{$\validForm{p \> M_1 \dotsm \> M_n}$}
        \bottomAlignProof
        \DisplayProof
        \\[7pt]
        \AxiomC{}
        \UnaryInfC{$\validForm{\top}$}
        \bottomAlignProof
        \DisplayProof
        \quad
        \AxiomC{$\validForm{\varphi}$}
        \AxiomC{$\validForm{\psi}$}
        \AxiomC{$\Box \in \set{\conj, \disj, \impl}$}
        \TrinaryInfC{$\validForm{\varphi \mathbin{\Box} \psi}$}
        \bottomAlignProof
        \DisplayProof
        \quad
        \AxiomC{$\validForm[\Gamma, x:\tau]{\varphi}$}
        \UnaryInfC{$\validForm{\all{x:\tau} \varphi}$}
        \bottomAlignProof
        \DisplayProof
        \quad
        \AxiomC{$\validForm[\Gamma, x : \tau]{\varphi}$}
        \UnaryInfC{$\validForm{\exist{x : \tau} \varphi}$}
        \bottomAlignProof
        \DisplayProof
      \end{empheq}
    \end{spreadlines}
    \vspace*{-1em}
    \caption{Well-formed Formulae}
    \label{fig:formulae}
  \end{figure}
\end{definition}

% \begin{remark}
%   Miller and Nadathur~\cite{MN12} define formulae as terms of propositional
%   type in the style of Church's simple theory of types.
%   It is easy to show that formulae as we defined them correspond
%   to $\beta$-normal forms of terms of propositional type, if there
%   are no higher-order quantifiers and higher-order predicate symbols.
%   % Since we will restrict to Horn clauses as goals in the course of this paper,
%   % this separation of terms and formulae will pose no restriction.
% \end{remark}

\begin{definition}
  \label{def:reductions}
  The \emph{reduction relation} $\reduce$ on terms in $\FixTerms{\TSig}$ is
  given as the compatible closure (reduction under applications and binders)
  of $\beta$- and $\fix$-reduction:
  \begin{equation*}
    (\lam{x} M) N \reduce M \subst{N/x}
    \qquad
    \fix[x] M \reduce M \subst{\fix[x] M/x}
  \end{equation*}
  We denote the reflexive, transitive closure of $\reduce$ by $\reduceIter$.
  Two terms $M$ and $N$ are called \emph{convertible}, if $M \conv N$, where
  $\conv$ is the equivalence closure of $\reduce$.
  % Conversion of terms extends from terms to formulae by taking the compatible
  % closure of conversion under predicate symbols:
  Conversion of terms extends to formulae in the obvious way:
  if $M_k \conv M'_k$ for $k = 1, \dotsc, n$, then
  $p \> M_1 \dotsm M_n \conv p \> M'_1 \dotsm \> M'_n$.
\end{definition}

We will use in the following that the above calculus features
subject reduction and confluence, cf.~\cite{Plotkin77:PCF}:
if $\typed{M}{\tau}$ and $M \conv N$, then $\typed{N}{\tau}$;
and $M \conv N$ iff there is a term $P$, such that $M \reduceIter P$ and
$N \reduceIter P$.

%\subsection{First-Order Signatures, Guarded Terms and Atoms}
%\label{sec:fo-sig-guarded-terms}

The \emph{order} of a type $\tau \in \Types$ is given as usual
by $\ord(\iota) = 0$ and
$\ord(\sigma \to \tau) = \max\set{\ord(\sigma) + 1, \ord(\tau)}$.
If $\ord(\tau) \leq 1$, then the arity of $\tau$ is given by
$\ar(\iota) = 0$ and $\ar(\iota \to \tau) = \ar(\tau) + 1$.
A signature $\TSig$ is called \emph{first-order}, if for all $f : \tau \in \TSig$
we have $\ord(\tau) \leq 1$.
We let the arity of $f$ then be $\ar(\tau)$ and denote it by $\ar(f)$.

\begin{definition}
  \label{def:guarded-terms}
  The set of \emph{guarded base terms} over a first-order signature
  $\TSig$ is given by the following type-driven rules.
  \begin{gather*}
    \AxiomC{$x : \tau \in \Gamma$}
    \AxiomC{$\ord(\tau) \leq 1$}
    \BinaryInfC{$\guarded{x}{\tau}$}
    \DisplayProof
    \quad
    \AxiomC{$f : \tau \in \TSig$}
    \UnaryInfC{$\guarded{f}{\tau}$}
    \DisplayProof
    \quad
    \AxiomC{$\guarded{M}{\sigma \to \tau}$}
    \AxiomC{$\guarded{N}{\sigma}$}
    \BinaryInfC{$\guarded{M \> N}{\tau}$}
    \DisplayProof
    \\
    \def\defaultHypSeparation{\hskip .1in}
    \AxiomC{$f : \sigma \in \TSig$}
    \AxiomC{$\ord(\tau) \leq 1$}
    \AxiomC{$
      \guarded[\Gamma, x : \tau, y_1 : \iota, \dotsc, y_{\ar(\tau)} : \iota]{
        M_i}{\iota}$}
    \AxiomC{$1 \leq i \leq \ar(f)$}
    \QuaternaryInfC{$\guarded{\fix[x] \lam{\vv{y}} f \> \vv{M}}{\tau}$}
    \DisplayProof
  \end{gather*}
  General \emph{guarded terms} are terms $M$, such that all
  $\fix$-subterms are guarded base terms, which means that they are generated
  by the following grammar.%
  \begin{equation*}
    G \coloncolonequals
    M \; (\text{with } \guarded[]{M}{\tau} \text{ for some type } \tau)
    \mid c \in \TSig \mid x \in \Vars \mid G \; G \mid \lam{x} G
  \end{equation*}
  Finally, $M$ is a \emph{first-order} term over
  $\TSig$ with $\typed{M}{\tau}$ if $\ord(\tau) \leq 1$ and the types of all
  variables occurring in $\Gamma$ are of order $0$.
  We denote the set of guarded first-order terms $M$ with $\typed{M}{\iota}$ by
  $\GuardedFOTerms{\TSig}(\Gamma)$ and the set of guarded terms in $\Gamma$ by
  $\GuardedTerms{\TSig}(\Gamma)$.
  If $\Gamma$ is empty, we just write $\GuardedFOTerms{\TSig}$ and
  $\GuardedTerms{\TSig}$, respectively.
\end{definition}

Note that an important aspect of guarded terms is that no free variable occurs
under a $\fix$-operator.
\emph{Guarded base terms} should be seen as specific fixed point terms
that we will be able to unfold into potentially infinite trees.
\emph{Guarded terms} close guarded base terms under operations of the
simply typed $\lambda$-calculus.

\hnote[disable]{
  Would we need to allow open terms for guarded terms, that is, generalise
  $\guarded{M}{\tau}$ instead of $\guarded[]{M}{\tau}$ in the grammar for $G$
  above?
  Note that the problem this generality would be that guarded terms are not
  stable under substitution: if $\guarded[x : \baseT]{M}{\tau}$
  and $N$ is guarded with $\typed[]{N}{\baseT}$, then
  $M \ssubst{N/x}$ is \emph{not} a guarded base term!
  Hence, we would not be able to show that if $M$ and $N$ are guarded that
  $M \ssubst{N/x}$ is guarded.
}

\begin{example}
  \label{ex:guarded-terms}
  Let us provide a few examples that illustrate (first-order) guarded terms.
  We use the first-order signature
  $\TSig = \set{\scons \from \baseT \to \baseT \to \baseT,
    s \from \baseT \to \baseT, 0 : \baseT}$.
  \begin{enumerate}
  \item Let
    $\fromFun = \fix[f] \lam{x} \scons \; x \; (f \; (s \; x))$
    be the function that computes the streams of numerals starting
    at the given argument.
    It is easy to show that
    $\guarded[]{\fromFun}{\baseT \to \baseT}$ and so
    $\fromFun \; 0 \in \GuardedFOTerms{\TSig}$.
  \item For the same signature $\TSig$ we also have
    $\guarded[x : \baseT]{x}{\baseT}$.
    Thus $x \in \GuardedFOTerms{\TSig}(x : \baseT)$ and
    $s \; x \in \GuardedFOTerms{\TSig}(x : \baseT)$.
  \item We have $\guarded[x : \baseT \to \baseT]{x \; 0}{\baseT}$, but
    $(x \; 0) \not\in \GuardedFOTerms{\TSig}(x : \baseT \to \baseT)$.
  % \item Note, however, that even though
  %   $\guarded[x : \baseT \to \baseT]{x \; 0}{\baseT}$, the term
  %   $x \; 0$ is not in $\GuardedFOTerms{\TSig}(x : \baseT \to \baseT)$
  %   because variables of type $\baseT \to \baseT$ are not allowed to
  %   be free in the first-order guarded terms of
  %   $\GuardedFOTerms{\TSig}(x : \baseT \to \baseT)$.
  %   Instead, both $x$ and  $x \; 0$ are guarded terms, as are
  %   $x \; (\fromFun \; 0)$ and $y \; \fromFun$ for
  %   $y : (\baseT \to \baseT) \to \baseT$.
  \end{enumerate}
\end{example}

The purpose of guarded terms is that these are productive, that is, we can
reduce them to a term that either has a function symbol at the root or is
just a variable.
In other words, guarded terms have head normal forms:
We say that a term $M$ is in \emph{head normal form}, if $M = f \; \vv{N}$ for
some $f \in \TSig$ or if $M = x$ for some variable $x$.
The following lemma is a technical result that is needed to show in
\lemRef*{guarded-computation} that all guarded terms have a head normal form.
\begin{lemmarep}
  \label{lem:guarded-subst}
  Let $M$ and $N$ be guarded base terms with
  $\guarded[\Gamma, x : \sigma]{M}{\tau}$
  and $\guarded[\Gamma]{N}{\sigma}$.
  Then $M \subst{N/x}$ is a guarded base term with
  $\guarded{M \subst{N/x}}{\tau}$.
\end{lemmarep}
\begin{appendixproof}
  Let $M$ and $N$ be as above.
  We proceed by induction on the derivation that $M$ is guarded to show that
  $M \subst{N/x}$ is guarded as well.
  \begin{itemize}
  \item Suppose $M = y$ for some variable.
    If $y = x$, then $M \subst{N/x} = N$ and $\tau = \sigma$.
    Thus, we have $\guarded{M \subst{N/x}}{\tau}$.
  \item The case for signature symbols is immediate, as for $f \in \TSig$ we
    have $f \subst{N/x} = f$.
  \item Suppose $\guarded[\Gamma, x : \sigma]{M \> P}{\tau}$.
    By the IH, we have
    $\guarded[\Gamma]{M \subst{N/x}}{\gamma \to \tau}$ and
    $\guarded[\Gamma]{P \subst{N/x}}{\gamma}$.
    Thus, we obain
    $\guarded[\Gamma]{(M \> P) \subst{N/x}}{\tau}$.
  \item Finally, assume that
    $\guarded[\Gamma, x : \sigma]{\fix[z] \lam{\vv{y}} f \> \vv{M}}{\tau}$.
    Then by IH, we have
    \begin{equation*}
      \guarded[\Gamma, z : \tau, y_1 : \iota, \dotsc, y_{\ar(\tau)} : \iota]{
        M_i \subst{N/x}}{\iota}
    \end{equation*}
    and so
    $\guarded[\Gamma]{
      \parens[\big]{\fix[z] \lam{\vv{y}} f \> \vv{M}} \subst{N/x}}{\tau}$.
    \qedhere
  \end{itemize}
\end{appendixproof}

\begin{lemmarep}
  \label{lem:guarded-computation}
  If $M$ is a first-order guarded term with $M \in \GuardedFOTerms{\TSig}(\Gamma)$,
  then $M$ reduces to a unique head normal form.
  This means that either
  \begin{enumerate*}[label=(\roman*)]
  \item there is a unique $f \in \TSig$ and terms
    $N_1, \dotsc, N_{\ar(f)}$ with $\guarded{N_k}{\baseT}$ and
    $M \reduceIter f \> \vv{N}$,
    and for all $L$ if $M \reduceIter f \> \vv{L}$,
    then $\vv{N} \conv \vv{L}$; or
  \item $M \reduceIter x$ for some $x : \baseT \in \Gamma$.
  \end{enumerate*}
  % This means that there is either a unique $f \in \TSig$ and terms
  % $N_1, \dotsc, N_{\ar(f)}$ with $\guarded{N_k}{\baseT}$ and moreover,
  % if
  % $M \reduceIter f \> \vv{N}$,
  % or $M \reduceIter x$ for some $x : \baseT \in \Gamma$.
  % Moreover, if $M \reduceIter f \> \vv{L}$, then $\vv{N} \conv \vv{L}$.
\end{lemmarep}
% \begin{proofsketch}
%   This follows by distinguishing the cases that can occur for guarded
%   terms of base type and then applying \lemRef*{guarded-subst} after
%   a reduction in the case of fixed point terms.
%   The uniqueness follows from confluence of reduction relation.
% \end{proofsketch}
\begin{appendixproof}
  The term $M$ with $\guarded{M}{\baseT}$ can have either of the following
  three shapes:
  \begin{enumerate}
  \item $x$, where $x : \baseT \in \Gamma$
  \item $f \> \vv{N}$ with $\guarded{N_k}{\baseT}$, or
  \item $(\fix[x] \lam{\vv{y}} f \> \vv{M}) \> \vv{N}$ with
    $\guarded[\Gamma,x : \tau, y_1 : \iota, \dotsc, y_{\ar(\tau)} : \iota]{
      M_k}{\iota}$ for $k = 1, \dotsc, \ar(f)$ and
    $\guarded{N_i}{\iota}$ for $i = 1, \dotsc, \ar(\tau)$,
  \end{enumerate}
  because variables can only occur in argument position due to the order
  restriction of the types in $\Gamma$.
  In the first two cases we are done immediately.
  For the third case, we let $P = \fix[x] \lam{\vv{y}} f \> \vv{M}$
  and then find that
  \begin{equation*}
    P \> \vv{N}
    \reduceIter f \> \parens*{\vv{M} \ssubst*{P/x, \vv{N}/\vv{y}}}.
  \end{equation*}
  \lemRef{guarded-subst} gives us now that each
  $M_i \ssubst*{P/x, \vv{N}/\vv{y}}$ is guarded.
  Finally, if $M \reduceIter f \> \vv{L}$, then $\vv{N} \conv \vv{L}$ by
  confluence of the reduction relation.
  \qedhere
\end{appendixproof}

\begin{toappendix}
\knote[disable]{please check Lemma and Lem. look odd next to each other }
  In \lemRef*{guarded-subst} we have shown that guarded base terms are stable
under substitution, that is, substituting a guarded base term into another
results into a guarded base term.
The following lemma shows that the same is true for guarded terms.
This result is necessary to define substitution for formulae over guarded terms,
see \defRef*{atoms}.
\begin{lemmarep}
  \label{lem:full-guarded-subst}
  Let $M \in \GuardedTerms{\TSig}(\Gamma, x)$
  and $N \in \GuardedTerms{\TSig}(\Gamma)$
  be guarded terms
  with $\typed[\Gamma, x : \tau]{M}{\sigma}$ and $\typed{N}{\tau}$.
  Then $M \subst{N/x} \in \GuardedTerms{\TSig}(\Gamma)$
  and $\typed{M \subst{N/x}}{\sigma}$.
\end{lemmarep}
\begin{proof}
  By an easy induction on $M$.
\end{proof}
\end{toappendix}

We end this section by introducing the notion of an atom and refinements thereof.
This will enable us to define the different logics and thereby to analyse
the strength of coinduction hypotheses, which we promised in the introduction.
\begin{definition}
  \label{def:atoms}
  A formula $\varphi$ of the shape
  $\top$ or $p \; M_1 \dotsm \; M_n$ is an \emph{atom} and a
  \begin{itemize}
  \item \emph{first-order atom}, if $p$ and all the terms $M_i$ are first-order;
  \item \emph{guarded atom}, if all terms $M_i$ are guarded; and
  \item \emph{simple atom}, if all terms $M_i$ are non-recursive, that is,
    are in $\Terms{\TSig}$.
  \end{itemize}
  First-order, guarded and simple atoms are denoted by
  $\foAt$, $\guardedAt$ and $\simpleAt$.
  We denote conjunctions of these predicates by
  $\foGuardedAt = \foAt \cap \guardedAt$ and
  $\foSimpleAt = \foAt \cap \simpleAt$.
\end{definition}

Note that the restriction for $\guardedAt$ only applies to fixed point terms.
Hence, any formula that contains terms without $\fix$ is already in
$\guardedAt$ and $\guardedAt \cap \simpleAt = \simpleAt$.
Since these notions are rather subtle, we give a few examples
\begin{example}
  We list three examples of first-order atoms.
  % and then give examples for the two higher-order variants.
  \begin{enumerate}
  \item For $x : \baseT$ we have $\stream \; x \in \foAt$,
    but there are also ``garbage'' formulae like
    ``$\stream \; (\fix[x] x)$'' in $\foAt$.
    Examples of atoms that are not first-order are $p \; M$, where
    $p : (\baseT \to \baseT) \to \propT$ or
    $\typed[x : \baseT \to \baseT]{M}{\tau}$.
  \item Our running example ``$\fromP \; 0 \; (\fromFun \; 0)$'' is a
    first-order guarded atom in $\foGuardedAt$.
  \item The formulae in $\foSimpleAt$ may not contain recursion and higher-order
    features.
    However, the atoms of Horn clauses in a logic program fit in here.
  % \item In the higher-order case we can quantify over the function that defines
  %   the $\fromP$-stream.
  %   That is, for $x : \baseT \to \baseT$, we have
  %   $\fromP \; 0 \; (x \; 0) \in \simpleAt$.
  % \item Finally, we can also apply functions to guarded terms
  %   in the higher-order case with recursion.
  %   For example, if $x : \baseT \to \baseT$, then
  %   $\fromP \; 0 \; (x \; (\fromFun \; 0)) \in \guardedAt$.
  %   But note that higher-order features may not appear under fixed points.
  %   For instance, if we define the term
  %   $M = \fix[F] \lam{f}
  %   s \; (f \; (F \; (\lam{x} f \; (s \; x))))$,
  %   where $F : (\baseT \to \baseT) \to \baseT$,
  %   $f : \baseT \to \baseT$ and $x : \baseT$,
  %   then $\nat \; (M \; (\lam{x} x))$ is not
  %   in $\guardedAt$.
  \end{enumerate}
\end{example}

\section{Coinductive Uniform Proofs}
\label{sec:uniform-proofs}

%\knote{the notion of an H-formula is used in a few places, but it is undefined. It must be defined here. I will do it now, but it needs a quick check later}
%\begin{itemize}
%\item The proof system simplified: no disjunction etc. (inductive connectives)
%  and simplified rules for guarded formulae
%\end{itemize}

This section introduces the eight logics of the coinductive uniform
proof framework announced and motivated in the introduction.
The major difference of uniform proofs with, say, a sequent calculus is
the ``uniformity'' property, which means that the choice of the application of
each proof rule is deterministic and all proofs are in normal form (cut free).
This subsumes the operational semantics of resolution, in which the proof search
is always goal directed.
Hence, the main challenge, that we set out to solve in this section, is to
extend the uniform proof framework with coinduction, while preserving this
valuable operational property.

We begin by introducing the different goal formulae and definite clauses
that determine the logics that were presented in the cube for coinductive
uniform proofs in the introduction.
These clauses and formulae correspond directly to those of the original work on
uniform proofs~\cite{MN12} with the only difference being that we need to
distinguish atoms with and without fixed point terms.
The general idea is that goal formulae ($G$-formulae) occur on the right of
a sequent, thus are the \emph{goal} to be proved.
Definite clauses ($D$-formulae), on the other hand, are selected from the
context as assumptions.
This will become clear once we introduce the proof system for coinductive
uniform proofs.
\begin{definition}
  \label{def:formula-classes}
  %\todo[inline]{Define $D$-, $G$- and $H$-formulae for the different calculi
  %  as in \tabRef*{up}.
  %  This needs a definition of atoms}
  %\ynote{Please note that given the new formula definition, there is no longer
  %  ``flexible atoms'', since all atoms start with $p:\rho\in\Sigma$.
  %  So the notation $A_r$ is no longer needed. }
  %\knote{Ok, I am doing this now, as i need correct notation for Heuristic
  % section. Please check there is no clash later. I am also removing $A_r$}
  Let $D_i$ be generated by the following grammar with $i \in \set{1,\omega}$.
  \begin{equation*}
    D_i \coloncolonequals
    \simpleAt[i] \mid G \impl D \mid  D \conj D \mid \all{x : \tau} D
  \end{equation*}
  The sets of definite clauses ($D$-formulae) and goals ($G$-formulae) of the
  four logics $\cofohc$, $\cofohh$, $\cohohc$, $\cohohh$ are the well-formed
  formulae of the corresponding shapes defined in \tabRef*{up}.
  For the variations $\cofohh_{\fix}$ etc. of these logics with
  fixed point terms, we replace upper index ``$s$'' with ``$g$'' everywhere
  in  \tabRef*{up}.
  %\footnote{
    % \begin{remark}
  %  Note that one could relax the definition of $D$-formulae to
  %  $D \coloncolonequals \dotsm \mid G \impl D \mid \dotsm$.
  %  However, this only complicates later proofs and the $D$-formulae
  %  arising from such a definition are logically equivalent to ours.
  %  Moreover, we will later restrict to Horn-clauses, which do not
  %  make use of this generality.
    % \end{remark}
  %}
  \begin{table}[bt]
    \footnotesize{
      \renewcommand{\arraystretch}{1.5}
      \begin{tabular}{l|l|l}
        & \text{Definite Clauses} & \text{Goals}
        \\ \hline
        $\begin{alignedat}[t]{1}
          & \cofohc \\
          & \cohohc \\
          & \cofohh \\
          & \cohohh
        \end{alignedat}$
        &
        $\begin{aligned}[t]
          & D_1
          \\
          & D_\omega
          \\
          & D_1
          \\
          & D_\omega
        \end{aligned}$
        % $\begin{alignedat}[t]{2}
        %   D & \coloncolonequals
        %   \foSimpleAt & & \mid G \impl D \mid  D \conj D \mid \all{x : \tau} D
        %   \\
        %   D & \coloncolonequals
        %   \simpleAt & & \mid G \impl D \mid  D \conj D \mid \all{x : \tau} D
        %   \\
        %   D & \coloncolonequals
        %   \foSimpleAt & & \mid G \impl D \mid  D \conj D \mid \all{x : \tau} D
        %   \\
        %   D & \coloncolonequals
        %   \simpleAt & & \mid G \impl D \mid  D \conj D \mid \all{x : \tau} D
        % \end{alignedat}$
        &
        $\begin{alignedat}[t]{2}
          G & \coloncolonequals
          \foGuardedAt & & \mid G \conj G \mid  G \disj G \mid \exist{x : \tau} G
          \\
          G & \coloncolonequals
          \guardedAt & & \mid G \conj G \mid G \disj G \mid \exist{x : \tau} G
          \\
          G & \coloncolonequals
          \foGuardedAt & & \mid G \conj G \mid  G \disj G \mid \exist{x : \tau} G
          \mid D \impl G \mid \all{x : \tau} G
          \\
          G & \coloncolonequals
          \guardedAt & & \mid G \conj G \mid  G \disj G
          \mid \exist{x : \tau} G \mid D \impl G \mid \all{x  : \tau} G
        \end{alignedat}$
      \end{tabular}
    }
    \caption{D- and G-formulae for coinductive uniform proofs.}
      \label{tab:up}
    \vspace*{-2em}
  \end{table}
  A $D$-formula of the shape
  $\all{\vv{x}} A_1 \conj  \dotsm \conj  A_n \impl A_0 $
  is called \emph{$H$-formula} or \emph{Horn clause} if $A_k \in \foSimpleAt$,
  and \emph{$H^g$-formula} if $A_k \in \foGuardedAt$.
  Finally, a \emph{logic program} (or \emph{program}) $P$ is a set of
  $H$-formulae.
  Note that any set of $D$-formulae in \emph{fohc} can be transformed into an
  intuitionistically equivalent set of $H$-formulae~\cite{MN12}.
\end{definition}

%% Direct adaption of the original table.
%% Note that the above table does alignment automatically and can be printed
%% in a slightly larger font.
% \begin{table}[h]
%   \begin{center}
%     \footnotesize{
%       \renewcommand{\arraystretch}{1.5}
%       \begin{tabular}{l | p{0.3\textwidth} | p{0.5\textwidth} }
%         & Program Clauses & Goals
%         \\ \hline
%         \textit{co-fohc}
%         & $D \coloncolonequals
%         A_1 \mid G \impl D \mid  D \conj D \mid \all{x} D$
%         & $G \coloncolonequals
%         A_1 \mid G \conj G \mid  G \disj G \mid \exist{x} G$
%         \\
%         \textit{co-hohc}
%         & $D \coloncolonequals
%         A_r \mid G \impl D \mid  D \conj D \mid \all{x} D$
%         & $G \coloncolonequals
%         A\hphantom{\scriptsize{1}} \mid G \conj G \mid  G\disj G
%         \mid \exist{x} G$
%         \\
%         \textit{co-fohh} & $D \coloncolonequals
%         A_1 \mid G \impl D \mid  D \conj D \mid \all{x} D$ &
%         $G \coloncolonequals
%         A_1 \mid G \conj G \mid  G \disj G \mid \exist{x} G \mid D \impl G
%         \mid \all{x} G$
%         \\
%         \textit{co-hohh} & $D \coloncolonequals
%         A_r \mid G \impl D \mid  D \conj D \mid \all{x} D$
%         & $G \coloncolonequals
%         A\hphantom{\scriptsize{1}} \mid G \conj G \mid  G \disj G
%         \mid \exist{x} G \mid D \impl G \mid \all{x} G$
%       \end{tabular}}
%   \end{center}
%   \caption{D- and G-formulae for abstract logic programming languages extended
%     with fixed-point terms.}
%   \label{tab:up}
% \end{table}

We are now ready to introduce the coinductive uniform proofs.
Such proofs are composed of two parts: an outer coinduction that has to be at
the root of a proof tree, and the usual the usual uniform proofs by
Miller et al.~\cite{Miller91:UniformProofs}.
The latter are restated in \figRef*{rules-UP}.
Of special notice is the rule $\Dec$ that mimics the operational behaviour of
resolution in logic programming, by choosing a clause $D$ from the given
program to resolve against.
The coinduction is started by the rule \Cofix{} in \figRef*{rules-CUP}.
Our proof system mimics the typical recursion with a guard condition found in
coinductive programs and proofs~\cite{Appel07:ModalTypeSystem,%
  Aczel03:InfTrees-CIT,Birkedal:GuardedRecUniverseFP,Coq94,Gimenez98}.
This guardedness condition is formalised by applying the guarding modality
$\langle \_ \rangle$ on the formula being proven by coinduction and the
proof rules that allow us to distribute the guard over certain logical
connectives, see \figRef*{rules-CUP}.
% , together with
%  and formation of an additional set of coinduction hypotheses $\Delta$\footnote{In a typical sequent $\SequentUP{G}$, the sets $P$ and $\Delta$ are considered as separate fields, \emph{not} as their union, and this fact is signified by the semicolon in notation $P;\Delta$.  In any sequent, two sets $P$ and $Q$ that are considered as their union are delimited by comma, as in $P,Q$.}.
The guarding modality may be discharged only if the guarded goal was resolved
against a clause in the initial program or any hypothesis, except for the
coinduction hypotheses. %, of course.\todo{Why ``of course''?}
This is reflected in the rule $\DecG$, where we may only pick a clause from
$P$, and is in contrast to the rule $\Dec$, in which we can pick \emph{any}
hypothesis.
The proof may only terminate with the $\Init$ step if the goal is no longer
guarded.

  % Also note that in a sequent, to denote a set  of formulae, we omit the braces. For example, instead of writing $\Sigma; P\cup\{D_1\};\{D_2\}\longrightarrow G$, we write $\Sigma; P,D_1;D_2\longrightarrow G$. Similarly, the expression $c_1:\tau_1,\ldots,c_n:\tau_n,\Sigma$ is short for $\Sigma \cup \{c_1:\tau_1,\ldots,c_n:\tau_n\}$. }. 

Note that the $\Cofix$ rule introduces a goal as a new hypothesis.
Hence, we have to require that this goal is also a definite clause.
Since coinduction hypotheses play such an important role,
they deserve a separate definition.
\begin{definition}
  Given a language $L$ from \tabRef*{up}, a formula $\varphi$ is a
  \emph{\CoindGoal} of $L$ if $\varphi$ simultaneously is a
  $D$- and a $G$-formula of $L$.
%   Formally, the syntax of coinduction hypotheses for each logic is given in
%   \tabRef{core form}.
%   For the variations  concerning $\cofohh_{\fix}$ etc., we replace again
%   ``$b$'' with ``$g$'' everywhere in the atoms.
%   % core formulae
% \begin{table}[h]
% \footnotesize{
%     \renewcommand{\arraystretch}{1.5}
%     \begin{tabular}{rlrl}
% 		\hline
% $\begin{alignedat}[t]{1}
%         & \textit{co-fohc} \\
%         & \textit{co-hohc}
%       \end{alignedat}$
%       &
% $\begin{alignedat}[t]{2}
%         M & \coloncolonequals
%         \foSimpleAt & & \mid M \conj M
%         \\
%         M & \coloncolonequals
%         \simpleAt & & \mid M \conj M
%       \end{alignedat}$
% 			&
% 	$\begin{alignedat}[t]{1}
%         & \textit{co-fohh} \\
%         & \textit{co-hohh}
%       \end{alignedat}$
%       &
% 		$\begin{alignedat}[t]{2}
%         M & \coloncolonequals
%         \foSimpleAt & & \mid M \conj M \mid M \impl M \mid \all{x : \tau} M
%         \\
%         M & \coloncolonequals
%         \simpleAt & & \mid M \conj M \mid M \impl M \mid \all{x : \tau} M
%       \end{alignedat}$
%     %% Simplified core formulae
%     % $\begin{alignedat}[t]{2}
%     %     M & \coloncolonequals
%     %     \foSimpleAt & & \mid M \conj M \mid M \impl    \foSimpleAt \mid \all{x : \tau} M
%     %     \\
%     %     M & \coloncolonequals
%     %     \simpleAt & & \mid M \conj M \mid M \impl \simpleAt  \mid \all{x : \tau} M
%     %   \end{alignedat}$
%       % \\
% 			%\hline
% \end{tabular}
%   }
%   \caption[Core formulae]{Core formulae that are allowed as coinduction
%     hypothesis.}
%   \label{tab:core form}
% \end{table}
\end{definition}

Note that the \CoindGoal* of \cofohc{} and \cofohh{} can be transformed into
equivalent $H$- or $H^g$-formulae, since any \CoindGoal is a $D$-formula.

Let us now formally introduce the coinductive uniform proof system.
% \begin{definition}
%   Let $P$ and $\Delta$ be finite sets of, respectively, definite clauses and
%   \CoindGoal*, over the signature $\TSig$, and suppose that $G$ is a goal
%   and $\varphi$ is a \CoindGoal.
%   A \emph{sequent} either has the form $\SequentUP{G}$, where the
%   \emph{uniform provability relation} $\UParrow$ is defined in
%   \figRef*{rules-UP}, or has the form $\SequentCUP$, where the
%   \emph{coinductive uniform provability relation} $\CUParrow$ is defined in
%   \figRef*{rules-CUP}.\todo{What about $\SequentUPb{D}$?}
%   Let $L$ be a language from \tabRef*{up}.
%   We say that $\varphi$ is \emph{coinductively provable} in $L$, if $P$ is a
%   set of $D$-formulae in $L$, $\varphi$ is a \CoindGoal in $L$ and $\SequentCUP$
%   holds.
% \end{definition}

\begin{definition}
  Let $P$ and $\Delta$ be finite sets of, respectively, definite clauses and
  \CoindGoal*, over the signature $\TSig$, and suppose that $G$ is a goal
  and $\varphi$ is a \CoindGoal.
  A \emph{sequent} is either a \emph{uniform provability sequent} of the form
  $\SequentUP{G}$ or $\SequentUPb{D}$ as defined in \figRef*{rules-UP}, or it
  is a \emph{coinductive uniform provability sequent} of the form $\SequentCUP$
  as defined in \figRef*{rules-CUP}.
  Let $L$ be a language from \tabRef*{up}.
  We say that $\varphi$ is \emph{coinductively provable} in $L$, if $P$ is a
  set of $D$-formulae in $L$, $\varphi$ is a \CoindGoal in $L$ and $\SequentCUP$
  holds.
\end{definition}

%\begin{definition}%[Proof in \textit{co-fohc}, \textit{co-fohh}, \textit{co-hohc} or \textit{co-hohh}]
%Let $\SequentCUP$ be true, and $L$ a language from \tabRef*{up}.  We say that $\SequentCUP$ is provable in $L$, if $P$ is a set of $D$-formulae in $L$ and $M$ is a \CoindGoal in $L$.
%\end{definition}

\begin{figure}[bt]
  \begin{spreadlines}{7pt}
    \begin{empheq}[box=\fbox]{gather*}
      \def\ScoreOverhang{1pt}
      \def\defaultHypSeparation{\hskip .05in}
      \def\labelSpacing{2pt}
      \AxiomC{$\SequentUPb{D}$}
			\AxiomC{$D\in P\cup\Delta$}
      \RightLabel{\Dec}
      \BinaryInfC{$\SequentUP{A}$}
      \bottomAlignProof
      \DisplayProof
      \quad
      \def\ScoreOverhang{1pt}
      \def\labelSpacing{2pt}
      \AxiomC{$A \conv A'$}
      \RightLabel{\Init}
      \UnaryInfC{$\SequentUPb{A'}$}
      \bottomAlignProof
      \DisplayProof
      \quad
      \def\ScoreOverhang{1pt}
      \def\labelSpacing{2pt}
      \AxiomC{}
      \RightLabel{$\TopR$}
      \UnaryInfC{$\SequentUP{\top}$}
      \bottomAlignProof
      \DisplayProof
      \\
      \AxiomC{$\SequentUPb{D}$}
      \AxiomC{$\SequentUP{G}$}
      \RightLabel{$\ImplL$}
      \BinaryInfC{$\SequentUPb{G\impl D}$}
      \DisplayProof
      \quad
      \AxiomC{$\SequentUP[P,D]{G}$}
      \RightLabel{$\ImplR$}
      \UnaryInfC{$\SequentUP{D\impl G}$}
      \DisplayProof
      \\
      \AxiomC{$\SequentUPb{D_x}$}
      \AxiomC{$x \in \set{1,2}$}
      \RightLabel{$\AndL$}
      \BinaryInfC{$\SequentUPb{D_1\conj D_2}$}
      \DisplayProof
      \quad
      \AxiomC{$\SequentUP{G_1}$}
      \AxiomC{$\SequentUP{G_2}$}
      \RightLabel{$\AndR$}
      \BinaryInfC{$\SequentUP{G_1\conj G_2}$}
      \DisplayProof
      \\
      \def\defaultHypSeparation{\hskip .1in}
      \AxiomC{$\SequentUPb{D\subst{N/x}}$}
      \AxiomC{$\guarded[\EMPTY]{N}{\tau}$}
      \RightLabel{$\AllL$}
      \BinaryInfC{$\SequentUPb{\all{x}D}$}
      \DisplayProof
      \quad
      \def\defaultHypSeparation{\hskip .1in}
      \AxiomC{$\SequentUP(c:\tau,\TSig){G\subst{c/x}}$}
			\AxiomC{$c:\tau\notin \TSig$}
      \RightLabel{$\AllR$}
      \BinaryInfC{$\SequentUP{\all{x:\tau}G}$}
      \DisplayProof
      \\
      \AxiomC{$\SequentUP{G\subst{N/x}}$}
      \AxiomC{$\guarded[\EMPTY]{N}{\tau}$}
      \RightLabel{$\ExistsR$}
      \BinaryInfC{$\SequentUP{\exist{x:\tau}G}$}
      \DisplayProof
      \quad
      \AxiomC{$\SequentUP{G_x}$}
      \AxiomC{$x \in \set{1,2}$}
      \RightLabel{$\OrR$}
      \BinaryInfC{$\SequentUP{G_1\disj G_2}$}
      \DisplayProof
    \end{empheq}
  \end{spreadlines}
%  \vspace*{-2em}
    \caption{Uniform Proof Rules}
    \label{fig:rules-UP}
%    \vspace*{-1.5em}
  \end{figure}
	
	% CUP rules
	
  \begin{figure}[bt]
  \begin{spreadlines}{7pt}
    \begin{empheq}[box=\fbox]{gather*}
		  \AxiomC{$\SequentUPg[P][\varphi]{\varphi}$}
      \RightLabel{$\Cofix$}
      \UnaryInfC{$\SequentCUP$}
      \DisplayProof
      \\
      \def\defaultHypSeparation{\hskip .05in}
      \AxiomC{$\SequentUPb{D}$}
      \AxiomC{$D \in P$}
      \RightLabel{$\DecG$}
      \BinaryInfC{$\SequentUPg{A}$}
      \DisplayProof
      \quad
      \def\defaultHypSeparation{\hskip .1in}
      \AxiomC{$\SequentUPg(c:\tau,\TSig){\varphi\subst{c/x}}$}
      \AxiomC{$c:\tau\notin \TSig$}
      \RightLabel{$\AllRG$}
      \BinaryInfC{$\SequentUPg{\all{x:\tau}\varphi}$}
      \DisplayProof
      \\
      \AxiomC{$\SequentUPg{\varphi_1}$}
      \AxiomC{$\SequentUPg{\varphi_2}$}
      \RightLabel{$\ConjRG$}
      \BinaryInfC{$\SequentUPg{\varphi_1\conj \varphi_2}$}
      \DisplayProof
      \quad
     \AxiomC{$\sequentUPg{\Sigma}{P}{\Delta, \varphi_1}{\varphi_2}$}
      \RightLabel{$\ImplRG$}
      \UnaryInfC{$\SequentUPg{\varphi_1\impl \varphi_2}$}
      \DisplayProof
    \end{empheq}
  \end{spreadlines}
%  \vspace*{-2em}
  \caption{Coinductive Uniform Proof Rules}
  \label{fig:rules-CUP}
%  \vspace*{-1.5em}
\end{figure}

%The \emph{coinduction hypothesis} $M$ must occur simultaneously as a program clause and as a goal on both sides of $\UParrow$ in the \Cofix{} rule, therefore, $M$ must be a \CoindGoal*. To guard the \Cofix{} rule from unsound applications, we introduce the guarding notation $\guardCUP{}$ within the \Cofix{} rule, signifying the guarded status of a formula. A sequent involving guarded formulae is reduced using $\ConjRG,\ImplRG,\AllRG$ and $\DecG$, which serve to pass on the guarding notation  until the guarded goal is reduced to atomic, and to ensure that a suitable clause is selected for back chaining. The main difference to notice between these rules and their uniform proof counterparts is the restriction imposed on the $\DecG$ rule regarding the set from where it can select clauses. The guarding notation $\guardCUP{}$ and the rules  $\ConjRG,\ImplRG,\AllRG$ and $\DecG$ constitute the \emph{safeguard} for the \Cofix{} rule. \appRef{safeguard motivation} illustrates the inconsistency that can result from a lack of safeguard.

% \hnote{Old Remark:

%   If we viewed H-formulae constituting a program $P$ as coinductive types, e.g.
%   viewed the clause
%   $   \kappa_{\fromP}: \all{x \; y}  \fromP \ (s\ x) \ y \  \impl \fromP \ x \ (\scons \ x \ y)$
%   as a coinductive type with a coinductive constructor $\kappa_{\fromP}$, our
%   guarding condition would amount to guarding the resulting proof terms by
%   constructors -- i.e. by applications of this clauses.
% }
The logics we have introduced impose different syntactic restrictions on
$D$- and $G$-formulae, and will therefore admit coinduction goals of different strength.
This ability to explicitly use stronger coinduction hypotheses within a goal-directed search
was missing in CoLP, for example. And it
allows us to account for different coinductive properties of Horn clauses as
described in the introduction.
% And different kinds of coinductive models demand coinductive goals of different strength. 
We finish this section by illustrating this strengthening.
%coinductive uniform proofs for the two
%opposite corners $\cofohc$ and $\cohohh_{\fix}$ of the cube in the introduction.
% Further examples are in the \appRef{ex}.

The first example is one for the logic \cofohc{}, in which we illustrate the
framework on the problem of type class resolution.
\begin{example}
  Let us restate the Haskell type class inference problem discussed in the
  introduction in terms of Horn clauses:
  \begin{alignat*}{2}
    \SwapAboveDisplaySkip
    \kappa_{\eqBase} &:    \eq \ \eqBase  \\
    \kappa_{\odd}  &: \all{x} \eq \ x \conj \eq \ (\even \ x  )  && \impl  \eq \ (\odd\ x) \\
    \kappa_{\even} &: \all{x} \eq \ x \conj \eq \ (\odd \ x) && \impl
    \eq \ (\even \ x )
    % \label{eq:horn clause: nats}
  \end{alignat*}

  To prove $\eq\ (\odd \ \eqBase)$ for this set of Horn clauses, it is
  sufficient to use this formula directly as coinduction hypothesis,
  as shown in \figRef*{cofohc Haskell}.
  Note that this formula is indeed a \CoindGoal of \cofohc{},
  hence we find ourselves in the simplest scenario of coinductive proof search.
  In \tabRef*{herbrand-overview},
  $\gamma_1$ is a representative for this kind of coinductive proofs with simplest atomic goals.

  \begin{figure}
    \footnotesize{
\begin{prooftree}
  \AxiomC{}
  	\RightLabel{$\Init$}
   \UnaryInfC{$\sequentUPb{\TSig}{P}{\varphi}{\varphi}{\eq\ (\odd \ \eqBase)}$}
        \RightLabel{$\Dec$}
         \UnaryInfC{$\vdots$}
        \RightLabel{$\AllL$}
        \UnaryInfC{$\sequentUPb{\TSig}{P}{\varphi}{\kappa_{\even}}{\eq\ (\even \ \eqBase)}$}
        \RightLabel{$\Dec$}
        \UnaryInfC{$\sequentUPb{\TSig}{P}{\varphi}{}{\eq\ (\even \ \eqBase)}$}
        \UnaryInfC{$\spadesuit$}
\end{prooftree}
\vspace*{-2em}
\begin{prooftree}
	\AxiomC{}
	\RightLabel{$\Init$}
	\UnaryInfC{$\sequentUPb{\TSig}{P}{\varphi}{  \eq\ (\odd \ \eqBase)}{\eq\ (\odd \ \eqBase)}$}
          \AxiomC{}
          \RightLabel{$\Init$}
          	\UnaryInfC{$\sequentUPb{\TSig}{P}{\varphi}{\kappa_{\eqBase}}{\eq\ \ \eqBase}$}
           \RightLabel{$\Dec$}
	\UnaryInfC{$\sequentUPb{\TSig}{P}{\varphi}{}{\eq\ \ \eqBase}$}
        \AxiomC{$\spadesuit$}
        \RightLabel{$\land R$}
  \BinaryInfC{$\sequentUP{\TSig}{P}{\varphi}{\eq \ \eqBase \land \  \eq\ (\even \ \eqBase  )}$}
	\RightLabel{$\ImplL$}
	\BinaryInfC{$\sequentUPb{\TSig}{P}{\varphi}{ \eq \ \eqBase \land \  \eq\ (\even \ \eqBase  )  \impl \eq\ (\odd \ \eqBase)}{\eq\ (\odd \ \eqBase)}$}
	\RightLabel{$\AllL$}
	\UnaryInfC{$\sequentUPb{\TSig}{P}{\varphi}{\kappa_{\odd}}{\eq\ (\odd \ \eqBase)}$}
	\RightLabel{$\DecG$}
	\UnaryInfC{$\sequentUP{\TSig}{P}{\varphi}{\guardCUP{\eq\ (\odd \ \eqBase)}}$}
	\RightLabel{\Cofix}
	\UnaryInfC{$\sequentCUP{\TSig}{P}{\eq\ (\odd \ \eqBase)}$}
\end{prooftree}
}
  \vspace*{-1em}
  \caption{The \cofohc{} proof for Horn clauses arising from Haskell Type class
    examples.
    $\varphi$ abbreviates the coinduction hypothesis $\eq\ (\odd \ \eqBase)$.
    Note its use in the branch $\spadesuit$.}
  \label{fig:cofohc Haskell}
  \vspace*{-1em}
  \end{figure}

  It was pointed out in~\cite{FKS15} that Haskell's type class inference can
  also give rise to  irregular corecursion.
  Such cases may require the more general coinduction hypothesis (e.g. universal
  and/or implicative) of \cofohh{} or \cohohh{}.
  The below set of Horn clauses is a simplified representation of a problem
  given in~\cite{FKS15}:
  \begin{alignat*}{2}
    \SwapAboveDisplaySkip
    % \label{eq:CUP clause Tom coin}
    \kappa_{\eqBase} &:  \eq \; \eqBase\\
    \kappa_{s}         &: \all{x} (\eq \; x) \conj \eq \; (s \; (g \; x))
    && \impl \eq \; (s \; x )\\
    % \label{eq:CUP clause Tom in}
    \kappa_{g}         &: \all{x} \eq \; x
    && \impl \eq \; (g \; x )
    % \label{eq:CUP clause Tom at}
    % \label{eq:FOL clause Tom coin}
    % \all{x} \later \eq(s(g\ x)) \conj \later (\eq\ x)  & \impl \eq(s\ x)\\
    % \label{eq:FOL clause Tom in}\all{x} \later (\eq\ x) & \impl \eq(g\ x)
  \end{alignat*}

  Trying to prove $\eq \; (s \; \eqBase)$ by using $\eq \; (s \; \eqBase)$
  directly as a coinduction hypothesis is deemed to fail, as the coinductive
  proof search is irregular and this coinduction hypothesis would not be
  applicable in any guarded context.
  But it is possible to prove $\eq \; (s \; \eqBase)$ as a corollary of
  another theorem: $\all{x}(\eq \; x) \impl \eq \; (s \; x)$.
  Using this formula as coinduction hypothesis leads to a successful proof,
  which we omit here.
  From this more general goal, we can derive the original goal by instantiating
  the quantifier with $\eqBase$ and eliminating the implication with
  $\kappa_{\eqBase}$.
  This second derivation is sound with respect to the models, as we show
  in \thmRef*{cut theorem infi model}.
\end{example}

We encounter $\gamma_2$ from \tabRef*{herbrand-overview} in a similar situation:
To prove $p \, a$, we first have to prove $\all{x}p \ x$ in  \cofohh{}, and
then obtain $p \ a$ as a corollary by appealing
to \thmRef*{cut theorem infi model}.
The next example shows that we can cover all cases in
\tabRef*{herbrand-overview} by providing a proof in $\cohohh_{\fix}$ that
involves irregular recursive terms.

\begin{example}
  \label{ex:from-cup-proof}
  Recall the clause
  $\all{x \; y}  \fromP \; (s\; x) \; y \impl \fromP \; x \; (\scons \; x \; y)$
  that we named $\kappa_{\fromP}$ in the introduction.
  Proving $\exist{y} \fromP \; 0 \; y$ is again not possible directly.
  %, simply
  %for the technical reason that it is not a \CoindGoal.
  %Intuitively, the reason is that the existential quantifier would force us to
  %construct a term $M$ and at the same time proving coinductively that
  %$\fromP \; 0 \; M$ holds.
  %This would be too much to ask from a proof system with operational semantics.%
  %\todo{Can we say this better?}{}
  Instead, we can use the term
  $\fromFun = \fix[f] \lam{x} \scons \; x \; (f \; (s \; x))$ from
  \exRef*{guarded-terms} and prove $\all{x} \fromP \; x \; (\fromFun \; x)$
  coinductively, as shown in \figRef*{LandscapeFigure cohohh}.
  This formula  gives a coinduction hypothesis of sufficient generality.
  Note that the correct coinduction hypothesis now requires the fixed point
  definition of an infinite stream of successive numbers and universal
  quantification in the goal.
  Hence the need for the richer language of $\cohohh_{\fix}$.
  From this more general goal we can derive our initial goal
  $\exists\ y. \fromP \; 0 \; y$ by instantiating $y$ with $\fromFun \; 0$.

  \begin{figure}
    \footnotesize{
      \[
        \infer[]{
          \spadesuit
        }{
          \infer[\Dec]
          {
            \sequentUP{c,\TSig}{P}{\varphi}{\fromPred{(s\ c)}{(\fromFun \ (s \ c))}}
          }
          {
            \infer[\AllL]
            {
              \sequentUPb{c,\TSig}{P}{\varphi}{\varphi}{\fromPred{(s\ c)}{(\fromFun \ (s \ c))}}
            }
            {
              \infer[\Init]
              {
                \sequentUPb{c,\TSig}{P}{\varphi}{\fromPred{(s\ c)}{(\fromFun \ (s \ c))}}{\fromPred{(s\ c)}{(\fromFun \ (s \ c))}}
              }
              {}
            }
          }
        }
      \]
      \vspace*{-1em}
      \[\infer[\Cofix]
        {
          \sequentCUP{\TSig}{P}{\all{x}\fromPred{x}{(\fromFun \ x)}}
        }
        {
          \infer[\AllRG]
          {
            \sequentUP{\TSig}{P}{\varphi}{\guardCUP{\all{x}\fromPred{x}{(\fromFun \ x)}}}
          }
          {
            \infer[\DecG]
            {
              \sequentUP{c,\TSig}{P}{\varphi}{\guardCUP{\fromPred{c}{(\fromFun \ c)}}}
            }
            {
              \infer[\AllL \left(\text{2 times}\right)]
              {
                \sequentUPb{c,\TSig}{P}{\varphi}{  \kappa_{\fromP}}{\fromPred{c}{(\fromFun \ c)}}
              }
              {
                \infer[\ImplL]
                {
                  \sequentUPb{c,\TSig}{P}{\varphi}{
                    \fromP \, (s\, c) \, (\fromFun \ (s \ c))
                    \impl \fromP \, c \, (\scons \, c \, (\fromFun \ (s \ c)))
                  }{\fromPred{c}{(\fromFun \ c)}}
                }
                {
                  \infer[\Init]
                  {
                    \sequentUPb{c,\TSig}{P}{\varphi}{\fromPred{c}{(\sconsFunc{c}{(\fromFun \ (s \ c))})} }{\fromPred{c}{(\fromFun \ c)}}
                  }
                  {}
                  &
                  \spadesuit
                }
              }
            }
          }
        }\]	
%\dotfill $\	\spadesuit\ $ \dotfill 	

%\[ \infer[\Dec]
%					%	{
					%	   \sequentUP{c,\TSig}{P}{\varphi}{\fromPred{(s\ c)}{(\fromFun \ (s \ c))}}
					%	}
					%	{
					%		\infer[\AllL]
					%		{
					%		\sequentUPb{c,\TSig}{P}{\varphi}{\varphi}{\fromPred{(s\ c)}{(\fromFun \ (s \ c))}}
					%		}
					%		{
					%			\infer[\Init]
					%			{
					%			\sequentUPb{c,\TSig}{P}{\varphi}{\fromPred{(s\ c)}{(\fromFun \ (s \ c))}}{\fromPred{(s\ c)}{(\fromFun \ (s \ c))}}
					%			}
					%			{}
					%		}
					%	}\]
	
          % } % resizebox
    }
    \vspace*{-1em}
    \caption{The $\cohohh_{\fix}$ proof for
      $\varphi = \all{x} \fromP \; x \; (\fromFun \; x)$.
      Note that
      % % $\fromstream{}$ is defined in Example%~\ref{eg: zeros and number stream and from as fix point}
      % $c$ is an arbitrary eigenvariable,  $\varphi$ abbreviates the coinduction
      % hypothesis $\all{x}\fromPred{x}{(\fromFun \ x)}$, and
      the last step of the leftmost branch involves
      ${\fromPred{c}{(\sconsFunc{c}{(\fromFun \ (s \ c))})} }\conv {\fromPred{c}{(\fromFun \ c)}}	$.
      % Moreover, the last $\AllL$ step uses the substitution
      % $\protect\subst{(s\ c)/x}$.
    }
    \vspace*{-1em}
    \label{fig:LandscapeFigure cohohh}
  \end{figure}
\end{example}

There are examples of coinductive proofs that require a fixed point
definition of an infinite stream, but do not require the syntax of higher-order
terms or hereditary Harrop formulae.
Such proofs can be performed in the $\cofohc_{\fix}$ logic.
A good example is a proof that the stream of zeros satisfies the Horn clause
theory defining the predicate $\stream$ in the introduction.
The goal $(\stream\  s_{0})$, with $s_{0} =  \fix[x] \scons \; 0 \; x$ can be
proven directly by coinduction. %, see \appRef{ex}.
% , as will be
% further discussed in \secRef*{heuristics}.\todo{Discuss here}
Similarly, one can type self-application with the infinite type
$a = \fix[t] t \to b$ for some given type $b$.
The proof for $\typedP \; [x:a] \; (\appP \; x \; x) \; b$ is then in
$\cofohc_{\fix}$.
Finally, the clause $\gamma_3$ is also in this group.
More generally, circular unifiers obtained from CoLP's~\cite{GuptaBMSM07} loop
detection yield immediately guarded fixed point terms, and thus CoLP
corresponds to coinductive proofs in the logic $\cofohc_{\fix}$.
A general discussion of Horn clause theories that describe infinite objects was
given in~\cite{KL17}, where the above logic programs were identified as being
productive.

% \knote{Consider actually showing this last fact in the extended version? or else leave it until we have
%   a next paper}

%\knote{Two examples of CUP proofs are probably enough at this stage?}

%\knote{below analysis I think old text by Yue, we may want to trim it?}
%\subsection*{Discussion}

%It is not surprising that the complexity of coinduction hypothesis allowable in coinductive reasoning is proportionate to the expressiveness of the language in which we perform the  reasoning. However, the \textsc{co-fix} rule adds one more dimension as to why \CoindGoal* are of interest: \CoindGoal* were studied because they can be stored as proved lemmas and be reused in later proof~\cite{MNPS91, HarlandPhDthesis}; now due to the same syntactic character, \CoindGoal* constitute the restricted set of formulae that we can coinductively prove using a fixed point style coinductive principle.
%The fact that the \textsc{co-fix} rule can only be applied once and as the first step in a proof,  is a simplification that helps highlighting the basic proof-theoretic procedure of coinduction. The absence of nested coinduction in a proof can be mitigated by a  model extension theorem %(Theorem~\ref{them: cut theorem infi model})
% which effectively takes up the role of a coinductive cut rule.

\section{Coinductive Uniform Proofs and Intuitionistic Logic}
\label{sec:loeb-translation}

In the last section, we introduced the framework of coinductive uniform proofs,
which gives an operational account to proofs for coinductively interpreted
logic programs.
Having this framework at hand, we need to position it in the existing ecosystem
of logical systems.
The goal of this section is to prove that coinductive uniform proofs are in
fact constructive.
We show this by first introducing an extension of intuitionistic first-order
logic that allows us to deal with recursive proofs for coinductive predicates.
Afterwards, we show that coinductive uniform proofs are sound relative to this
logic by means of a proof tree translation.
The model-theoretic soundness proofs for both logics will be provided in
 \secRef{soundness}.

% \subsection{Löb Logic}
% \label{sec:loeb-logic}

We begin by introducing an extension of intuitionistic first-order logic
with the so-called \emph{later modality}, written $\later$.
This modality is the essential ingredient that allows us to equip proofs
with a controlled form of recursion.
The later modality stems originally from provability logic, which characterises
transitive, well-founded Kripke
frames~\cite{Clouston15:SequentCalcToposOfTrees,%
  Solovay1976:ProvabilityModalLogic},
and thus allows one to carry out induction without an explicit induction
scheme~\cite{Beklemishev1999:ParameterFreeInduction}.
Later, the later modality was picked up by the type-theoretic community to
control recursion in coinductive programming~\cite{%
  Appel07:ModalTypeSystem,Atkey13:GuardedRec,Bizjak16:GuardedDTT,%
  Mogelberg14:GuardedRec,Nakano00:ModalityRec},
mostly with the intent to replace syntactic guardedness checks for coinductive
definitions by type-based checks of well-definedness.

Formally, the logic $\iFOLm$ is given by the following definition.
\begin{definition}
  \label{def:loeb-logic}
  The formulae of $\iFOLm$ are given by \defRef*{formulae}
  and the rule:
  \begin{equation*}
    \def\defaultHypSeparation{\hskip .05in}
    %% NOTE: We do not need bottom and negation
    % \AxiomC{}
    % \UnaryInfC{$\validForm{\bot}$}
    % \bottomAlignProof
    % \DisplayProof
    % \quad
    \AxiomC{$\validForm{\varphi}$}
    \UnaryInfC{$\validForm{\later \varphi}$}
    \bottomAlignProof
    \DisplayProof
  \end{equation*}
  Conversion extends to these formulae in the obvious way.
  Let $\varphi$ be a formula and $\Delta$ a sequence of formulae in $\iFOLm$.
  We say $\varphi$ is
  \emph{provable in context $\Gamma$ under the assumptions $\Delta$}
  in $\iFOLm$, if $\inferFOL{\varphi}$ holds.
  The \emph{provability relation} $\vdash$ is thereby given
  inductively by the rules
  in \figRef*{rules-std-connectives} and \figRef*{rules-later}.
  \begin{figure}[bt]
  \begin{spreadlines}{7pt}
    \begin{empheq}[box=\fbox]{gather*}
      \def\defaultHypSeparation{\hskip .1in}
      \AxiomC{$\validForm{\Delta}$}
      \AxiomC{$\varphi \in \Delta$}
      \RightLabel{\Proj}
      \BinaryInfC{$\inferFOL{\varphi}$}
      \DisplayProof
      \quad
      \def\defaultHypSeparation{\hskip .1in}
      \AxiomC{$\inferFOL{\varphi'}$}
      \AxiomC{$\varphi \conv \varphi'$}
      \RightLabel{\Conv}
      \BinaryInfC{$\inferFOL{\varphi}$}
      \DisplayProof
      \quad
      \AxiomC{$\validForm{\Delta}$}
      \RightLabel{\TopI}
      \UnaryInfC{$\inferFOL{\top}$}
      \DisplayProof
      % \quad
      % \AxiomC{$\validForm{\varphi}$}
      % \AxiomC{$\inferFOL{\bot}$}
      % \RightLabel{\BotE}
      % \BinaryInfC{$\inferFOL{\varphi}$}
      % \DisplayProof
      \\
      \AxiomC{$\inferFOL{\varphi}$}
      \AxiomC{$\inferFOL{\psi}$}
      \RightLabel{\AndI}
      \BinaryInfC{$\inferFOL{\varphi \conj \psi}$}
      \DisplayProof
      \quad
      \AxiomC{$\inferFOL{\varphi_1 \conj \varphi_2}$}
      \AxiomC{$i \in \set{1,2}$}
      \RightLabel{\AndE[i]}
      \BinaryInfC{$\inferFOL{\varphi_i}$}
      \DisplayProof
      \\
      \def\defaultHypSeparation{\hskip .1in}
      \AxiomC{$\inferFOL{\varphi_i}$}
      \AxiomC{$\validForm{\varphi_j}$}
      \AxiomC{$j \ne i$}
      \RightLabel{\OrI[i]}
      \TrinaryInfC{$\inferFOL{\varphi_1 \disj \varphi_2}$}
      \DisplayProof
      \quad
      \AxiomC{$\inferFOL[\Delta, \varphi_1]{\psi}$}
      \AxiomC{$\inferFOL[\Delta, \varphi_2]{\psi}$}
      \RightLabel{\OrE}
      \BinaryInfC{$\inferFOL[\Delta, \varphi_1 \disj \varphi_2]{\psi}$}
      \DisplayProof
      \\
      \AxiomC{$\inferFOL[\Delta, \varphi]{\psi}$}
      \RightLabel{\ImplI}
      \UnaryInfC{$\inferFOL{\varphi \impl \psi}$}
      \DisplayProof
      \quad
      \AxiomC{$\inferFOL{\varphi \impl \psi}$}
      \AxiomC{$\inferFOL{\varphi}$}
      \RightLabel{\ImplE}
      \BinaryInfC{$\inferFOL{\psi}$}
      \DisplayProof
      \\
      \AxiomC{$\inferFOL(\Gamma, x : \tau){\varphi}$}
      \AxiomC{$x \not\in \Gamma$}
      \RightLabel{\AllI}
      \BinaryInfC{$\inferFOL{\all{x : \tau} \varphi}$}
      \DisplayProof
      \quad
      \def\defaultHypSeparation{\hskip .1in}
      \AxiomC{$\inferFOL{\all{x : \tau} \varphi}$}
      \AxiomC{$M : \tau \in \GuardedTerms{\TSig}(\Gamma)$}
      \RightLabel{\AllE}
      \BinaryInfC{$\inferFOL{\varphi \subst{M/x}}$}
      \DisplayProof
      \\
      \def\ScoreOverhang{1pt}
      \def\defaultHypSeparation{\hskip .1in}
      \AxiomC{$M : \tau \in \GuardedTerms{\TSig}(\Gamma)$}
      \AxiomC{$\inferFOL{\varphi \subst{M/x}}$}
      \RightLabel{\ExI}
      \BinaryInfC{$\inferFOL{\exist{x : \tau} \varphi}$}
      \DisplayProof
      \quad
      \def\ScoreOverhang{1pt}
      \def\defaultHypSeparation{\hskip .1in}
      \AxiomC{$\validForm{\psi}$}
      \AxiomC{$\inferFOL(\Gamma, x : \tau)[\Delta, \varphi]{\psi}$}
      \AxiomC{$x \not\in \Gamma$}
      \RightLabel{\ExE}
      \TrinaryInfC{$\inferFOL[\Delta, \exist{x : \tau} \varphi]{\psi}$}
      \DisplayProof
    \end{empheq}
  \end{spreadlines}
    \caption{Intuitionistic Rules for Standard Connectives}
    \label{fig:rules-std-connectives}
%    \vspace*{-1em}
  \end{figure}
  \begin{figure}[bt]
  \begin{empheq}[box=\fbox]{gather*}
    \AxiomC{$\inferFOL{\varphi}$}
    \RightLabel{\Next}
    \UnaryInfC{$\inferFOL{\later \varphi}$}
    \DisplayProof
    \quad
    \AxiomC{$\inferFOL{\later (\varphi \impl \psi)}$}
    \RightLabel{\Mon}
    \UnaryInfC{$\inferFOL{\later \varphi \impl \later \psi}$}
    \DisplayProof
    \quad
    \AxiomC{$\inferFOL[\Delta, \later \varphi]{\varphi}$}
    \RightLabel{\FP}
    \UnaryInfC{$\inferFOL{\varphi}$}
    \DisplayProof
  \end{empheq}
    \caption{Rules for the Later Modality}
    \label{fig:rules-later}
%    \vspace*{-1em}
  \end{figure}
\end{definition}

The rules in \figRef*{rules-std-connectives} are the usual rules for
intuitionistic first-order logic and should come at no surprise.
More interesting are the rules in \figRef*{rules-later}, where the
rule \FP{} introduces recursion into the proof system.
Furthermore, the rule \Mon{} allows us to to distribute the later
modality over implication, and consequently over conjunction and
universal quantification.
This is essential in the translation in
\thmRef*{CUP-iFOLm-sound Version II} below.
Finally, the rule \Next{} gives us the possibility to proceed
without any recursion, if necessary.

Note that so far it is not possible to use the assumption
$\later \varphi$ introduced in the \FP{}-rule.
The idea is that the formulae of a logic program provide us the
obligations that we have to prove, possibly by recursion, in order
to prove a coinductive predicate.
This is cast in the following definition.
\begin{definition}
  \label{def:guarding}
  Given an $H^g$-formula $\varphi$ of the shape
  $\all{\vv{x}} (A_1 \conj \dotsm \conj A_n) \impl \psi$, we define
  its \emph{guarding} $\guard{\varphi}$ to be
  $\all{\vv{x}} (\later A_1 \conj \dotsm \conj \later A_n) \impl \psi$.
  For a logic program $P$, we define its guarding $\guard{P}$
  by guarding each formula in $P$.
  % Moreover, we denote by $\later \Delta$ the collection of formulae
  % that arises by putting every formula in $\Delta$ under $\later$.
\end{definition}

\begin{toappendix}
The following admissible rules are easily derivable in the logic $\iFOLm$ and
are essential in showing soundness of $\cohohh_{\fix}$ with respect to $\iFOLm$.
\begin{lemma}
  \label{lem:technical-derivations}
  \begin{equation*}
    \AxiomC{$\validForm{\varphi}$}
    \AxiomC{$\varphi \conv \psi$}
    \BinaryInfC{$\validForm{\psi}$}
    \bottomAlignProof
    \DisplayProof
    \quad
    \AxiomC{$\inferFOL{\varphi}$}
    \UnaryInfC{$\validForm{\varphi}$}
    \bottomAlignProof
    \DisplayProof
    \quad
    \AxiomC{$\inferFOL{\varphi}$}
    \AxiomC{$x : \tau \not\in \Gamma$}
    \RightLabel{\Weak}
    \BinaryInfC{$\inferFOL(\Gamma,x:\tau){\varphi}$}
    \bottomAlignProof
    \DisplayProof
  \end{equation*}
  \begin{gather*}
    \AxiomC{$\inferFOL[\later \varphi_1 \impl \later \varphi_2]{\psi}$}
    \RightLabel{\MonL}
    \UnaryInfC{$\inferFOL[\later(\varphi_1 \impl \varphi_2)]{\psi}$}
    \DisplayProof
    \\[7pt]
    \AxiomC{$\inferFOL{\later(\varphi \conj \psi)}$}
    \RightLabel{\LaterPresConjR}
    \UnaryInfC{$\inferFOL{\later \varphi \conj \later \psi}$}
    \DisplayProof
    \quad
    \AxiomC{$\inferFOL[\Delta, \later \varphi_1 \conj \later \varphi_2]{\psi}$}
    \RightLabel{\LaterPresConjL}
    \UnaryInfC{$\inferFOL[\later(\varphi_1 \conj \varphi_2)]{\psi}$}
    \DisplayProof
    \\[7pt]
    \AxiomC{$\inferFOL{\later(\all{x : \tau} \varphi)}$}
    \RightLabel{\LaterPresAllR}
    \UnaryInfC{$\inferFOL{\all{x : \tau} \later\varphi}$}
    \DisplayProof
    \quad
    \AxiomC{$\inferFOL[\Delta, \all{x : \tau} \later \varphi]{\psi}$}
    \RightLabel{\LaterPresAllL}
    \UnaryInfC{$\inferFOL[\Delta, \later(\all{x : \tau} \varphi)]{\psi}$}
    \DisplayProof
    \\[7pt]
    \AxiomC{$\inferFOL{\later \varphi \conj \later \psi}$}
%    \RightLabel{\LaterPresConjI}
    \UnaryInfC{$\inferFOL{\later(\varphi \conj \psi)}$}
    \DisplayProof
    \quad
    \AxiomC{$\inferFOL{\all{x : \tau} \later\varphi}$}
%    \RightLabel{\LaterPresAll}
    \UnaryInfC{$\inferFOL{\later(\all{x : \tau} \varphi)}$}
    \DisplayProof
  \end{gather*}
\end{lemma}
\begin{proof}
  \begin{itemize}
  \item Preservation of well-formed formulae under conversion follows
    by induction on formulae from type preservation of reductions.
  \item The well-formedness of $\varphi$ follows from provability
    by induction on $\varphi$.
  \item Weakening is also given by induction on $\varphi$.
  \item The other rules follow from implication introduction and elimination,
    and monotonicity of $\later$.
    \qedhere
\end{itemize}
\end{proof}
\end{toappendix}

The translation given in \defRef*{guarding} of a logic program into formulae
that admit recursion corresponds unfolding a coinductive predicate,
cf.~\cite{Basold18:BreakingTheLoop}.
We show now how to transform a coinductive uniform proof tree into
a proof tree in $\iFOLm$, such that the recursion and guarding mechanisms
in both logics match up.
\begin{theoremrep}
  \label{thm:CUP-iFOLm-sound Version II}
  If $P$ is a logic program over a first-order signature $\TSig$
  and the sequent $\sequentCUP{\TSig}{P}{\varphi}$ is provable in
  $\cohohh_{\fix}$, then $\inferFOL()[\guard{P}]{\varphi}$ is provable
  in $\iFOLm$.
\end{theoremrep}
To prove this theorem, one uses that each coinductive uniform proof tree
starts with an initial tree that has an application of the $\Cofix$-rule at
the root and that eliminates the guard by using the rules in \figRef*{rules-CUP}.
At the leaves of this tree, one finds proof trees that proceed only by
means of the rules in \figRef*{rules-UP}.
The initial tree is then translated into a proof tree in $\iFOLm$ that
starts with an application of the \FP{}-rule, which corresponds to the
$\Cofix$-rule, and that simultaneously transforms the coinduction
hypothesis and applies introduction rules for conjunctions etc.
This ensures that we can match the coinduction hypothesis with the
guarded formulae of the program $P$.
% Further details of the proof, and an example of such proof transformation are
% given in the \appRef{proofs}.

\begin{appendixproof}
  We provide a sketch of the proof.
  First, we note that the coinduction goal in $\cohohh_{\fix}$ is given by the
  following grammar.
  \begin{equation*}
    \mathrm{CG} \coloncolonequals
    \guardedAt
    \mid \mathrm{CG} \impl \mathrm{CG}
    \mid  \mathrm{CG} \conj \mathrm{CG}
    \mid \all{x : \tau} \mathrm{CG}
  \end{equation*}
  Thus, a coinduction goal is the restriction of FOL to implication, conjunction
  and universal quantification.
  Note that such a coinduction goal is intuitionistically equivalent to a
  conjunction of Horn-clauses.

  Assume that we are given a uniform proof tree $T$.
  We translate this tree into a proof tree $T'$ in $\iFOLm$.
  The proof proceeds in the following steps.
  \begin{enumerate}
  \item The first step of a proof tree $T$ starting in
    $\sequentCUP{\Sigma}{P}{\varphi}$ must be an application of the \Cofix{}
    rule to a proof tree $T_1$ ending in
    $\sequentUP{\Sigma}{P}{\varphi}{\guardCUP{\varphi}}$.
    This step can be directly translated into an application of the \Lob{} rule.
    Hence, if $T'_1$ is the translation of $T_1$ with conclusion
    $\inferFOL()[\guard{P}, \later \varphi]{\varphi}$, then
    $T'$ is given by applying \FP{} to $T'_1$, thereby obtaining a proof tree
    ending in the desired sequent
    $\inferFOL()[\guard{P}]{\varphi}$.
  \item The next step must then be either $\AllRG$, $\ConjRG$, $\ImplRG$
    or $\DecG$.
    To prove this by induction on the proof tree, we need to define coinduction
    goal contexts.
    These are contexts $\varphi [-]$ with a hole $[-]$, such that plugging an
    atom from $\simpleAt[\omega]$ into the hole yields a coinduction goal.
    More generally, we will need contexts with multiple holes $[-]_i$ that are
    indexed from $0$ to $n$ for some $n \in \N$.
    Formally, such contexts are given by the following grammar.
    \begin{align*}
      H & \coloncolonequals
      [-]_i
      \mid [-]_i \impl H
      \mid \all{x : \tau} H
      \\
      C & \coloncolonequals
      H
      \mid C \conj C
    \end{align*}
    Let $C$ be a context, we write $C[\vv{\varphi}]$ for the
    formula that arises by replacing the holes $[-]_i$ by $\varphi_i$.
    Note that this may result in binding of free variables in
    $\varphi_i$ and $\psi$.

    We prove by induction on proof trees that for any context $\Gamma$,
    any set of formulae $P$, any context $C$ and any proof for
    $\sequentUP{\Sigma, \Gamma}{P}{C[\varphi], \Delta}{\guardCUP{\varphi}}$
    that there is a proof for
    $\inferFOL(\Gamma)[\guard{P}, C[\later \varphi], \Delta]{\varphi}$.
    The translation for this step follows then by taking $\Gamma$ and $\Delta$
    to be empty and $C$ to be $[-]$.

    In the $\AllRG$ case, we get a proof tree for
    $\sequentUP{\Sigma, \Gamma}{P}{C[\all{x} \varphi], \Delta}{
      \guardCUP{\all{x} \varphi}}$
    that has the sequent
    $\sequentUP{\Sigma, \Gamma, x}{P}{C[\all{x} \varphi], \Delta}{
      \guardCUP{\varphi}}$
    as its only premise.
    By putting $C' = C[\all{x} [-]]$, this premise can be written as
    $\sequentUP{\Sigma, \Gamma, x}{P}{C'[\varphi], \Delta}{
      \guardCUP{\varphi}}$
    from which we obtain by induction a proof tree for
    $\inferFOL(\Gamma, x)[\guard{P}, C'[\later \varphi], \Delta]{\varphi}$.
    Using the derived rule \LaterPresAllL{} from \lemRef*{technical-derivations}
    and by the rule \AllI{}, we thus obtain a proof tree for the sequent
    $\inferFOL(\Gamma)[\guard{P}, C[\later (\all{x} \varphi)]]{
      \all{x} \varphi}$.

    For the cases $\ConjRG$ and $\ImplRG$, one proceeds similarly as
    for $\AllRG$ by appealing to the fact that $\later$ preserves conjunction
    and implication, respectively.
    The only things to be taken care of are the multi-contexts in the
    conjunction and the extension of $\Delta$ in the implication case.
    Finally, the $\DecG$ rule is dealt with in the next step.

  \item
    For an application of either of the decide rules, there are generally two
    cases to consider: either the clause $D$ is selected from $P \cup \Delta$ by
    $\DecG$ or $\Dec$, or $\Dec$ selects $C[\vv{A}]$.
    In both cases, we proceed by induction to analyse of the proof tree for
    $\sequentUPb{\Sigma, \Gamma}{P}{C[\vv{A}], \Delta}{D}{B}$.
    \todo[inline,disable]{Procedure for transforming such proofs is
      outlined in the picture in the wiki.}

    Define $H \colonequals C[\vv{\later A}]$.
    We then obtain the following cases from the fact that $D$ and $H$ are
    Horn-clauses with the later modality in specific places.
    \begin{enumerate}
    \item $D \in P$ is selected.
      Then the proof tree in $\iFOLm$ will have at its root
      $\inferFOL[\guard{P}, H, \Delta]{B}$ and at its leaves sequents of the form
      $\inferFOL[\guard{P}, H, \Delta]{\later C}$ for some atoms $C$.
    \item $C[\vv{A}]$ is selected.
      Then the resulting proof tree in $\iFOLm$ will have at its root
      $\inferFOL[\guard{P}, H, \Delta]{\later A_k}$ for some $k$,
      and as its leaves sequents of the form
      $\inferFOL[\guard{P}, H, \Delta]{\later A_i}$ for some $i$.
    \end{enumerate}
    Our goal is now to combine such proof trees.
    The only mismatch might occur when we have a proof tree that has
    $\inferFOL[\guard{P}, H]{B}$ as root (first case) that has to be
    attached to a leaf of another proof tree (from either case), which
    will be of the form $\inferFOL[\guard{P}, H]{\later C}$ for some atom $C$.
    Since this match arises from a uniform proof, we have that $C = B$.
    Hence, we can combine these two trees by appealing to the \Next{} rule:
    \begin{equation*}
      \AxiomC{$\vdots$}
      \UnaryInfC{$\inferFOL[\guard{P}, H]{B}$}
      \RightLabel{\Next}
      \UnaryInfC{$\inferFOL[\guard{P}, H]{\later B}$}
      \UnaryInfC{$\vdots$}
      \DisplayProof
    \end{equation*}
    In all the other cases, the trees can be combined directly.
    \qedhere
  \end{enumerate}
\end{appendixproof}

\begin{toappendix}
\begin{example}\label{ex:LOB TRAN FROM}

  Recall the following clause from the introduction.
\begin{equation}
\all{x}\all{t} \fromPred{(s\ x)}{t} \impl \fromPred{x}{(\sconsFunc{x}{t})}
\label{eq:horn clause: from}
\end{equation}

In \exRef*{from-cup-proof}, we provided the CUP proof for
$\all{x} \fromP \; x \; (\fromFun \; x)$.
In this example, we show how that proof is translated in a proof in $\iFOLm$.

The guarding of clause \eqRef{horn clause: from}
is given by the clause \eqRef{FOL clause: from}.
\begin{equation}
  \label{eq:FOL clause: from}
  \all{x}\all{t} \later (\fromPred{(s\ x)}{t})
   \impl \fromPred{x}{(\sconsFunc{x}{t})} \\
\end{equation}

To save space, when we build a proof in $\iFOLm$ using
\AllI, \AllE{} or \Conv, etc., we may omit the condition branch, which
is $x : \tau \notin \Gamma$, $\typed{M}{\tau}$ or $\psi \conv \psi'$
respectively, if and only if we know that the condition holds.

Now let $\guard{P}$ denote the singleton set of clause \eqRef{FOL clause: from}.
In \figRef*{FOL proof of from v2} we display the $\iFOLm$ proofs for
$\inferFOL()[\guard{P}]{\all{x}\fromPred{x}{(\fromFun \; x)}}$ that arises
from the CUP proof.
\end{example}

\begin{figure}[bt]
  \begin{prooftree}
    \AxiomC{}
    \RightLabel{\Proj}
    \UnaryInfC{$\inferFOL(x )[\Delta]{
        \all{x}\all{t} \later\fromPred{(s\ x)}{t} \impl
        \fromPred{x}{(\sconsFunc{x}{t})}}$}
    \RightLabel{\AllE, \AllE}
    \UnaryInfC{$\inferFOL(x)[\Delta]{
        \later\fromPred{(s\ x)}{(\fromFun \;{(s\ x)})} \impl
        \fromPred{x}{(\sconsFunc{c}{(\fromFun \;{(s\ x)})})}}$}
    \AxiomC{$\spadesuit$}
    \RightLabel{\ImplE}
    \BinaryInfC{$\inferFOL(x )[\Delta]{
        \fromPred{x}{(\sconsFunc{x}{(\fromFun \;{(s\ x)})})}}$}
    \RightLabel{\Conv}
    \UnaryInfC{$\inferFOL(x )[\Delta]{
        \fromPred{x}{(\fromFun \;{x})}}$}
		\RightLabel{\AllI}
		\UnaryInfC{$\inferFOL()[\guard{P}, \all{x}\later\fromPred{x}{(\fromFun \;{x})}]{
        \all{x}\fromPred{x}{(\fromFun \;{x})}}$}		
    \RightLabel{\LaterPresAllL}
    \UnaryInfC{$\inferFOL()[\guard{P},\later (\all{x}\fromPred{x}{(\fromFun \;{x})})]{
        \all{x}\fromPred{x}{(\fromFun \;{x})}}$}
    \RightLabel{\FP}
    \UnaryInfC{$\inferFOL()[\guard{P}]{
        \all{x}\fromPred{x}{(\fromFun \;{x})}}$}
  \end{prooftree}

\dotfill $\ \spadesuit\ $\dotfill

\begin{prooftree}
  \AxiomC{}
  \RightLabel{\Proj}
  \UnaryInfC{$\inferFOL(x )[\Delta]{
      \all{x} \later(\fromPred{x}{(\fromFun \;{x})})}$}
  \RightLabel{\AllE}
  \UnaryInfC{$\inferFOL(x )[\Delta]{
      \later\fromPred{(s \ x)}{(\fromFun \;{(s\ x)})}}$}
\end{prooftree}

\caption{$\iFOLm$ proof for \protect\exRef{LOB TRAN FROM}.
  $\Delta$ abbreviates $\guard{P}, \all{x} \later (\fromPred{x}{(\fromFun \;{x})})$.}
\label{fig:FOL proof of from v2}
\end{figure}
\end{toappendix}

The results of this section show that it is irrelevant whether the guarding
modality is used on the right (CUP-style) or on the left ($\iFOLm$-style),
as the former can be translated into the latter.
However, CUP uses the guarding on the right to preserve proof uniformity,
whereas $\iFOLm$ extends a general sequent calculus.
Thus, to obtain the reverse translation, we would have to have an admissible
cut rule in CUP.
The main ingredient to such a cut rule is the ability to prove several
coinductive statements simultaneously.
This is possible in CUP by proving the conjunction of these statements.
Unfortunately, we cannot eliminate such a conjunction into one of its components,
since this would require non-deterministic guessing in the proof construction,
which in turn breaks uniformity.
Thus, we leave a solution of this problem for future work.

% give an answer to the common discussion of the
% significance of imposing guarding modalities on the left (in $\iFOLm$ style)
% or the right (in CUP style) of the sequents in coinductive proof rules.
% CUP uses modality on the right to preserve proof uniformity (and thus normal
% form property), whereas $\iFOLm$ gives an account for any proof in
% sequent calculus.
% However, one style of coinductive proofs is convertible into another.

% \input{content/subsections/CUP_to_LOB_Abstract_TransForm}

\section{Herbrand Models and Soundness}
\label{sec:soundness}

In \secRef*{loeb-translation} we showed that coinductive uniform proofs are
sound relative to the intuitionistic logic $\iFOLm$.
This gives us a handle on the constructive nature of coinductive uniform proofs.
Since $\iFOLm$ is a non-standard logic, we still need to provide semantics for
that logic.
We do this by interpreting in \secRef*{soundness-loeb-herbrand} the formulae of
$\iFOLm$ over the well-known (complete) Herbrand models and prove the soundness
of the accompanying proof system with respect to these models.
Although we obtain soundness of coinductive uniform proofs over Herbrand models
from this, this proof is indirect and does not give a lot of information
about the models captured by the different calculi \cofohc{} etc.
For this reason, we will give in \secRef*{soundness-cup-herbrand} a direct
soundness proof for coinductive uniform proofs.
We also obtain coinduction invariants from this proof for each of the calculi,
which allows us to describe their proof strength.

\subsection{Coinductive Herbrand Models and Semantics of Terms}
\label{sec:coalg-sem-guarded-terms}

Before we come to the soundness proofs, we introduce in this section
(complete) Herbrand models by using the terminology of final coalgebras.
We then utilise this description to give operational and denotational semantics
to guarded terms.
These semantics show that guarded terms allow the description and computation
of potentially infinite trees.

The coalgebraic approach has been proven very successful both in
logic and programming~\cite{Abbott06:Container,%HamanaFiore11:GADT,%
  Turner95:StrongFP,vdBerg07:Non-wellfoundedTrees}.
We will only require very little category theoretical vocabulary and assume that
the reader is familiar with the category $\SetC$ of sets and functions,
and functors, see for example~\cite{BarrWells95:CatTheoryComputing,%
  Borceux08:Handbook1,Lambek-Scott88:HOCL}.
The terminology of algebras and
coalgebras~\cite{Aczel02:AlgCoalg,Jacobs16:IntroCoalg,%
  Rutten00:UniversalCoalgebra,%
  Sangiorgi2011:IntroCoind}
is given by the following definition.

\begin{definition}
  \label{def:coalg}
  A \emph{coalgebra} for a functor $F \from \SetC \to \SetC$ is a
  map $c \from X \to FX$.
  Given coalgebras $d : Y \to FY$ and $c \from X \to FX$, we say that
  a map $h : Y \to X$ is a \emph{homomorphism} $d \to c$ if
  $Fh \comp d = c \comp h$.
  We call a coalgebra
  $c \from X \to FX$  \emph{final}, if for every coalgebra
  $d$ there is a unique homomorphism $h \from d \to c$.
  We will refer to $h$ as the \emph{coinductive extension} of $d$.
  % Dually, an \emph{algebra} is a function $a \from FA \to A$, and
  % a \emph{homomorphism} from $a$ to $b \from FB \to B$ is a function $h : A \to B$
  % with $h \comp a = b \comp Fh$.
  % An \emph{initial algebra} is an algebra $a \from FA \to A$, such that for every
  % $F$-algebra $b$ there is a unique homomorphism $h$ from $a$ to $b$.%
  % This homomorphism $h$ is called the \emph{inductive extension} of $b$.
\end{definition}

The idea of (complete) Herbrand models is that a set of Horn clauses determines
for each predicate symbol a set of potentially infinite terms.
Such terms are (potentially infinite) trees, whose nodes are labelled
by function symbols and whose branching is given by the arity of these function
symbols.
To be able to deal with open terms, we will allow such trees to have leaves
labelled by variables.
Such trees are a final coalgebra for a functor determined by the signature.
\begin{definition}
  \label{def:signature-extension}
  Let $\TSig$ be first-order signature.
  The \emph{extension} of a first-order signature $\TSig$ is a
  (polynomial) functor~\cite{GambinoKock:PolyFunctors}
  $\ext{\TSig} : \SetC \to \SetC$ given by
  \begin{equation*}
    \ext{\TSig}(X) = \Coprod*[f \in \TSig] X^{\ar(f)},
    % \quad
    % \text{and}
    % \quad
    % \ar(f \from \underbrace{\iota \to \dotsm \to \iota}_{n} \to \tau) = n,
    % \text{ where } \iota \in \BaseT, \tau \in \BaseT \cup \set{\propT}.
  \end{equation*}
  where $\ar \from \TSig \to \N$ is defined in
  \secRef*{terms-formulae} and $X^n$ is the $n$-fold product of $X$.
  We define for a set $V$ a functor
  $\ext{\TSig} + V \from \SetC \to \SetC$ by
  $(\ext{\TSig} + V)(X) = \ext{\TSig}(X) + V$,
  where $+$ is the coproduct (disjoint union) in $\SetC$.
\end{definition}

To make sense of the following definition, we note that we can view $\PSig$ as
a signature and we thus obtain its extension $\PSigExt$.
Moreover, we note that the final coalgebra of $\TSigExt + V$ exists because
$\TSigExt$ is a polynomial functor.
\begin{definition}
  \label{def:herbrand-stuff}
  Let $\TSig$ be a first-order signature.
  % We denote an initial algebra of the functor $\TSigExt + V$ by
  % $\inC \from \TSigExt(\Hterms{\TSig}(V)) + V \to \Hterms{\TSig}(V)$
  % and refer to $\Hterms{\TSig}(V)$ as \emph{(basic) terms} over $\TSig$.
  The \emph{coterms} over $\TSig$ are given by a final coalgebra
  $\rootO_V \from \coterms{\TSig}(V) \to \TSigExt(\coterms{\TSig}(V)) + V$.
  For brevity, we denote the coterms with no variables,
  i.e. $\coterms{\TSig}(\emptyset)$, by
  $\rootO \from \coterms{\TSig} \to \TSigExt(\coterms{\TSig})$,
  and call it the \emph{(complete) Herbrand universe}
  and its elements \emph{ground} coterms.
  % $\Hterms{\TSig}(\emptyset)$ by $\Hterms{\TSig}$ and call it the \emph{Herbrand universe}, while
  % Elements of
  % % $\Hterms{\TSig}$ and
  % $\coterms{\TSig}$ are called \emph{ground}
  % %terms and
  % coterms. %, respectively.
  Finally, we let
  % the \emph{Herbrand base} $\HBase$ be the set $\PSigExt(\Hterms{\TSig})$ and
  the \emph{(complete) Herbrand base} $\coHBase$
  be the set $\PSigExt(\coterms{\TSig})$.
%   We denote the set of \emph{ground terms}, also called
%   \emph{Herbrand universe},

%   Let $\terms{\TSig}(V)$ be the set of terms over $V$, and
% $\TSig^\infty$ be the set of possibly infinite
% ground terms over $\TSig$.
% A \emph{formula} $\varphi$ is given by
% $Q(\vv{t})$ for some $Q \in \Delta$ and a tuple
% $\vv{t} = (t_1, \dotsc, t_{\ar Q})$ of terms in $\terms{\TSig}(V)$.

\end{definition}

The construction
% $\Hterms{\TSig}(V)$ and
$\coterms{\TSig}(V)$ gives rise to
a functor
%$\Hterms{\TSig} \from \SetC \to \SetC$ and
$\coterms{\TSig} \from \SetC \to \SetC$, called
% the \emph{free monad}~\cite{Barr85:ToposesTriples} and
the \emph{free completely iterative monad}~\cite{Aczel03:InfTrees-CIT}.
% Intuitively,
% % $\Hterms{\TSig}(V)$ consists of finite trees with nodes in $\TSig$
% % and leafs in $V$, whereas
% coterms in $\coterms{\TSig}(V)$ are potentially
% infinite trees with nodes whose labels and branching is given by $\TSig$ and
% with leaves in $V$.
% To help readability, we adapt the usual convention and leave the use of the map
% ``$\inC$'' implicit.
% The set $\Hterms{\TSig}(V)$ is then given inductively by the following two rules.
% \begin{gather*}
%   \AxiomC{$x \in V$}
%   \UnaryInfC{$x \in \Hterms{\TSig}(V)$}
%   \bottomAlignProof
%   \DisplayProof
%   \quad
%   \AxiomC{$f \in \TSig$}
%   \AxiomC{$\vv{t} \in \parens*{\Hterms{\TSig}(V)}^{\ar(f)}$}
%   \BinaryInfC{$f\parens*{\vv{t}} \in \Hterms{\TSig}(V)$}
%   \bottomAlignProof
%   \DisplayProof
% \end{gather*}
If there is no ambiguity, we will drop the injections
$\inj_i$ when describing elements of $\coterms{\TSig}(V)$.
Note that $\coterms{\TSig}(V)$ is final with property that for every
$s \in \coterms{\TSig}(V)$ either there are $f \in \TSig$ and
$\vv{t} \in (\coterms{\TSig}(V))^{\ar(f)}$ with $\rootO_V(s) = f(\vv{t})$, or
there is $x \in V$ with $\rootO_V(s) = x$.
Finality allows us to specify unique maps into $\coterms{\TSig}(V)$ by giving
a coalgebra $X \to \TSigExt(X) + V$.
% or through guarded systems of equations~\cite{Aczel03:InfTrees-CIT}.
In particular, one can define
% an inclusion map
% $\iota_V \from \coterms{\TSig} \to \coterms{\TSig}(V)$ for any $V$ and
for each $\theta \from V \to \coterms{\TSig}$ the substitution $t[\theta]$ of
variables in the coterm $t$ by $\theta$ as the coinductive extension
of the following coalgebra.
% two coalgebras.
\begin{equation*}
  % \SwapAboveDisplaySkip
  % & \coterms{\TSig}
  % \xrightarrow{\rootO} \TSigExt(\coterms{\TSig})
  % \xrightarrow{\inj_1} \TSigExt(\coterms{\TSig}) + V
  % \\
  % &
  \coterms{\TSig}(V)
  \xrightarrow{\rootO_V} \TSigExt(\coterms{\TSig}(V)) + V
  \xrightarrow{\copair{\id, \rootO \comp \theta}} \TSigExt(\coterms{\TSig}(V))
\end{equation*}

Now that we have set up the basic terminology of coalgebras, we can give
semantics to guarded terms from \defRef*{guarded-terms}.
The idea is that guarded terms guarantee that we can always compute with them
so far that we find a function symbol in head position,
see~\lemRef*{guarded-computation}.
This function symbol determines then the label and branching of a node in the
tree generated by a guarded term.
If the computation reaches a constant or a variable, then we stop creating
the tree at the present branch.
This idea is captured by the following lemma.
\begin{lemmarep}
  \label{lem:guarded-term-interpretation}
  There is a map $\foSem{-} \from \GuardedFOTerms{\TSig}(\Gamma) \to \coterms{\TSig}(\Gamma)$
  that is unique with
  \begin{enumerate}
  \item if $M \conv N$, then $\foSem{M} = \foSem{N}$, and
  \item for all $M$, if $M \reduceIter f \> \vv{N}$ then
    $\rootO_{\Gamma} \parens*{\foSem{M}} = f\parens[\big]{\vv{\foSem{N}}}$,
    and if $M \reduceIter x$ then $\rootO_{\Gamma} \parens*{\foSem{M}} = x$.
  \end{enumerate}
\end{lemmarep}
% \begin{proofsketch}
%   By \lemRef*{guarded-computation}, we can define a coalgebra
%   on the quotient of guarded terms by convertibility
%   $c \from \quot{\GuardedFOTerms{\TSig}(\Gamma)}{\conv}
%   \to \TSigExt\parens*{\quot{\GuardedFOTerms{\TSig}(\Gamma)}{\conv}} + \Gamma$
%   with
%   $c[M] = f [\vv{N}]$ if $M \reduceIter f \> \vv{N}$
%   and $c[M] = x$ if $M \reduceIter x$.
%   This yields a homomorphism
%   $h \from \quot{\GuardedFOTerms{\TSig}(\Gamma)}{\conv} \to \coterms{\TSig}(\Gamma)$
%   and we can define $\foSem{-} = h \comp [-]$.
%   The rest follows from uniqueness of $h$.
% \end{proofsketch}
\begin{appendixproof}
  We define a coalgebra
  $c \from \quot{\GuardedFOTerms{\TSig}(\Gamma)}{\conv}
  \to \TSigExt\parens*{\quot{\GuardedFOTerms{\TSig}(\Gamma)}{\conv}} + \Gamma$
  on the quotient of guarded terms by convertibility as follows.
  \begin{equation*}
    c[M] =
    \begin{cases}
      f [\vv{N}], & \text{ if } M \reduceIter f \> \vv{N} \\
      x,  & \text{ if } M \reduceIter x
  \end{cases}
  \end{equation*}
  This is a well-defined map by \lemRef*{guarded-computation}.
  By finality of $\coterms{\TSig}(\Gamma)$, we obtain a unique homomorphism
  $h \from \quot{\GuardedFOTerms{\TSig}(\Gamma)}{\conv} \to \coterms{\TSig}(\Gamma)$.
  This allows us to define $\foSem{-} = h \comp [-]$, which gives
  us immediately for $M \conv N$ that
  $\foSem{M} = h [M] = h [N] = \foSem{N}$.
  Moreover, we have
  \begin{equation*}
    \begin{split}
      \rootO(\foSem{M})
      & = \rootO(h [M]) \\
      & = (\TSigExt(h) + \id)(c [M]) \\
      & =
      \begin{cases}
        \TSigExt(h)(f \> [\vv{N}])
        = f \> \vv{h[N]}
        = f \> \vv{\foSem{N}},
        & \text{ if } M \reduceIter f \> \vv{N}
        \\
        \id(x)
        = x,
        & \text{ if } M \reduceIter x
      \end{cases}
    \end{split}
  \end{equation*}
  Finally, assume that we are given a map
  $k \from \GuardedFOTerms{\TSig}(\Gamma) \to \coterms{\TSig}(\Gamma)$ with the
  above two properties.
  The first allows us to lift $k$ to a map
  $k' \from \quot{\GuardedFOTerms{\TSig}(\Gamma)}{\conv} \to \coterms{\TSig}(\Gamma)$
  with $k' \comp [-] = k$.
  Due to the second property we know that $k'$ is then a coalgebra
  homomorphism and by finality $k' = h$.
  Hence, we obtain from $\foSem{-} = h \comp [-] = k' \comp [-] = k$ that
  $\foSem{-}$ is unique.
  \qedhere
\end{appendixproof}

\begin{toappendix}
Let us illustrate the semantics of guarded terms on our running example.
\begin{example}
  Recall $\fromFun$ from \exRef*{guarded-terms} and
  note that
  $\fromFun \; 0 \reduceIter \scons \; 0 \; (\fromFun \; (s \; 0))$.
  Hence, we have
  $\rootO(\foSem{\fromFun \; 0})
  = \scons \; \foSem{0} \; \foSem{\fromFun \; (s \; 0)}$.
  If we continue unfolding $\foSem{\fromFun \; 0}$, then we obtain the
  infinite tree
  $\scons \; 0
  \longrightarrow \scons \; (s \; 0)
  \longrightarrow \scons \; (s \; (s \; 0))
  \longrightarrow \dotsm$.
  % \begin{equation*}
  %   \begin{tikzcd}[row sep=tiny, column sep=small, shorten >=-2pt]
  %     \scons \arrow[rr] \arrow[dr]
  %     & & \scons \arrow[rr] \arrow[dr]
  %     & & \scons \arrow[rr] \arrow[dr]
  %     & & \dotsm
  %     \\
  %     & 0 & & s \; 0 & & s \; (s \; 0)
  %   \end{tikzcd}
  % \end{equation*}
\end{example}
\end{toappendix}

\subsection{Interpretation of Basic Intuitionistic First-Order Formulae}
\label{sec:interpretation-iFOL}

In this section, we give an interpretation of the formulae in
\defRef*{formulae}, in which we restrict ourselves to guarded terms.
This interpretation will be relative to models in the complete Herbrand
universe.
Since we later extend these models to Kripke models to be able to handle the
later modality, we formulate these models already now in the language of
fibrations~\cite{Benabou1985:FibredCats,Jacobs1999-CLTT}.
% Streicher2018:FibredCats}.

\begin{definition}
  Let $p \from \TCat \to \BCat$ be a functor.
  Given an object $I \in \BCat$, the \emph{fibre} $\TCat_I$ above $I$ is the
  category of objects $A \in \TCat$ with $p(A) = I$ and morphisms
  $f \from A \to B$ with $p(f) = \id_I$.
  The functor $p$ is a \emph{(split) fibration} if for every morphism
  $u \from I \to J$ in $\BCat$ there is functor
  $\reidx{u} \from \TCat_J \to \TCat_I$, such that
  $\reidx{\id_I} = \Id_{\TCat_I}$ and
  $\reidx{(v \comp u)} = \reidx{u} \comp \reidx{v}$.
  We call $\reidx{u}$ the \emph{reindexing along $u$}.
\end{definition}

% To give an interpretation of formulae, let us consider the following two
% categories $\Arit$ and $\Pred$.
% \begin{equation*}
%   \Arit = \CatDescr{
%     n \in \N
%   }{
%     f \from n \to m \text{ is a map } \\
%     & f \from (\coterms{\TSig})^n \to (\coterms{\TSig})^m
%   }
%   \quad
%   \Pred = \CatDescr{
%     (n, P) \text{ with } n \in \N \text{ and } P \subseteq (\coterms{\TSig})^n
%   }{
%     f \from (n, P) \to (m, Q) \text{ is a map } \\
%     & (\coterms{\TSig})^n \to (\coterms{\TSig})^m
%     \text{ with } f(P) \subseteq Q
%   }
% \end{equation*}
% The functor $\predF \from \Pred \to \Arit$ with $\predF(n, P) = n$ and
% $\predF(f) = f$ is a split fibration, see~\cite[Ex.~4.2.3]{Jacobs1999-CLTT},
% where the reindexing functor for $f \from n \to m$ is given by taking preimages:
% $\reidx{f}(Q) = f^{-1}(Q)$.
% Note that each fibre $\Pred_n$ is isomorphic to the complete lattice of $n$-ary
% predicates over $\coterms{\TSig}$ ordered by set inclusion.
% Thus, we refer to this fibration as the \emph{predicate fibration}.

To give an interpretation of formulae, consider the following category
$\PredC$.
\begin{equation*}
  \PredC = \CatDescr{
    (X, P) \text{ with } X \in \SetC \text{ and } P \subseteq X
  }{
    f \from (X, P) \to (Y, Q) \text{ is a map }
    f \from X \to Y
    \text{ with } f(P) \subseteq Q
  }
\end{equation*}
The functor $\predF \from \PredC \to \SetC$ with $\predF(X, P) = X$ and
$\predF(f) = f$ is a split fibration, see~\cite{Jacobs1999-CLTT},
where the reindexing functor for $f \from X \to Y$ is given by taking preimages:
$\reidx{f}(Q) = f^{-1}(Q)$.
Note that each fibre $\PredC_X$ is isomorphic to the complete lattice of
predicates over $X$ ordered by set inclusion.
Thus, we refer to this fibration as the \emph{predicate fibration}.

Let us now expose the logical structure of the predicate fibration.
This will allow us to conveniently interpret first-order formulae over this
fibration, but it comes at the cost of having to introduce a good amount of
category theoretical language.
However, doing so will pay off in \secRef*{soundness-loeb-herbrand}, where
we will construct another fibration out of the predicate fibration.
We can then use category theoretical results to show that this new fibration
admits the same logical structure and allows the interpretation of the later
modality.

The first notion we need is that of fibred products, coproducts and exponents,
which will allow us to interpret conjunction, disjunction and implication.
\begin{definition}
  A fibration $p \from \TCat \to \BCat$ has \emph{fibred finite products}
  $(\Fin, \times)$, if each fibre $\TCat_I$ has finite products
  $(\Fin_I, \times_I)$ and these are preserved by reindexing:
  for all $f \from I \to J$, we have $\reidx{f}(\Fin_J) = \Fin_I$
  and $\reidx{f}(A \times_J B) = \reidx{f}(A) \times_I \reidx{f}(B)$.
  Fibred finite coproducts and exponents are defined analogously.
\end{definition}

The fibration $\predF$ is a so-called first-order fibration, which allows us to
interpret first-order logic, see~\cite[Def.~4.2.1]{Jacobs1999-CLTT}.
\begin{definition}
  A fibration $p \from \TCat \to \BCat$ is a
  \emph{first-order fibration} if%
  \footnote{Technically, the quantifiers should also fulfil the
    Beck-Chevalley and Frobenius conditions, and the fibration should
    admit equality.
    Since these are fulfilled in all our models and we do not need
    equality, we will not discuss them here.}
  \begin{itemize}
  \item $\BCat$ has finite products and the
    fibres of $p$ are preorders;
  \item $p$ has fibred finite products $(\top, \conj)$ and
    coproducts $(\bot, \disj)$ that distribute;
  \item $p$ has fibred exponents $\impl$; and
  \item $p$ has existential and universal quantifiers
    $\exists_{I,J} \dashv \reidx{\pi_{I,J}} \dashv \forall_{I,J}$
    for all projections $\pi_{I,J} \from I \times J \to I$.
    % that are preserved by reindexing.
  \end{itemize}
  A \emph{first-order $\lambda$-fibration} is a first-order fibration
  with Cartesian closed base $\BCat$.
\end{definition}

% The fibration $\predF \from \Pred \to \Arit$ is a first-order fibration, as
% all its fibres are posets and $\Arit$ has finite products given by addition; and
% $\predF$ has fibred finite products $(\top, \cap)$, given by
% $\top_n = (\coterms{\TSig})^n$ and intersection;
% fibred distributive coproducts $(\emptyset, \cup)$;
% fibred exponents $\predImpl$, given by
% $(P \predImpl Q)
% = \setDef{\vv{t}}{\text{if } \vv{t} \in P \text{, then } \vv{t} \in Q}$; and
% universal and existential quantifiers given for
% $P \in \Pred_{n+m}$ by
% \begin{equation*}
%   \forall_{n,m} P  = \setDef{\vv{t} \in (\coterms{\TSig})^n}{
%     \all{\vv{s} (\coterms{\TSig})^m} \vv{t}\vv{s} \in P}
%   \qquad
%   \exists_{n,m} P = \setDef{\vv{t}}{\exist{\vv{s}} \vv{t}\vv{s} \in P}.
% \end{equation*}
% To simplify notation, we will usually denote the quantifier $\forall_{n,1}$
% and $\exists_{n,1}$ by $\forall_n$ and $\exists_n$.

The fibration $\predF \from \PredC \to \SetC$ is a first-order
$\lambda$-fibration, as all its fibres are posets and $\SetC$ is
Cartesian closed;
$\predF$ has fibred finite products $(\top, \cap)$, given by
$\top_X = X$ and intersection;
fibred distributive coproducts $(\emptyset, \cup)$;
fibred exponents $\predImpl$, given by
$(P \predImpl Q)
= \setDef{\vv{t}}{\text{if } \vv{t} \in P \text{, then } \vv{t} \in Q}$; and
universal and existential quantifiers given for
$P \in \PredC_{X \times Y}$ by
\begin{equation*}
  \forall_{X,Y} P
  = \setDef*{x \in X}{\all{y \in Y} (x,y) \in P}
  \quad
  \exists_{X,Y} P
  = \setDef*{x \in X}{\exist{y \in Y} (x,y) \in P}.
\end{equation*}
% \begin{itemize}
% \item $\Arit$ has finite products given by addition;
% \item $\predF$ has fibred finite products $(\top, \cap)$, given by
%   $\top_n = (\coterms{\TSig})^n$ and intersection;
% \item $\predF$ has fibred distributive coproducts $(\emptyset, \cup)$;
% \item $\predF$ has fibred exponents $\predImpl$, where
% %  \begin{equation*}
%   $(P \predImpl Q)
%   = \setDef{\vv{t}}{\text{if } \vv{t} \in P \text{, then } \vv{t} \in Q}$; and
%   % \end{equation*}
% \item $\predF$ has universal and existential quantifiers given for
%   $P \in \Pred_{n+m}$ by
%   \begin{equation*}
%     \forall_{n,m} P  = \setDef{\vv{t} \in (\coterms{\TSig})^n}{
%       \all{\vv{s} (\coterms{\TSig})^m} \vv{t}\vv{s} \in P}
%     \qquad
%     \exists_{n,m} P = \setDef{\vv{t}}{\exist{\vv{s}} \vv{t}\vv{s} \in P}.
%   \end{equation*}
% \end{itemize}

% ,
% fibred distributive coproducts $(\emptyset, \cup)$,
% existential and universal quantifiers
% $\forall_n, \exists_n \from \Pred_{n+1} \to \Pred_n$, where
% for $P, Q \in \Pred_n$ and $R \in \Pred_{n+1}$ one defines
% \begin{equation*}
%   \begin{aligned}[t]
%     \Fin & = (\coterms{\TSig})^n
%     \\
%     \forall_n R & = \setDef{\vv{t}}{\all{s} \vv{t}s \in R}
%     & \exists_n R & = \setDef{\vv{t}}{\exist{s} \vv{t}s \in R}
%   \end{aligned}
% \end{equation*}
% Note that $\predF$ also has fibred equality, but since we will not use it
% explicitly, we will not discuss it any further.

The purpose of first-order fibrations is to capture the essentials
of first-order logic, while the $\lambda$-part takes care of higher-order
features of the term language.
In the following, we interpret types, contexts, guarded terms and
formulae in the fibration $\predF \from \PredC \to \SetC$:
We define for types $\tau$ and context $\Gamma$ sets $\sem{\tau}$
and $\sem{\Gamma}$; for guarded terms $M$ with $\typed{M}{\tau}$ we define a map
$\sem{M} \from \sem{\Gamma} \to \sem{\tau}$ in $\SetC$; and for
a formula $\validForm{\varphi}$ we give a predicate
$\sem{\varphi} \in \PredC_{\sem{\Gamma}}$.

\begin{toappendix}
\begin{remark}
  It should be noted that we give in the following an interpretation over
  concrete fibrations with their base over $\SetC$.
  However, the interpretations could also be given over general first-order
  $\lambda$-fibrations $p \from \TCat \to \BCat$.
  The main issues is to get an interpretation of guarded terms over a final
  coalgebra for $\TSigExt$ in a general category $\BCat$.
  Currently, this interpretation crucially requires the category of sets
  as base category, see \lemRef*{guarded-term-interpretation}.
\end{remark}
\end{toappendix}

The semantics of types and contexts are given inductively in the
Cartesian closed category $\SetC$, where the base type
$\baseT$ is interpreted as coterms, as follows.
\begin{align*}
%  \SwapAboveDisplaySkip
  \sem{\baseT} & = \coterms{\TSig}
  & \sem{\emptyset} & = \Fin
  \\
  \sem{\tau \to \sigma} & = \sem{\sigma}^{\sem{\tau}}
  & \sem{\Gamma, x : \tau} & = \sem{\Gamma} \times \sem{\tau}
\end{align*}

We note that a coterm $t \in \coterms{\TSig}(V)$ can be seen as a map
$(\coterms{\TSig})^V \to \coterms{\TSig}$ by applying a
substitution in $(\coterms{\TSig})^V$ to $t$:
$\sigma \mapsto t[\sigma]$.
In particular, the semantics of a guarded first-order term
$M \in \GuardedFOTerms{\TSig}(\Gamma)$ is equivalently a map
$\foSem{M} \from \sem{\Gamma} \to \coterms{\TSig}$.
We can now extend this map inductively to
$\sem{M} \from \sem{\Gamma} \to \sem{\tau}$ for all guarded terms
$M \in \GuardedTerms{\TSig}(\Gamma)$ with $\typed{M}{\tau}$ by
\begin{align*}
  \sem{M}(\gamma)\parens[\big]{\vv{t}}
  & = \foSem{M \> \vv{x}}\parens[\big]{\brak[\big]{\vv{x} \mapsto \vv{t}}}
  && \guarded[]{M}{\tau}
  \text{ with } \ar(\tau) = \card[\big]{\vv{t}} = \card[\big]{\vv{x}} \\
  \sem{c}(\gamma)\parens[\big]{\vv{t}} & = c \; \vv{t} \\
  \sem{x}(\gamma) & = \gamma(x) \\
  \sem{M \; N}(\gamma) & = \sem{M}(\gamma)\parens[\big]{\sem{N}(\gamma)} \\
  \sem{\lam{x} M}(\gamma)(t) & = \sem{M}(\gamma[x \mapsto t])
\end{align*}

\begin{lemmarep}
  \label{lem:term-interpret-prefunctor}
  The mapping $\sem{-}$ is a well-defined function from guarded terms to
  functions, such that $\typed{M}{\tau}$ implies
  $\sem{M} \from \sem{\Gamma} \to \sem{\tau}$.
\end{lemmarep}
\begin{appendixproof}
  Immediate by induction on $M$.
\end{appendixproof}

Since $\predF \from \PredC \to \SetC$ is a first-order fibration, we can
interpret inductively all logical connectives of the formulae from
\defRef*{formulae} in this fibration.
The only case that is missing is the base case of predicate symbols.
Their interpretation will be given over a Herbrand model that is constructed
as the largest fixed point of an operator over all predicate interpretations
in the Herbrand base.
Both the operator and the fixed point are the subjects of the following
definition.

\begin{definition}
  \label{def:herbrand-model}
  We let the set of \emph{interpretations} $\Interp$ be the powerset
  $\pow{\coHBase}$ of the complete Herbrand base.
  For $I \in \Interp$ and $p \in \PSig$, we denote by $\I{p}$
  the interpretation of $p$ in $I$ (the fibre of $I$ above $p$)
  \begin{equation*}
    \I{p} = \setDef[\big]{\vv{t} \in \parens*{\coterms{\TSig}}^{\ar(p)}}{p(\vv{t}) \in I}.
  \end{equation*}
  Given a set $P$ of $H^g$-formulae, we define a monotone map
  $\ProgOp{P} \from \Interp \to \Interp$ by
  \begin{equation*}
    \ProgOp{P}(I) =
    \setDef{\foSem{\psi} [\theta]}{
      \parens*{\all{\vv{x}} \Conj*[k=1][n] \varphi_k \impl \psi} \in P,
      \theta \from \card{\vv{x}} \to \coterms{\TSig},
      \all{k} \foSem{\varphi_k} [\theta] \in I},
  \end{equation*}
  where $\foSem{-} [\theta]$ is the extension of semantics and
  substitution from coterms to
  the Herbrand base by functoriality of $\PSigExt$.
  The \emph{(complete) Herbrand model} $\Model[P]$ of $P$ is the largest fixed
  point of $\ProgOp{P}$, which exists because $\Interp$ is a complete lattice.
\end{definition}

\begin{toappendix}
  \begin{remark}
  Note that if $P$ is a set of Horn clauses (logic program), then
  the definition of the operator $\ProgOp{P}$ in \defRef*{herbrand-model} just
  becomes
  \begin{equation*}
    \ProgOp{P}(I) =
    \setDef{\psi [\theta]}{
      \parens*{\all{\vv{x}} \Conj*[k=1][n] \varphi_k \impl \psi} \in P,
      \theta \from \card{\vv{x}} \to \coterms{\TSig},
      \all{k} \varphi_k [\theta] \in I},
  \end{equation*}
  as we do not have to unfold fixed point terms.
  Thus, in most cases, except in the proof of \thmRef*{cut theorem infi model},
  we will drop the semantic brackets $\foSem{-}$.
\end{remark}
\end{toappendix}

Given a formula $\varphi$ with $\validForm{\varphi}$ that contains only
guarded terms, we define the semantics of $\varphi$ in $\PredC$ from an
interpretation $I \in \Interp$ inductively as follows.
\begin{align*}
   \SwapAboveDisplaySkip
  \sem{\validForm{p \> \vv{M}}}_I
  & = \reidx{\parens*{\vv{\sem{M}}}} (\I{p})
  \\
  \sem{\validForm{\top}}_I
  & = \top_{\sem{\Gamma}}
  \\
  \sem{\validForm{\varphi \binbox \psi}}_I
  & = \sem{\validForm{\varphi}}_I \binbox \sem{\validForm{\psi}}_I
  & \binbox \in \set{\conj, \disj, \impl}
  % & = \sem{\validForm{\varphi}} \mathbin{\overline{\binbox}} \sem{\validForm{\psi}}
  % & \overline{\conj} = \cap, \overline{\disj} = \cup \text{ and } \overline{\impl} = \predImpl
  \\
  \sem{\validForm{\bind{Q}{x : \tau} \varphi}}_I
  & = Q_{\sem{\Gamma}, \sem{\tau}} \;
  \sem{\validForm[\Gamma, x : \tau]{\varphi}}_I
  & Q \in \set{\forall, \exists}
\end{align*}

\hnote[disable]{The reindexing should be explained in the long version}

\begin{lemmarep}
  \label{lem:formula-interpret-resp-typing}
  The mapping $\sem{-}_I$ is a well-defined function from formulae to predicates,
  such that $\validForm{\varphi}$ implies
  $\sem{\varphi}_I \subseteq \sem{\Gamma}$ or, equivalently,
  $\sem{\varphi}_I \in \PredC_{\sem{\Gamma}}$.
\end{lemmarep}
\begin{appendixproof}
  Immediate by induction on $\varphi$.
\end{appendixproof}

\begin{toappendix}
Let us demonstrate the interpretation of formulae on an example.
\begin{example}
  Recall the formula
  $\all{x \> y} \fromPred{(s \> x)}{y} \impl \fromPred{x}{(scons \> x \> y)}$,
  which we introduced as clause $\kappa_{\fromP 0}$.
  We spell out the interpretation of this formula.
  Note that $\rootO(\sem{s \> x}) = s \> \sem{x} = s \> x$.
  Abusing notation, we write $s \> u$ for $\sem{s \> x}\subst{u/x}$,
  and analogously for the terms $y$, $x$ and $scons \> x \> y$.
  We then have
  \begin{align*}
    \sem{\fromPred{(s \> x)}{y}}_I
    & = \reidx{(\sem{s \> x}, \sem{y})}(\I{\fromP})
    \\
    & = \setDef{(u,v) \in (\coterms{\TSig})^2}{(s \> u, v) \in \I{\fromP}}
  \end{align*}
  Using similar calculations for the other terms in the clause $\kappa_{\fromP 0}$,
  we obtain
  \begin{align*}
    \sem{\kappa_{\fromP 0}}_I
    & =
    \sem{\all{x \> y} \fromPred{(s \> x)}{y} \impl \fromPred{x}{(scons \> x \> y)}}_I
    \\
    & = \forall_1 \forall_2 \> \parens{
      \sem{\fromPred{(s \> x)}{y}}_I \predImpl \sem{\fromPred{x}{(scons \> x \> y)}}_I}
    \\
    % & = \forall_1 \forall_2 \> \setDef{(u,v)}{
    %   \text{if } (s \> u, v) \in I(\fromP )
    %   \text{, then } (u, scons \> u \> v) \in I(\fromP)}
    % \\
    & = \setDef{\ast}{
      \all{u,v}
      \text{if } (s \> u, v) \in \I{\fromP}
      \text{, then } (u, scons \> u \> v) \in \I{\fromP}
    }
  \end{align*}
  As expected, we thus have $\sem{\kappa_{\fromP 0}}_I = \set{\ast}$ if $I$
  validates $\kappa_{\fromP 0}$.
\end{example}
\end{toappendix}

This concludes the semantics of types, terms and formulae.
We now turn to show that coinductive uniform proofs are sound for this
interpretation.

% \begin{align*}
%   \valid{I}()[]{p \> \vv{M}}
%   & \iff \vv{\sem{M}} \in I_p
%   \\
%   \valid{I}()[]{\varphi \binbox \psi}
%   & \iff \valid{I}()[]{\varphi} \binbox \valid{I}()[]{\psi}
%   & \Box \in \set{\conj, \disj, \impl}
%   \\
%   \valid{I}()[]{\bind{Q}{x : \tau} \varphi}
%   & \iff \bind{Q}{M \in \GuardedFOTerms{\TSig}} \valid{I}()[]{\varphi \subst{M/x}}
%   & Q \in \set{\forall, \exists}
% \end{align*}

% and a substitution $\theta$ of guarded terms for variables in $\Gamma$

% A \emph{guarded context morphism} $\Gamma_1 \to \Gamma_2$ is a tuple of
% guarded terms $\vv{t}$ with $\Gamma_2 = x_1 : \tau_1, \dotsc, x_n : \tau_n$
% and $\guarded[\Gamma_1]{t_k}{\tau_k}$.
% \begin{equation*}
%   \Ctx = \CatDescr{
%     \text{contexts } \Gamma
%   }{
%     \text{guarded context morphisms } \vv{t} \from \Gamma_1 \to \Gamma_2
%   }
% \end{equation*}
% \begin{equation*}
%   \iFOL = \CatDescr{
%     \text{formulae } \varphi \text{ with } \validForm{\varphi}
%   }{
%     \vv{t} \from \validForm[\Gamma_1]{\varphi} \to \validForm[\Gamma_2]{\psi}
%     \text{ if }
%     \vv{t} \text{ is guarded context morphism } \Gamma_1 \to \Gamma_2
%     \\
%     & \text{ and }
%     \inferFOL(\Gamma_1)[]{\varphi \impl \psi[\vv{t}]}
%     \text{ is provable in } \iFOL
%   }
% \end{equation*}

% The interpretation of terms and formulae gives rise to a map of fibrations
% \begin{equation*}
%   \begin{tikzcd}
%     \iFOL \dar{p} \rar{\sem{-}}
%     & \Interp % \dar{q}
%     \\
%     \Ctx
%     &
%   \end{tikzcd}
% \end{equation*}

\subsection{Soundness of Coinductive Uniform Proofs for Herbrand Models}
\label{sec:soundness-cup-herbrand}

In this section, we give a direct proof of soundness for the coinductive
uniform proof system from \secRef*{uniform-proofs}.
Later, we will obtain another soundness result by combining the proof
translation from \thmRef*{CUP-iFOLm-sound Version II} with the soundness of
$\iFOLm$ (\thmRef*{loeb-sound-programs} and \ref{thm:chains-sound}).
The purpose of giving a direct soundness proof for uniform proofs is that
it allows the extraction of a coinduction invariant,
see \lemRef*{cup-soundness-invariant}.

The main idea is as follows.
Given a formula $\varphi$ and a uniform proof $\pi$ for
$\SequentCUP[\TSig][P][\varphi]$, we construct an interpretation
$I \in \Interp$ that validates $\varphi$, i.e. $\sem{\varphi}_I = \top$,
and that is contained in the complete Herbrand model $\Model[P]$.
Combining these two facts, we obtain that
$\sem{\varphi}_{\Model[P]} = \top$, and thus the soundness of uniform proofs.

To show that the constructed interpretation $I$ is contained in $\Model[P]$,
we use the usual coinduction proof principle:
%, as it is given in the following
%definition.
\begin{definition}
  % An \emph{invariant for $t \in \coHBase$} is an $I \in \Interp$,
  % such that $t \in I$ and $I$ is a $\ProgOp{P}$-invariant, that is,
  % $I \subseteq \ProgOp{P}(I)$.
  An \emph{invariant for $K \in \Interp$} is a set $I \in \Interp$,
  such that $K \subseteq I$ and $I$ is a $\ProgOp{P}$-invariant, that is,
  $I \subseteq \ProgOp{P}(I)$.
  If $t \in \coHBase$, we also say that $I$ is an invariant for $t$,
  if it is an invariant for $\set{t}$.
\end{definition}

Since $\Model[P]$ is the largest fixed point of $\ProgOp{P}$, we immediately
have that, if $K$ has an invariant, then $K \subseteq \Model[P]$, see also~\cite{KL18-2}.

In the remainder of this section, we will often have to refer to substitutions
by coterms and their composition.
The following definition will make dealing with this easier by organising these
substitutions into a (Kleisli-)category.
These notations are derived from the monad $(\coterms{\TSig}, \eta, \mu)$
with $\eta \from \Id \natTo \coterms{\TSig}$
and $\mu \from \coterms{\TSig}\coterms{\TSig} \natTo \coterms{\TSig}$,
cf.~\cite{Aczel03:InfTrees-CIT}:
\begin{definition}
  A \emph{(Kleisli-)substitution} $\theta$ from $V$ to $W$, written
  $\theta \from V \klTo W$, is a map
  $V \to \coterms{\TSig}(W)$.
  Composition of $\theta \from V \klTo W$ and
  $\delta \from U \klTo V$ is given by
  \begin{equation*}
    \theta \klComp \delta =
    U \xrightarrow{\delta} \coterms{\TSig}(V)
    \xrightarrow{\coterms{\TSig}(\theta)} \coterms{\TSig}(\coterms{\TSig}(W))
    \xrightarrow{\mu_{W}} \coterms{\TSig}(W).
\end{equation*}
\end{definition}

In what follows, we extract, for any instance of a formula $\varphi$, an explicit invariant
from a proof $\pi$
of $\SequentCUP[\TSig][P][\varphi]$, which then yields the soundness of CUP.
More precisely, let $\varphi$ be an $H^g$-formula with
$\varphi = \all{\vv{x}}A_1 \conj\dotsm \conj A_n \impl A_0$ and
let $X$ be the set of variables in $\vv{x}$.
Given a substitution $\theta \from X \klTo \emptyset$, we need to show that
if for all $1 \leq k \leq n$ we have $\foSem{A_k}[\theta] \in \Model[P]$, then
$\foSem{A_0}[\theta] \in \Model[P]$.
The remainder of this section is devoted to constructing an invariant
for $\foSem{A_0}[\theta]$.

We note that a uniform proof
$\pi$ for $\SequentCUP[\TSig][P][\varphi]$ starts with
\begin{prooftree}
  \AxiomC{$\ldots$}
  \UnaryInfC{$\SequentUPg(\vv{c}:\baseT,\TSig)[P][\varphi, hp]{A_0 [\vv{c}/\vv{x}]}$}
  \RightLabel{$\ImplRG$}
  \UnaryInfC{$\SequentUPg(\vv{c}:\baseT,\TSig)[P][\varphi]{A_1 \conj\dotsm \conj A_n \impl A_0 [\vv{c}/\vv{x}]}$}
  \AxiomC{$\vv{c} : \baseT \notin \TSig$}
  \RightLabel{$\AllRG$}
  \BinaryInfC{$\SequentUPg[P][\varphi]{\all{\vv{x}}A_1 \conj\dotsm \conj A_n \impl A_0}$}
  \RightLabel{$\Cofix$}
  \UnaryInfC{$\SequentCUP[\TSig][P][\all{\vv{x}}A_1 \conj\dotsm \conj A_n \impl A_0]$}
\end{prooftree}
where $\varphi = \all{\vv{x}}A_1 \conj\dotsm \conj A_n \impl A_0$,
$hp = (A_1 \conj\dotsm \conj A_n)[\vv{c}/\vv{x}]$, and the eigenvariables $\vv{c}$ are all distinct.
%\footnote{I am assuming the corrected rules, or else $A_1 \conj\dotsm \conj A_n$ would join $P$ in $\pi$}
Let $C$ the set of variables in $\vv{c}$ and $\TSig^C$ the signature $\vv{c} : \baseT, \TSig$.
For brevity, we define $A^C_k = A_k\ssubst{\vv{c}/\vv{x}}$.
Note that from here, the further proof of the given goal will only be based on the signature
$\TSig^C$, that is, no new eigenvariables will be introduced higher in the proof.
Thus, we can focus on $A^C_0$ in our construction of an invariant:
Given a substitution $\theta_0 \from C \klTo \emptyset$, we need to construct
an invariant for $\foSem{A^C_0}[\theta_0]$, given that we already have an invariant for
the assumptions $\foSem{A^C_k}[\theta_0]$ with $1 \leq k \leq n$.

We need to refer to the levels of the proof $\pi$, which is the distance from the root sequent
$\SequentCUP[\TSig][P][\varphi]$.
For example, the above proof tree displays levels $0$ to $3$ of the proof $\pi$.
%In what follows, we will use the \todo{Better would be an upper index C -- OK}
%upper index $c$ to refer  to  formulas in
%which all variables are given  by eigenvariables from  $C$, e.g. we will write
%$A_0^C$ instead of $A_0 [\vv{c}/\vv{x}]$.
%\footnote{note that I am not binding notation $F^C$ here to refer to any  specific  substitution $ [\vv{c}/\vv{x}]$, it is any substutution where variables are substituted by eigenvariables from $C$. }

From here, the proof of $\langle A_0[\vv{c}/\vv{x}] \rangle$ can only proceed by applying
the rule $\DecG$, with a chosen  clause $\kappa$ from $P$.
%\footnote{ I am assuming the corrected rules, or else there will be another,
%  un-productive, way to proceed and choose a hypothesis, i.e.
%  $A_1 \conj\dotsm \conj A_n$}
If $\kappa$ is of the form $\all{\vv{y}}A$, then we define
$I = \setDef{\foSem{A}[\theta]}{\theta \text{ is a substitution}}$.
It is straightforward to show that this is an invariant.
Since $A_0[\vv{c}/\vv{x}] \conv A[\gamma]$ for some
substitution $\gamma$, we can find a substitution $\theta \from Y \klTo \emptyset$, such that
$\foSem{A_0^C}[\theta_0] = \foSem{A}[\theta]$
(put $\theta(y_i) = \foSem{\gamma(y_i)}[\theta_0]$).
Thus, we have that $\foSem{A_0^C}[\theta_0] \in I$, which
shows that $\foSem{A_0^C}[\theta_0] \in \Model[P]$.
%\todo{Why does this use $\foSem{-}$ and not $\sem{-}$?
%  NB: $\foSem{-}$ is not defined on formulas! -- OK, you changed it now, but Thm 34 was formulated with  $\foSem{-}$ on formulas, I assumed that was correct. Also Df 27 uses  $\foSem{-}$ on formulas, no? Either way, please change to whichever is correct}
%\footnote{If the rule is not corrected, we would need to show
%  $I \subseteq \ProgOp{P}(I)$ while $\all{y}A \notin P$.
%  It is impossible to do. One can only start a non-constructive argument of
%  the shape:
%  either   $A_0[\vv{c} / \vv{x}]$ is in the model or it is not in the
%  model, and then consider these two cases... Why would one do it if a
%  constructive proof is possible?  }

Having considered this simple case, we will now analyse the case
when a chosen $\kappa$ is of the form $\all{\vv{y}} \psi \impl \psi'$, and
%\todo{Why would the substitution be given by $\ssubst{\protect\vv{c}/\protect\vv{x}}$
  %if the variables in $\psi$ are $\protect\vv{y}$.
  %This should be a place where a new agent comes from.  --- well $\psi'^C$ does not refer to any particular substution, except for saying it will have eigenvariables from C. Generally, variables in $\kappa$ will be different from those in $A_0$ }
$A_0^C \conv \psi'[\gamma]$ for some substitution $\gamma$.
%\footnote{This of course does not hold if we do not
%  correct the rules. Then $\kappa$ may be also hp}.
In this case, applications of  $\DecG$ and $\ImplL$ and \Init \ will eventually deliver the subgoal:

\begin{prooftree}
  \AxiomC{$\sequentUP{\TSig^C}{P}{\varphi, hp}{}{\psi[\gamma]}$}
  \UnaryInfC{\ldots}
  %  \sequentUP{c,\TSig}{P}{\varphi}{\fromPred{(s\ c)}{(\fromFun \ (s \ c))}}
 \end{prooftree} 

 We will refer to this stage in the proof $\pi$ as
 %\todo{What is $l$? A number? How is is
 %  it defined concretely? -- I defined it now. }
 the \emph{level $l$ in $\pi$}.
  We now consider the rest of the proof tree above level $l$,  i.e. we consider the proof for $\psi[\gamma]$. This is where the non-trivial part of the invariant
 construction will be obtained.  In general, $\psi[\gamma]$ will be given by $A'_1 \conj \ldots \conj A'_m$,
 and the rule $\AndR$ will require  subproofs for each of $A'_j$. Let us consider a sub-proof of $\pi$ for an arbitrary such $A'_j$.
%\footnote{
%   Henning's old text:
   
%Suppose the following is a valid subtree of $\pi$.
%\begin{prooftree}
%  \AxiomC{$\SequentUPb[\TSig^C]{\varphi [\vv{N}/\vv{x}]}$}
%  \RightLabel{$\AllL$}
%  \UnaryInfC{$\SequentUPb[\TSig^C]{\all{\vv{x}} \varphi\in\Delta}$}
%  \RightLabel{\Dec}
%  \UnaryInfC{$\SequentUP(\TSig^C){A}$}
%\end{prooftree}
%Note: it is technically incorrect, as either $P$ or $\Delta$ have changed (depending on the choice of $\ImplR$), and it is strange to talk about them abstractly (if that was the intention), while referring to the same proof tree $\pi$,
%in which both $P$ and $\Delta$ are defined.

%\textcolor{red}{Indeed, they should have different names or should be extensions
%  like $\Delta, \Delta'$.}
%}
The proof can only proceed here by applying the rule \Dec, in which case there are three options: to choose a program clause from $P$, or choose  $hp$ or $\varphi$. Only the latter case is interesting for the invariant construction, as this is where
the coinductive goal $\varphi$ is instantiated, giving rise to a substitution that we will use in the invariant construction:

\begin{prooftree}
  \AxiomC{$\sequentUPb{\TSig^C}{P}{\varphi, hp}{A_1 \conj\dotsm \conj A_n \impl A_0 [\vv{N}/\vv{x}]}{A'_j}$}
  \RightLabel{$\AllL$}
  \UnaryInfC{$\sequentUPb{\TSig^C}{P}{\varphi, hp}{\varphi}{A'_j}$}
  \RightLabel{\Dec}
  \UnaryInfC{$\sequentUP{\TSig^C}{P}{\varphi, hp}{}{A'_j}$}
\end{prooftree}
  %  \sequentUP{c,\TSig}{P}{\varphi}{\fromPred{(s\ c)}{(\fromFun \ (s \ c))}}

Generally,  $\varphi$ may be used with different substitutions
multiple times within the proof tree
$\pi$. However, $\varphi$ itself is uniquely determined by the only use of the rule  $\Cofix$ in the root of $\pi$.
The above fragment of the proof for $A'_j$ gives rise to a substitution
$\rho_1 =  [\vv{N}/\vv{x}]$ that we can extend to a substitution
$\theta_1 \from C \klTo C$ by defining $\theta_1(c_j) = \sem{\rho_1(x_j)}$,
where $c_j$ is the eigenvariable that was initially substituted for $x_j$.
Since $\varphi$ is the goal of the coinductive proof $\pi$, we are lead to use
$\theta_1$ and its iterations in the invariant that will prove coinduction goal $\varphi$.

The notions in the following definition will allow us to easily organise and
iterate the substitutions that occur in a uniform proof. Recall that in general,
$\varphi$ can be used $n$ times in the proof $\pi$, giving rise to $n$ substitutions
$\theta_1, \ldots , \theta_n$.
The following abstract definition is motivated by this observation.

\begin{definition}
  Let $S$ be a set with $S = \set{1, \dotsc, n}$ for some $n \in \N$.
  We call the set $\lists{S}$ of lists over $S$ the set of
  \emph{substitution identifiers}.
  Suppose that we have substitutions $\theta_0 \from V \klTo \emptyset$ and
  $\theta_k \from V \klTo V$ for each $k \in S$.
  Then we can define a map
  $\Theta \from \lists{S} \to \parens*{\coterms{\TSig}}^V$,
  which turns each substitution identifier into a substitution, by iteration
  from the right:
  \begin{equation*}
    \Theta (\varepsilon) = \theta_0
    \quad \text{and} \quad
    \Theta (w : k) = \Theta(w) \klComp \theta_k
  \end{equation*}
\end{definition}

%After introducing this notation, we can give the model construction.

%\hnote{The reference to $\varphi$ in what follows goes very far back!
%  If I understand the definition correctly, then
%  $\delta(c) = N_x$, where $N_x$ is the term substituted for $x$ in
%  $\varphi$ in the little proof tree above and $c$ is the constant
%  substituted for $x$ in the beginning.
%}

Coming back to the analysis of the proof $\pi$,
we assign to each substitution $\rho_i = [\vv{N}/\vv{x}]$ with $n \geq 1$,
which arises from a use of $\varphi$ in the proof tree $\pi$,
a substitution $\theta_i \from C \klTo C$ by
$\theta_i(c_j) = \sem{\rho_i(x_j)}$.
Note that each $N_j$ in $\rho_i$ has only variables from $C$, that is, $\Sigma^C \vdash N_j: \tau$.
We call $\theta_i$ an \emph{agent} of $\pi$.

We let $D \subseteq \foGuardedAt$ be the set of atoms that are proven in $\pi$:
\begin{equation*}
  D = \setDef{A}{
    \SequentUPg(\TSig^C){A} \text{ or } \SequentUP(\TSig^C){A}
    \text{ appears in } \pi
  }
\end{equation*}
From the agents and atoms in $\pi$ we extract an invariant
for the goal formula.
In the following lemma we take $S = \{ 1, \ldots , n\}$ to be the set of identifiers for
the $n$ uses of $\varphi$ in the given proof $\pi$.
%\footnote{$S$ needs to be well-defined for a proof of the Lemma. We never actualy bind $S$ in the Lemma statement, so I am binding it here.} 

\begin{lemmarep}
  \label{lem:cup-soundness-invariant}
  Suppose that $\varphi$ is an $H^g$-formula of the form
  $\all{\vv{x}}A_1 \conj\dotsm \conj A_n \impl A_0$
  and that there is a proof $\pi$ for $\SequentCUP[\TSig][P][\varphi]$.
  Let $D$ be the proven atoms in $\pi$, $\theta_1, \dotsc, \theta_n$ be
  the agents of $\pi$
  and $\theta_0 \from C \klTo \emptyset$ some initial substitution.
  Define $A^C_k = A_k\ssubst{\vv{c}/\vv{x}}$
  and
  suppose further that $I_1$ is an invariant for
  $\setDef*{\foSem{A^C_k}[\Theta(\varepsilon)]}{1\leq k \leq n}$.
  If we put
  \begin{equation*}
    I_2= \bigcup_{w \in \lists{S}}  \foSem{D} \ssubst*{\Theta(w)}
  \end{equation*}
  then $I = I_1 \cup I_2$ is an invariant for $\foSem{A^C_0} \ssubst*{\Theta(\varepsilon)}$.
\end{lemmarep}
\begin{appendixproof}
  This proof refers to the notation and the construction of the proof $\pi$ as
  given above.

  We first need show that  $\foSem{A^C_0} \ssubst*{\Theta(\varepsilon)} \in I$,
  which follows trivially from the fact that
  $ \foSem{A^C_0} \ssubst*{\Theta(\varepsilon)}  \in I_2$ by construction of $I_2$.
  It remains to show that $I \subseteq  \ProgOp{P}(I)$. We consider two cases: either
  $y \in I_1$ or $y \in I_2$. If the former, then $y \in  \ProgOp{P}(I)$ by
  definition of $I_1$ as an invariant and monotonicity of $\ProgOp{P}$.

  Consider the case  when $y \in I_2$, i.e. when $y \in  \foSem{D} \ssubst*{\Theta(w)}$
  for some $w$, and therefore $y = \foSem{B}  \ssubst*{\Theta(w)}$ for some atom $B$
  proven in $\pi$.
  We have to show that $y \in  \ProgOp{P}(I)$, that is, we have to show that
  there is a clause $\all{\vv{y}}\bigwedge_{k=1}^m \psi_k \impl \psi'$ and a
  substitution $\theta \from \card{\vv{y}} \klTo \emptyset$ such that
  $y = \foSem{\psi'}[\theta]$ and for all $1 \leq k \leq m$ we have
  $\foSem{\psi_k}[\theta] \in I$.
  We show that by case analysis on the proof of $B$ and induction on $w$.

  As discussed in the outline of proof $\pi$, only the rules \Dec{} or $\DecG$
  are  applicable to an atomic goal, and there are 3 possibilities of choosing
  a formula via these: it may be a program clause from $P$, the hypothesis $hp$
  or $\varphi$.
  When one of these options is taken in a proof,
  we will say $B$ is \emph{resolved against} a clause from $P$, $hp$ or
  $\varphi$, respectively.
  Moreover, chosing an atomic clause in $P$ gives us the base case, for which the proof has
  been given already.
  The remaining cases are:

  \begin{enumerate}[label=\alph*)]

  \item Suppose $B$ is resolved against a clause
    $\all{\vv{y}}\bigwedge_{k=1}^m B_k \impl \psi'$ in $P$, that is, we have
    $\psi' \ssubst*{\gamma} \conv B$ for some substitution $\gamma$.
    Note that if we define $\theta = \Theta(w) \klComp \foSem{\gamma}$, then
    \begin{equation*}
      y
      = \foSem{B}\ssubst*{\Theta(w)}
      = \foSem{\psi' \ssubst*{\gamma}}\ssubst*{\Theta(w)}
      = \foSem{\psi'}\ssubst*{\Theta(w) \klComp \foSem{\gamma}}
      = \foSem{\psi'}\ssubst*{\theta}.
    \end{equation*}
    Since for all $1 \leq k \leq m$ the atom $B_k[\gamma]$ must have a proof
    somewhere in $\pi$, we have $B_k[\gamma] \in D$.
    Thus, also $\foSem{B_k}\ssubst*{\theta} \in I_2$ and so
    $y \in \ProgOp{P}(I)$ with the initial program clause and the substitution
    $\theta$.
  \item If $B$ is resolved against $\varphi$,
    then this can only occur above the level $l$ in the proof tree $\pi$.
    By the already given
    schematic analysis of $\pi$,
    we also know that this case
    requires that $B \conv A_0 [\gamma_r]$ for some $\gamma_r$, and moreover
    this substitution  is already incorporated in the construction of $\Theta$,
    as $\gamma_r$ gives rise to an agent $\theta_r$.
    Thus, we have
    \begin{equation*}
      y
      = \foSem{B}  \ssubst*{\Theta(w)}
      = \foSem{A_0 [\gamma_r]}  \ssubst*{\Theta(w)}
      = \foSem{A_0^C} \ssubst*{\Theta(w) \klComp \theta_r}
      = \foSem{A_0^C} \ssubst*{\Theta(w:r)},
    \end{equation*}
    the latter equality follows from the definition of $\Theta(w:r)$ and of
    % \todo{and a lemma that relates substitution and semantics}
    substitution composition.% , cf. also~\cite{KL18-2}.

    Note that $A_0^C$ was initially, below level $l$, resolved against a
    program clause $\all{\vv{y}}\bigwedge_{k=1}^m \psi_k \impl \psi'$
    with $A_0^C \conv \psi'[\gamma]$ for some substitution $\gamma$.
    Thus, by putting $\theta = \Theta(w:r) \klComp \foSem{\gamma}$ we further have
    \begin{equation*}
      y
      = \foSem{\psi'[\gamma]} \ssubst*{\Theta(w:r)}
      = \foSem{\psi'} \ssubst*{\Theta(w:r) \klComp \foSem{\gamma}}
      = \foSem{\psi'} \ssubst*{\theta}.
    \end{equation*}
    Since each $\psi_k[\gamma]$ is an atom in $\pi$, we also have that
    $\psi_k[\gamma] \in D$ and thus
    $\foSem{\psi_k[\gamma]}\ssubst*{\Theta(w:r)} \in I_2$.
    This yields in turn that $\foSem{\psi_k}\ssubst*{\theta} \in I_2$.
    Putting this all together, we have that
    $y \in \ProgOp{P}(I)$ by using the initial program clause and the substitution $\theta$.
  \item $B$ is resolved against $hp$.
    Once again, this can only occur
    above the level $l$ in the proof tree $\pi$.
    Since $hp = A_1 \conj\dotsm \conj A_n\ssubst{\vv{c}/\vv{x}}$,
    we know that this case requires that
    $B \conv A_k\ssubst{\vv{c}/\vv{x}}$, for some $A_k$  in $A_1 \conj\dotsm \conj A_n$.
    Thus,
    \begin{equation*}
      y
      = \foSem{B}  \ssubst*{\Theta(w)}
      = \foSem{A_k\ssubst{\vv{c}/\vv{x}}} \ssubst*{\Theta(w)}
      = \foSem{A_k^C}  \ssubst*{\Theta(w)}.
    \end{equation*}
    We proceed now by distinguishing the cases for $w$ and invoking the induction hypothesis.
    \begin{enumerate}[label=\roman*)]
    \item $w = \varepsilon$. In this case, $y = \foSem{A^C_k}  \ssubst*{\Theta(w)}$
      is already in $I_1$, and because $I_1$ is an invariant for
      $\setDef*{\foSem{A^C_k}  \ssubst*{\Theta(\varepsilon)}}{1\leq k \leq n}$,
      we have $y \in  \ProgOp{P}(I_1) $ and hence $y \in  \ProgOp{P}(I)$.
    \item $w = v:i$ for some $i \in S$ and $v \in \lists{S}$.
      Let $\gamma_i$ be the syntactic substitution from which the agent $\theta_i$
      arises.
      Then we have that
      $y
      = \foSem{A_k^C}  \ssubst*{\Theta(v) \klComp \theta_i}
      = \foSem{A_k[\gamma_i]} \ssubst*{\Theta(v)}$.
      Since $\theta_i$ is an agent, there must be a use of $\varphi$ in $\pi$ with the
      substitution $\gamma_i$.
      Thus, the premises of $\varphi$, and in particular $A_k[\gamma_i]$, must all have a
      proof in $\pi$.
      From this we obtain that $A_k[\gamma_i] \in D$ and so
      $\foSem{A_k[\gamma_i]} \ssubst*{\Theta(v)} \in I_2$.
      By induction, we obtain now, from any of the cases in this proof, that
      $\foSem{A_k[\gamma_i]} \ssubst*{\Theta(v)} \in \ProgOp{P}(I)$,
      as required.
    \end{enumerate}
  \end{enumerate}
  This induction and case analysis shows that for any $y \in I$, we have
  $y \in \ProgOp{P}(I)$.
  Thus, $I$ is an invariant.
\end{appendixproof}

Once we have \lemRef*{cup-soundness-invariant} the following soundness
theorem is easily proven.
\begin{theoremrep}
  \label{thm:cup-soundness}
  If $\varphi$ is an $H^g$-formula and $\SequentCUP[\TSig][P][\varphi]$, then
  $\sem{\varphi}_{\Model[P]} = \top$.
\end{theoremrep}
\begin{appendixproof}
  We construct an invariant $I$ for any instance of $\SequentCUP[\TSig][P][\varphi]$,
  as per \lemRef*{cup-soundness-invariant}.
  Since $I \subseteq  \ProgOp{P}(I)$, we obtain $\sem{\varphi}_{\Model[P]} = \top$.
\end{appendixproof}

Finally, we show that extending logic programs with coinductively proven
lemmas is sound.
This follows easily by coinduction.
%% The restrictive statement from the submitted paper.
% \begin{theoremrep}
%   \label{thm:cut theorem infi model}
%   Let $\varphi$ be an $H^g$-formula of the shape
%   $\all{\vv{x}} \psi_1 \impl \psi_2$ and
%   $\vv{M} \in \parens*{\GuardedFOTerms{\TSig}}^{\card{\vv{x}}}$.
%   If $\SequentCUP[\TSig][P][\varphi]$ and
%   $\sem*{\psi_1 \ssubst*{\vv{M}/\vv{x}}}_{\Model[P]} = \top$,
%   then $\Model[P \cup \set*{\psi_2 \ssubst*{\vv{M}/\vv{x}}}] = \Model[P]$.
%   Hence, $P \cup \set*{\psi_2 \ssubst*{\vv{M}/\vv{x}}}$ is a conservative
%   extension of $P$ with respect to the complete Herbrand model.
% \end{theoremrep}
% \begin{appendixproof}
%   Let $\Model = \Model[P]$ and
%   $\Model' = \Model[P \cup \set*{\psi_2 \ssubst*{\vv{M}/\vv{x}}}]$.
%   One first shows that $\ProgOp{P \cup Q} = \ProgOp{P} \sqcup \ProgOp{Q}$
%   for any set $Q$ of ground atoms.
%   The direction $\Model \subseteq \Model'$ follows easily from this by
%   coinduction.
%   For the other direction, one uses soundness and coinduction again.
% \end{appendixproof}
\begin{theoremrep}
  \label{thm:cut theorem infi model}
  Let $\varphi$ be an $H^g$-formula.\footnote{I removed the condition now}
  %of the shape
  %$\all{\vv{x}} \psi_1 \impl \psi_2$, such that,
  %for all substitutions $\theta$ if
  %$\foSem{\psi_1}[\theta] \in \Model[P,\varphi]$, then
  %$\foSem{\psi_1}[\theta] \in \Model[P]$.
  Then $\SequentCUP[\TSig][P][\varphi]$ implies
  $\Model[P \cup \set{\varphi}] = \Model[P]$,
  that is, $P \cup \set{\varphi}$ is a conservative
  extension of $P$ with respect to the Herbrand model.
\end{theoremrep}
\begin{appendixproof}
  Suppose $\varphi$ is an $H^g$-formula of the shape
  $\all{\vv{x}} \psi_1 \impl \psi_2$.
  Let $\Model = \Model[P]$ and $\Model' = \Model[P, \varphi]$.
  First, we note that
  $\ProgOp{P, \varphi} = \ProgOp{P} \lub \ProgOp{\varphi}$,
  where $\lub$ is the point-wise union.
  This gives us immediately that
  $\Model \subseteq \ProgOp{P}(\Model) \subseteq \ProgOp{P,\varphi}(\Model)$
  and thus $\Model \subseteq \Model'$ by coinduction.
  For the other direction, that is $\Model' \subseteq \Model$ one uses soundness
  and coinduction as follows.
  We have
  \begin{align*}
    \Model'
    & = \ProgOp{P,\varphi}(\Model')
    \tag{by definition of $\Model'$} \\
    & = \ProgOp{P}(\Model')
    \cup \setDef{\foSem{\psi_2}[\theta]}{
      \theta \from \vv{x} \to \coterms{\Sigma},
      \foSem{\psi_1}[\theta] \in \Model'}
    \tag{by definition of $\ProgOp{\varphi}$} \\
    & \subseteq \ProgOp{P}(\Model')
    \cup \setDef{\foSem{\psi_2}[\theta]}{
      \theta \from \vv{x} \to \coterms{\Sigma},
      \foSem{\psi_1}[\theta] \in \Model}
    \tag{by assumption} \\
    & \subseteq \ProgOp{P}(\Model') \cup \Model
    \tag{by $\SequentCUP[\TSig][P][\varphi]$ and \thmRef*{cup-soundness}}
  \end{align*}
  Now we use the soundness of a so-called
  up-to-technique~\cite{Pous-UpToComplLattices}.
  Specifically, let $F$ be the monotone map on $\Interp$ given by
  $F(I) = I \cup \Model$.
  Then $F$ is \emph{$\ProgOp{P}$-compatible}, that is,
  $F \comp \ProgOp{P} \sqsubseteq \ProgOp{P} \comp F$ because
  $\Model$ is the largest $\ProgOp{P}$ fixed point.
  It follows~\cite{Pous-UpToComplLattices} for every $I \in \Interp$ that
  whenever $I \subseteq \ProgOp{P}(F(I))$ then $I \subseteq \Model$.
  By the above calculation, we have that
  $\Model' \subseteq \ProgOp{P}(F(\Model'))$.
  Thus, $\Model' \subseteq \Model$ as we wanted to show.
  Altogether, this gives us that
  $\Model[P] = \Model[P,\varphi]$.
  \qedhere
\end{appendixproof}

As a corollary we obtain that, if there is a proof for
$\SequentCUP[\TSig][P][\varphi]$, then a proof for
$\SequentCUP[\TSig][P,\varphi][\psi]$ is sound with respect to  $\Model[P]$.
Indeed, by \thmRef*{cut theorem infi model} we have that
$\Model[P] = \Model[P \cup \set{\varphi}]$ and by \thmRef*{cup-soundness}
that $\SequentCUP[\TSig][P,\varphi][\psi]$ is sound with respect to
$\Model[P \cup \set{\varphi}]$.
Thus, the proof of $\SequentCUP[\TSig][P,\varphi][\psi]$ is also sound  with
respect to $\Model[P]$.
We use this property implicitely in our running examples, and refer the reader
to~\cite{BKL18,KL18-2,KL18} for proofs, further examples and discussion.

% Previous Henning's version without modifications:
% As a consequence of this theorem, we can prove a formula $\varphi$ by
%a coinductive uniform proof for $\SequentCUP[\TSig][P][\varphi]$
%and then prove another formula $\psi$ by extending the logic program $P$ with
%atomic instances $A$ of $\varphi$.
%In other words, proofs for $\SequentCUP[\TSig][P,A][\psi]$ are
%sound with respect to $\Model[P]$.
%This method is discussed in~\cite{KL18} and examples are given in the extended
%version of this paper~\cite{BKL18}.

\subsection{Soundness of $\iFOLm$ over Herbrand Models}
\label{sec:soundness-loeb-herbrand}

In this section, we demonstrate how the logic $\iFOLm$ can be interpreted over
Herbrand models.
Recall that we obtained a fixed point model from the monotone map
$\ProgOp{P}$ on interpretations.
In  what follows, it is crucial that we construct the greatest fixed point
of $\ProgOp{P}$ by iteration,
c.f.~\cite{Adamek03:FinCoalgContinuousFunctors,%
  Cousot1979:ConstructiveFixedPoint,Worrell05:FinalSeqFinitary}:
Let $\Ord$ be the class of all ordinals equipped with their
(well-founded) order.
We denote by $\op{\Ord}$ the class of ordinals with their reversed order
and define a monotone function
$\finChain{\ProgOp{P}} \from \op{\Ord} \to \Interp$, where we write the
argument ordinal in the subscript, by
\begin{equation*}
  \parens[\big]{\finChain{\ProgOp{P}}}_{\alpha}
  = \bigcap\nolimits_{\beta < \alpha}
  \ProgOp{P}\parens[\big]{\finChain{\ProgOp{P}}_{\beta}}.
\end{equation*}
Note that this definition is well-defined because $<$ is well-founded
and because $\ProgOp{P}$ is monotone, see~\cite{Basold18:BreakingTheLoop}.
Since $\Interp$ is a complete lattice, there is an ordinal $\alpha$ such that
$\finChain{\ProgOp{P}}_{\alpha}
= \ProgOp{P}\parens[\big]{\finChain{\ProgOp{P}}_{\alpha}}$,
at which point $\finChain{\ProgOp{P}}_{\alpha}$ is the largest fixed point
$\Model[P]$ of $\ProgOp{P}$.
In what follows, we will utilise this construction to give semantics to
$\iFOLm$.

The fibration $\predF \from \PredC \to \SetC$ gives rise to another
fibration as follows.
We let $\ch{\PredC}$ be the category of functors (monotone maps)
with fixed predicate domain:
\begin{equation*}
  \ch{\PredC} = \CatDescr{
    u \from \op{\Ord} \to \PredC,
    \text{such that }
    \predF \comp u \text{ is constant}
  }{
    u \to v \text{ are natural transformations } f \from u \natTo v, \\
    & \text{such that }
    \predF f \from \predF \comp u \natTo \predF \comp v
    \text{ is the identity}
  }
\end{equation*}
The fibration $\ch{\predF} \from \PredC \to \SetC$ is defined by
evaluation at any ordinal (here $0$), i.e.
by $\ch{\predF}(u) = \predF(u(0))$ and $\ch{\predF}(f) = (\predF f)_0$,
% \begin{equation*}
%   \ch{\predF}(u) = \predF(u(0))
%   \qquad
%   \ch{\predF}(f) = (\predF f)_0
% \end{equation*}
and reindexing along $f \from X \to Y$ by applying the reindexing of $\predF$
point-wise, i.e. by $\reidx[\#]{f}(u)_\alpha = \reidx{f}(u_\alpha)$.
% \begin{equation*}
%   \reidx[\#]{f}(u)_\alpha = \reidx{f}(u_\alpha)
%   \qquad
%   \reidx[\#]{f}(g)_\alpha = \reidx{f}(g_\alpha).
% \end{equation*}

Note that there is a (full) embedding $K \from \PredC \to \ch{\PredC}$ that is
given by $K(X, P) = (X, \ch{P})$ with $\ch{P}_\alpha = P$.
One can show~\cite{Basold18:BreakingTheLoop} that $\ch{\predF}$ is again a
first-order fibration and that it models the later modality, as in the
following theorem.
\begin{theorem}
  \label{thm:chains-fo-fibrations-with-later}
  The fibration $\ch{\predF}$ is a first-order fibration.
  If necessary, we denote the first-order connectives by
  $\predCTop$, $\predCConj$ etc. to distinguish them from those in $\PredC$.
  Otherwise, we drop the dots.
  Finite (co)products and quantifiers are given point-wise, while
  for $X \in \SetC$ and $u,v \in \ch{\PredC}_X$ exponents are given by
  \begin{equation*}
    (v \predCImpl u)_{\alpha} =
    \bigcap\nolimits_{\beta \leq \alpha} (v_\beta \predImpl u_\beta).
  \end{equation*}
  There is a fibred functor
  $\later \from \ch{\PredC} \to \ch{\PredC}$
  with $\ch{\pi} \comp \later = \ch{\pi}$
  given on objects by
  \begin{equation*}
    (\later u)_{\alpha} = \bigcap\nolimits_{\beta < \alpha} u_\beta
  \end{equation*}
  and a natural transformation $\nextOp \from \Id \natTo \later$
  from the identity functor to $\later$.
  The functor $\later$ preserves reindexing, products, exponents and
  universal quantification:
  $\later(\reidx[\#]{f}u) = \reidx[\#]{f}(\later u)$,
  $\later(u \conj v) = \later u \conj \later v$,
  $\later (u^v) \to (\later u)^{\later v}$,
  $\later \parens*{\ch{\forall}_n u} = \ch{\forall}_n \parens{\later u}$.
  % \begin{align*}
  %   \later(\reidx[\#]{f}u) & = \reidx[\#]{f}(\later u)
  %   & \later(u \conj v) & = \later u \conj \later v
  %   \\
  %   \later (u^v) & \to (\later u)^{\later v}
  %   & \later \parens*{\ch{\forall}_n u} & = \ch{\forall}_n \parens{\later u}
  % \end{align*}
  Finally, for all $X \in \SetC$ and $u \in \ch{\PredC}_X$, there is
  $\lob \from (\later u \predCImpl u) \to u$ in $\ch{\PredC}_X$.
\end{theorem}

\begin{toappendix}
  Intuitively, the later modality shifts a given sequence by one position
  and concatenates it with the terminal object.
  This can be seen if we have a description ordinals through successor and
  limit ordinals.
  Given $\sigma \in \ch{\PredC}_X$, we can visualise the beginning of
  $\sigma$ and $\later \sigma$ as follows.
  \begin{equation*}
    \begin{tikzcd}[arrows={phantom}, labels={above}]
      \sigma:
      & \sigma_0
      & \lar{\supseteq} \sigma_1
      & \lar{\supseteq} \sigma_2
      & \lar{\supseteq} \sigma_3
      & \lar{\supseteq} \sigma_4
      & \lar{\supseteq} \dotsb
      \\
      \later \sigma:
      & X \ar[u, "\supseteq"{rotate=90}]
      & \lar{\supseteq} \ar[u, "\supseteq"{rotate=90}] \sigma_0
      & \lar{\supseteq} \ar[u, "\supseteq"{rotate=90}] \sigma_1
      & \lar{\supseteq} \ar[u, "\supseteq"{rotate=90}] \sigma_2
      & \lar{\supseteq} \ar[u, "\supseteq"{rotate=90}] \sigma_3
      & \lar{\supseteq} \dotsb
    \end{tikzcd}
  \end{equation*}
\end{toappendix}

Using the above theorem, we can extend the interpretation of formulae to
$\iFOLm$ as follows.
Let $u \from \op{\Ord} \to \Interp$ be a descending sequence of interpretations.
As before, we define the restriction of $u$ to a predicate symbol $p \in \PSig$
by
$\parens[\big]{\I[u]{p}}_{\alpha} = \I[u_\alpha]{p}
= \setDef[\big]{\vv{t}}{p\parens[\big]{\vv{t}} \in u_{\alpha}}$.
The semantics of formulae in $\iFOLm$ as objects in $\ch{\PredC}$ is given
by the following iterative definition.

\begin{align*}
  \SwapAboveDisplaySkip
  \sem{\validForm{p \> \vv{M}}}_u
  & = \reidx[\#]{\parens*{\vv{\sem{M}}}} (\I[u]{p})
  \\
  \sem{\validForm{\top}}_u
  & = \predCTop_{\sem{\Gamma}}
  \\
  \sem{\validForm{\varphi \binbox \psi}}_u
  & = \sem{\validForm{\varphi}}_u \binbox \sem{\validForm{\psi}}_u
  & \Box \in \set{\conj, \disj, \impl}
  \\
  \sem{\validForm{\bind{Q}{x : \tau} \varphi}}_u
  & = Q_{\sem{\Gamma}, \sem{\tau}} \> \sem{\validForm[\Gamma, x : \tau]{\varphi}}_u
  & Q \in \set{\forall, \exists}
  \\
  \sem{\validForm{\later \varphi}}_u
  & = \later \sem{\validForm{\varphi}}_u
\end{align*}

The following lemma is the analogue of \lemRef*{formula-interpret-resp-typing}
for the interpretation of formulae without the later modality.
\begin{lemmarep}
  \label{lem:loeb-formula-interpret-resp-typing}
  The mapping $\sem{-}_u$ is a well-defined map from formulae in $\iFOLm$ to
  sequences of predicates, such that $\validForm{\varphi}$ implies
  $\sem{\varphi}_u \in \ch{\PredC}_{\sem{\Gamma}}$.
\end{lemmarep}
\begin{appendixproof}
  Immediate by induction on $\varphi$.
\end{appendixproof}

\begin{lemmarep}
  \label{lem:loeb-sound}
  All rules of $\iFOLm$ are sound with respect to the interpretation
  $\sem{-}_u$ of formulae in $\ch{\PredC}$, that is,
  if \ $\inferFOL{\varphi}$, then
  $\parens[\big]{
    \Conj[\psi \in \Delta] \sem{\psi}_u \predCImpl \sem{\varphi}_u
  } = \predCTop$.
  In particular, $\inferFOL[]{\varphi}$ implies $\sem{\varphi}_u = \predCTop$.
\end{lemmarep}
\begin{appendixproof}
  The soundness for the rules of first-order logic in
  \figRef*{rules-std-connectives} is standard for the given interpretation over
  a first-order fibration as in \thmRef*{chains-fo-fibrations-with-later},
  see \cite[Sec.~4.3]{Jacobs1999-CLTT}.
  Soundness of the rules for the rules of the later modality in
  \figRef*{rules-later} follows from the existence of the morphisms $\nextOp$
  and $\lob$, and functoriality of $\later$ that were proved in
  \thmRef*{chains-fo-fibrations-with-later},
  cf.~\cite[Sec.~5.2]{B18} and~\cite{Basold18:BreakingTheLoop}.
\end{appendixproof}

The following lemma shows that the guarding of a set of formulae is valid in
the chain model that they generate.
% that arises from the chain that we used to construct the fixed point model.
\begin{lemmarep}
  \label{lem:loeb-guarded-formulae-sound}
  If $\varphi$ is an $H$-formula in $P$,
  then
  $\sem{\guard{\varphi}}_{\finChain{\ProgOp{P}}} = \predCTop$.
\end{lemmarep}
\begin{appendixproof}
  Let $\varphi$ be an $H$-formula in $P$ of shape
  $\all{\vv{x} : \vv{\tau}}
  \Conj[i=1,\dotsc,n] p_i \; \vv{M}_i \impl q \; \vv{N}$
  and let $\Gamma$ be the context $\vv{x} : \vv{\tau}$.
  Our goal is to show that
  $\sem{\guard{\varphi}}_{\finChain{\ProgOp{P}}} = \predCTop$.
  First, we have by definition of the semantics for all
  $\alpha \in \Ord$ that
  \begin{align*}
    &
    \sem{\Conj*[i=1,\dotsc,n] \later (p_i \; \vv{M}_i)
      \impl q \; \vv{N}}_{\alpha}
    \\
    & =
    \bigcap\nolimits_{\beta < \alpha}
    \sem{\Conj*[i=1,\dotsc,n] \later (p_i \; \vv{M}_i)}_{\beta}
    \predImpl \sem{q \; \vv{N}}_{\beta}
    \\
    & =
    \bigcap\nolimits_{\beta < \alpha}
    \parens*{\bigcap\nolimits_{i=1,\dotsc,n}
      \sem{\later (p_i \; \vv{M}_i)}_{\beta}}
    \predImpl \sem{q \; \vv{N}}_{\beta}
    \\
    & =
    \bigcap\nolimits_{\beta < \alpha}
    \parens*{
      \bigcap\nolimits_{i=1,\dotsc,n}
      \bigcap\nolimits_{\gamma < \beta}
      \sem{p_i \; \vv{M}_i}_{\gamma}}
    \predImpl \sem{q \; \vv{N}}_{\beta}
    \\
    & =
    \begin{aligned}[t]
      \bigcap\nolimits_{\beta < \alpha}
      \big\{
      \sigma \in \sem{\Gamma} % \from \card{\vv{x}} \to \coterms{\TSig}
      \; \big| \;
      & \parens[\big]{\all{i} \all{\gamma < \beta}
        \vv{M}_i[\sigma] \in \finChain{\ProgOp{P}}_{\gamma}(p)}
      \\
      & \implies \vv{N}[\sigma] \in \finChain{\ProgOp{P}}_{\beta}(q)
      \big\}
    \end{aligned}
    \\
    & =
    \begin{aligned}[t]
      \bigcap\nolimits_{\beta < \alpha}
      \big\{
      \sigma \in \sem{\Gamma} % \from \card{\vv{x}} \to \coterms{\TSig}
      \; \big| \;
      & \parens[\big]{\all{i} \all{\gamma < \beta}
        \vv{M}_i[\sigma] \in \finChain{\ProgOp{P}}_{\gamma}(p)}
      \\
      & \implies \parens[\big]{
        \all{\gamma < \beta}
        \vv{N}[\sigma]
        \in \ProgOp{P}\parens[\big]{\finChain{\ProgOp{P}}}_{\gamma}(q)}
      \big\}
    \end{aligned}
  \end{align*}
  We intend to show now that this set is equal to $\predTop_{\sem{\Gamma}}$.
  Let $\sigma \in \sem{\Gamma}$, such that
  $\all{i} \all{\gamma < \beta}
  \vv{M}_i[\sigma] \in \finChain{\ProgOp{P}}_{\gamma}(p)$.
  We have to show that
  $\all{\gamma < \beta}
  \vv{N}[\sigma]
  \in \ProgOp{P}\parens[\big]{\finChain{\ProgOp{P}}}_{\gamma}(q)$.
  To this end, suppose $\gamma < \beta$.
  Then
  $\all{i} \vv{M}_i[\sigma] \in \finChain{\ProgOp{P}}_{\gamma}(p)$
  by assumption.
  By definition of $\ProgOp{P}$ we obtain
  $\vv{N}[\sigma]
  \in \ProgOp{P}\parens[\big]{\finChain{\ProgOp{P}}}_{\gamma}(q)$
  as required.
  Hence,
  $\sem{\Conj*[i=1,\dotsc,n] \later (p_i \; \vv{M}_i)
    \impl q \; \vv{N}} = \predCTop_{\sem{\Gamma}}$.
  But then
  $\sem{\guard{\varphi}}_{\finChain{\ProgOp{P}}} =
  \forall_{\Gamma} \predCTop_{\sem{\Gamma}} = \predCTop$.
\end{appendixproof}

Combining this with soundness from \lemRef*{loeb-sound}, we obtain that
provability in $\iFOLm$ relative to a logic program $P$ is sound for
the model of $P$.
\begin{theoremrep}
  \label{thm:loeb-sound-programs}
  For all logic programs $P$, if
  $\inferFOL[\guard{P}]{\varphi}$ then
  $\sem{\varphi}_{\finChain{\ProgOp{P}}} = \predCTop$.
\end{theoremrep}
\begin{appendixproof}
  Combine \lemRef*{loeb-sound} and \lemRef*{loeb-guarded-formulae-sound}.
\end{appendixproof}

The final result of this section is to show that the descending
chain model, which we used to interpret formulae of $\iFOLm$, is sound and
complete for the fixed point model, which we used to interpret the formulae
of coinductive uniform proofs.
This will be proved in \thmRef*{chains-sound} below.
The easiest way to prove this result is by establishing a functor
$\ch{\PredC} \to \PredC$ that maps the chain $\finChain{\ProgOp{P}}$
to the model $\Model[P]$, and that preserves and reflects truth of first-order
formulae (\propRef*{limit-functor}).
We will phrase the preservation of truth of first-order formulae by a functor
by appealing to the following notion of fibrations maps,
cf.~\cite[Def.4.3.1]{Jacobs1999-CLTT}.
\begin{definition}
  \label{def:map-fo-fibrations}
  Let $p \from \TCat \to \BCat$ and $q \from \Cat{D} \to \Cat{A}$ be fibrations.
  A \emph{fibration map} $p \to q$ is a pair
  $(F \from \TCat \to \Cat{D}, G \from \BCat \to \Cat{A})$ of functors,
  s.t. $q \comp F = G \comp p$
  % the following diagram commutes
  % \begin{equation*}
  %   \begin{tikzcd}
  %     \TCat \dar{p} \rar{F}
  %     & \Cat{D} \dar{q} \\
  %     \BCat \rar{G}
  %     & \Cat{A}
  %   \end{tikzcd}
  % \end{equation*}
  and $F$ preserves Cartesian morphisms: if $f \from X \to Y$ in $\TCat$ is
  Cartesian over $p(f)$, then $F(f)$ is Cartesian over $G(p(f))$.
  % $(F,G)$ is a map of \emph{first-order ($\lambda$-)fibrations}, if
  % $p$ and $q$ are first-order ($\lambda$-)fibrations, and $F$ and $G$ preserve
  % this structure.
\end{definition}

Let us now construct a first-order $\lambda$-fibration map
$\ch{\PredC} \to \PredC$.
We note that since every fibre of the predicate fibration is a complete lattice,
for every chain $u \in \ch{\PredC}_X$ there exists an ordinal $\alpha$ at which
$u$ stabilises.
This means that there is a limit $\lim u$ of $u$ in $\PredC_X$,
% , which we denote by $\lim u$,
which is the largest subset of $X$, such that
$\all{\alpha} \lim u \subseteq u_\alpha$.
% , with the property that
% \begin{equation*}
%   x \in \lim u \iff \all{\alpha} x \in u_\alpha.
% \end{equation*}
% Equivalently, $\lim u$ is the largest subset of $X$, such that
% $\all{\alpha} \lim u \subseteq u_\alpha$.
This allows us to define a map $L \from \ch{\PredC} \to \PredC$ by
\begin{align*}
  % \SwapAboveDisplaySkip
  & L(X, u) = \parens*{X, \lim u} \\
  & L(f \from (X, u) \to (Y, v)) = f.
\end{align*}

In the following proposition, we show that $L$ gives us the ability to
express first-order properties of limits equivalently through their
approximating chains.
This, in turn, provides soundness and completeness for the interpretation of
the logic $\iFOLm$ over descending chains with respect to the largest Herbrand
model.
\begin{propositionrep}
  \label{prop:limit-functor}
  The functor $L \from \ch{\PredC} \to \PredC$, as defined above, is a map of
  fibrations and preserves fibred (co)products, and existential and universal
  quantification.
  Furthermore, $L$ is right-adjoint to the embedding
  $K \from \PredC \to \ch{\PredC}$.
  Finally, for each $p \in \PSig$ and $u \in \ch{\PredC}_{\coHBase}$,
  we have $L\parens[\big]{\I[u]{p}} = \I[L(u)]{p}$.
\end{propositionrep}
\begin{appendixproof}
  First, we show that if $f \from (X, u) \to (Y, v)$, then $f$ is indeed a
  morphism $\parens*{X, \lim u} \to \parens*{Y, \lim v}$.
  This means that we have to show that $f(\lim u) \subseteq \lim v$.
  By the limit property, it suffices to show for all $\alpha \in \Ord$
  that $f(\lim u) \subseteq v_\alpha$:
  \begin{align*}
    f(\lim u)
    & \subseteq f(u_\alpha)
    \tag*{$\lim u \subseteq u_\alpha$ and image of $f$ monotone} \\
    & \subseteq v_\alpha
    \tag*{$f$ is morphism $(X, u) \to (Y, v)$}
  \end{align*}
  That $L$ preserves identities and composition is evident, as is the
  preservation if indices: $\ch{\pi} = \pi \comp L$.

  Next, we show that Cartesian morphisms are preserved as well.
  Let $f \from (X, u) \to (Y, v)$ be Cartesian in $\ch{\PredC}$,
  and suppose we are given $g$ and $h$ as in the lower triangle in the following
  diagram in $\SetC$ and $(Z,P)$ in $\PredC$.
  \begin{equation*}
    \begin{tikzpicture}[commutative diagrams/every diagram, node distance={0.5cm and 0.5cm}]
      \node (cPred) {
        \begin{tikzcd}
          (Z, w) \ar[dr,"g"] \ar[d,dashed,"h"] &       \\
          (X, u) \ar[r,"f"]                    & (Y, v)
        \end{tikzcd}
      };
      \node[below right=of cPred] (Set) {
        \begin{tikzcd}
          Z \ar[dr,"g"] \ar[d,"h"] &    \\
          X \rar{f}                & Y
        \end{tikzcd}
      };
      \node[above right=of Set] (Pred) {
        \begin{tikzcd}
          (Z, P) \ar[dr,"g"] \ar[d,dashed,"h"] & \\
          (X, \lim u) \rar{f}                  & (Y, \lim v)
        \end{tikzcd}
      };
      \path[commutative diagrams/.cd, every arrow, every label, mapsto]
      (cPred) edge node        {$L$}        (Pred)
              edge node [swap] {$\ch{\pi}$} (Set.north west)
      (Pred)  edge node        {$\pi$}      (Set.north east)
      ;
    \end{tikzpicture}
  \end{equation*}
  We have to show that $h$ is a morphism $(Z,P) \to (X, \lim u)$ in $\PredC$.
  To that end, we define a constant chain $w \from \Ord \to \PredC_Z$ by
  $w_\alpha = P$.
  Note that $\lim w = P$, thus $L(Z, w) = (Z,P)$.
  Moreover, for all $\alpha \in \Ord$ we have that
  $g(w_\alpha) = g(Z) \subseteq \lim v$.
  Thus, $g(w_\alpha) \subseteq v_\alpha$ and $g$ is a morphism in $\ch{\PredC}$.
  Since $f$ is Cartesian, we obtain that $h$ is a morphism $(Z, w) \to (X,u)$
  in $\ch{\PredC}$, that is, for all $\alpha$,
  $h(P) = h(w_\alpha) \subseteq u_\alpha$.
  This gives us in turn that $h(P) \subseteq \lim u$, which means that
  $h$ is a morphism $(Z,P) \to (Y, \lim v)$ in $\PredC$.

  Showing that $L$ is preserves coproducts and existential quantifiers is somewhat nasty,
  while products and universal quantification are straightforward.
  First, we prove that conjunction is preserved, that is, we want to prove
  that $\lim (u \predCDisj v) = \lim u \disj \lim v$.
  We note now that, because $u$ and $v$ are descending, that there are ordinals
  $\alpha, \beta, \gamma$ such that
  $\lim (u \predCDisj v) = (u \predCDisj v)_\gamma$ and
  $\lim u \disj \lim v = u_\alpha \disj v_\beta$.
  Let now $\gamma' = \alpha \lub \beta \lub \gamma$ be the least upper bound of
  these ordinals.
  Then we have by the above assumptions that
  \begin{align*}
%    \SwapAboveDisplaySkip
    \lim (u \predCDisj v)
    & = (u \predCDisj v)_\gamma \\
    & = (u \predCDisj v)_{\gamma'}
    \tag*{Descending chains}
    \\
    & = u_{\gamma'} \disj v_{\gamma'}
    \tag*{Point-wise def. of $\predCDisj$}
    \\
    & = u_{\alpha} \disj v_{\beta}
    \tag*{Descending chains}
    \\
    & = \lim u \disj \lim v.
  \end{align*}
  Thus, $L\parens[\big]{(X, u) \predCDisj (X, v)} = L(X, u) \disj L(X, v)$
  as desired.

  Similarly, to prove $\lim (\predCExists_{X,Y} \sigma) = \exists_{X,Y} (\lim \sigma)$,
  we let $\beta$ be such that $\lim \sigma = \sigma_\beta$.
  The inclusion $\exists_{X,Y} (\lim \sigma) \subseteq L(\predCExists_{X,Y} \sigma)$
  is, as usual, unconditionally true.
  For the other direction, we have
  \begin{align*}
    x \in L(\predCExists_{X,Y} \sigma)
    & \iff x \in \lim_{\alpha} \setDef{x \in X}{\exist{y \in Y} (x,y) \in \sigma_\alpha} \\
    & \iff \all{\alpha} \exist{y} (x,y) \in \sigma_\alpha \\
    & \implies \exist{y} (x,y) \in \sigma_\beta \\
    & \implies \exist{y} \all{\alpha} (x,y) \in \sigma_\alpha
    \tag*{$\sigma$ descending and stable at $\beta$} \\
    & \iff x \in \exists_{X,Y} (\lim \sigma).
  \end{align*}
  This proves that also existential quantification is preserved by $L$.

  Finally, to show that there is an adjunction $K \dashv L$, we have to show
  for all $(X, P) \in \PredC$ and $(Y, u) \in \ch{\PredC}$
  that there is a natural isomorphism
  $\Hom{\PredC}{(X, P)}{(Y, \lim u)}
  \cong \Hom{\ch{\PredC}}{(X, \ch{P})}{(Y, u)}$.
  This boils down to showing that for any map $f \from X \to Y$ we have
  $f(P) \subseteq \lim u \iff \all{\alpha} f(P) \subseteq u_\alpha$.
  In turn, this is immediately given by the limit property of $\lim u$.
\end{appendixproof}

We get from \propRef*{limit-functor} soundness of
$\finChain{\ProgOp{P}}$ for Herbrand models.
More precisely, if $\varphi$ is a goal formula that has only implication-free formulas
on the left of an implication (first-order goal), then its interpretation in the
coinductive Herbrand model is true if its interpretation over the chain approximation of the
Herbrand model is true.
\begin{theoremrep}
  \label{thm:chains-sound}
  % For each $(p \from \tau \to \propT) \in \PSig$ and $P \in \PredC_{\sem{\tau}}$
  % there is a one-to-one
  % correspondence between inclusions of $P$ into $\Model[P]$ and
  % inclusions of the constant sequence $K_P$ into $\finChain{\ProgOp{P}}$.%
  If $\varphi$ is a first-order goal and
  $\sem{\varphi}_{\finChain{\ProgOp{P}}} = \predCTop$,
  then
  $\sem{\varphi}_{\Model[P]} = \top$.
\end{theoremrep}
% \begin{proofsketch}
%   First, one shows for all $\later$-free formulae $\varphi$ that
%   $L(\sem{\varphi}_{\finChain{\ProgOp{P}}}) = \sem{\varphi}_{\Model[P]}$
%   by induction on $\varphi$ and using \propRef*{limit-functor}.
%   Using this identity and $K \dashv L$,
%   the result is then obtained from the following adjoint correspondence.
%   \begin{equation*}
%     \begin{adjunction}[1.5]
%       & \predCTop = K(\top) & \sem{\varphi}_{\finChain{\ProgOp{P}}}
%       & & \text{in } \ch{\PredC}
%       \\ \adjointLine
%       & \top
%       & L\parens[\big]{\sem{\varphi}_{\finChain{\ProgOp{P}}}} = \sem{\varphi}_{\Model[P]}
%       & & \text{in } \PredC
%       \tag*{\qed}
%     \end{adjunction}
%   \end{equation*}
% \end{proofsketch}
\begin{appendixproof}
  First, we show for an implication-free $D$-formula $\psi$ that
  \begin{equation}
    \label{eq:limit-pres-semantics}
    L(\sem{\psi}_{\finChain{\ProgOp{P}}}) = \sem{\psi}_{\Model[P]}
  \end{equation}
  by induction on $\psi$ and using \propRef*{limit-functor} as follows.
  For atoms, we have that
  \begin{align*}
    L(\sem{p \; \vv{M}}_{\finChain{\ProgOp{P}}})
    & = L \parens*{
      \reidx[\#]{\sem{\vv{M}}}\parens[\big]{\I[\finChain{\ProgOp{P}}]{p}}
    }
    \\
    & = \reidx{\sem{\vv{M}}}\parens*{
      L\parens[\big]{\I[\finChain{\ProgOp{P}}]{p}}}
    \tag*{$L$ preserves reindexing}
    \\
    & = \reidx{\sem{\vv{M}}}\parens*{
      \I[L\parens[\big]{\finChain{\ProgOp{P}}}]{p}}
    \tag*{$L$ preserves restrictions}
    \\
    & = \reidx{\sem{\vv{M}}}\parens*{\I[\Model[P]]{p}}
    \tag*{$\Model[P]$ is limit of $\finChain{\ProgOp{P}}$}
    \\
    & = \sem{p \; \vv{M}}_{\Model[P]}.
  \end{align*}
  The cases for universal quantification and conjunction are given by using that $L$
  preserves these connectives (again \propRef*{limit-functor}).
  From this, we obtain for a first-order goal $\varphi$ that
  $L(\sem{\varphi}_{\finChain{\ProgOp{P}}}) \subseteq \sem{\varphi}_{\Model[P]}$
  by induction on $\varphi$ and using again \propRef*{limit-functor}.

  To show that the semantics over $\PredC$ and $\ch{\PredC}$ coincide, that is,
  that we have the following correspondence.
  \begin{equation*}
    \def\fCenter{\ = \ }
    \Axiom$\sem{\varphi}_{\finChain{\ProgOp{P}}} \fCenter \predCTop$
    \UnaryInf$\sem{\varphi}_{\Model[P]} \fCenter \top$
    \DisplayProof
  \end{equation*}
  Since any predicate is included in the maximal predicate $\top$, it suffices
  to show that there is a correspondence as in
  \begin{equation*}
    \def\fCenter{\ \to \ }
    \Axiom$\predCTop \fCenter \sem{\varphi}_{\finChain{\ProgOp{P}}}$
    \UnaryInf$\top \fCenter \sem{\varphi}_{\Model[P]}$
    \DisplayProof
  \end{equation*}
  Note that $\predCTop$ is given by the embedding $K(\top)$.
  Using \propRef*{limit-functor} and \eqRef{limit-pres-semantics} we obtain the
  desired correspondence as follows.
  \begin{equation*}
    \def\fCenter{\ \to \ }
    \Axiom$\predCTop = K(\top) \fCenter \sem{\varphi}_{\finChain{\ProgOp{P}}} \qquad\, \text{in } \ch{\PredC}$
    \doubleLine
    \UnaryInf$\phantom{\predCTop = K(\top)} \top \fCenter L\parens[\big]{\sem{\varphi}_{\finChain{\ProgOp{P}}}} \quad\text{in } \PredC$
    \UnaryInf$\top \fCenter \sem{\varphi}_{\Model[P]}$
    \DisplayProof
  \end{equation*}
  % \begin{equation*}
  %   \begin{adjunction}[1.5]
  %     & \predCTop = K(\top) & \sem{\varphi}_{\finChain{\ProgOp{P}}}
  %     & & \text{in } \ch{\PredC}
  %     \\ \adjointLine
  %     % & K(\top) & \sem{\varphi}_{\finChain{\ProgOp{P}}}
  %     % & & \text{in } \ch{\PredC}
  %     % \\ \adjointLine
  %     & \top
  %     & L\parens[\big]{\sem{\varphi}_{\finChain{\ProgOp{P}}}} \to \sem{\varphi}_{\Model[P]}
  %     & & \text{in } \PredC
  %     % \\ \adjointLine
  %     % & \top & \sem{\varphi}_{\Model[P]}
  %     % & & \text{in } \PredC
  %   \end{adjunction}
  % \end{equation*}
  This concludes the proof of soundness for first-order goals
  with respect to the Herbrand model.
\end{appendixproof}

\section{Conclusion, Related Work and the Future}
\label{sec:conclusion}

In this paper, we provided a comprehensive theory of resolution in
coinductive Horn-clause theories and coinductive logic programs.
This theory comprises of a uniform proof system that features a form
of guarded recursion and that provides operational semantics for proofs
of coinductive predicates.
Further, we showed how to translate proofs in this system into proofs
for an extension of intuitionistic FOL with guarded recursion, and we
provided sound semantics for both proof systems in terms of coinductive
Herbrand models.
The Herbrand models and semantics were thereby presented in a modern style
that utilises coalgebras and fibrations to provide a conceptual view
on the semantics.

\paragraph*{Related Work.}
% \label{sec:related}

It may be surprising that automated
\emph{proof search for coinductive predicates} in first-order logic does not
have a coherent and comprehensive theory, even after
three decades~\cite{Aczel88,Park81:AutomataInfSeq},
despite all the attention that it received as
programming~\cite{AbelPTS13,Capretta05:GeneralRecCoind,Hagino-Dialg,%
  Howard96:IndCoindPointedTypes}
and proof~\cite{Dax06:ProofSysLinearTimeMuCalc,%
  EndrullisHHP015,GieslABEFFHOPSS17,%
  Gimenez98,Hur13:ParameterizedCoind,Niwinski96:GamesMuCalc,%
  Rutten00:UniversalCoalgebra,Sangiorgi2011:IntroCoind,%
  Santocanale02:CircProofs,Santocanale02:muBicompleteParity}
method.
The work that comes close to algorithmic proof search is the system
CIRC~\cite{RosoLucanu09:CircCoinduction}, but it cannot handle general
coinductive predicates and corecursive programming.
% However, CIRC is limited to equality proofs, and cannot handle general
% coinductive predicates and corecursive programming.
Inductive and coinductive data types are also being added to SMT
solvers~\cite{BPR18,ReynoldsK15}.
However, both CIRC and SMT solving are inherently based on classical logic
and are therefore not suited to situations where proof objects are relevant,
like programming, type class inference or (dependent) type theory.
Moreover, the proposed solutions, just like those
in~\cite{GuptaBMSM07,SimonBMG07} can only deal with regular data, while our
approach also works for irregular data, as we saw in the $\fromP$-example.

This paper subsumes Haskell type class inference~\cite{Lammel:2005,FKS15} and
exposes that the inference presented in those papers
corresponds to coinductive proofs in $\cofohc$ and $\cohohh$.
Given that the proof systems proposed in this paper are
constructive and that uniform proofs provide proofs (type inhabitants) in normal
form, we could give a propositions-as-types interpretation to all eight
coinductive uniform proof systems.
This was done for $\cofohc$ and $\cohohh$ in~\cite{FKS15}, but we leave the
remaining cube from the introduction for future work.

\paragraph*{Future Work.}
\label{sec:future-work}

There are several directions that we wish to pursue in the future.
First, we know that CUP is incomplete for the presented models, as
it is intuitionistic and it lacks an admissible cut rule.
The first can be solved by moving to Kripke/Beth-models, as done
by Clouston and Goré~\cite{Clouston15:SequentCalcToposOfTrees} for
the propositional part of $\iFOLm$.
However, the admissible cut rule is more delicate.
To obtain such a rule one has to be able to prove several
propositions simultaneously by coinduction, as discussed at the end of
\secRef*{loeb-translation}.
In general, completeness of recursive proof systems depends largely
on the theory they are applied to,
see~\cite{Simpson17:CyclicArithmeticPeano}
and~\cite{Berardi17:MLTT-IndNotCyclic}.
However, techniques from cyclic proof
systems~\cite{Brotherston11:SequentCalcIndInfDecent,%
  Shamkanov14:CircProofsProvabilityLogic} may help.
We also aim to extend our ideas to other situations like
higher-order Horn clauses~\cite{HashimotoU15,BurnOR17}
and interactive proof
assistants~\cite{Agda:system,BaeldeCGMNTW14,Coq94,BlanchetteM0T17},
typed logic programming,
and logic programming that mix inductive and coinductive predicates.

\paragraph*{Acknowledgements.}
We would like to thank Damien Pous and the anonymous reviewers for their
valuable feedback.

%% Acknowledgments
%\begin{acks}                            %% acks environment is optional
                                        %% contents suppressed with 'anonymous'
  %% Commands \grantsponsor{<sponsorID>}{<name>}{<url>} and
  %% \grantnum[<url>]{<sponsorID>}{<number>} should be used to
  %% acknowledge financial support and will be used by metadata
  %% extraction tools.
 % This material is based upon work supported by the
 % \grantsponsor{GS100000001}{National Science
 %   Foundation}{http://dx.doi.org/10.13039/100000001} under Grant
 % No.~\grantnum{GS100000001}{nnnnnnn} and Grant
 % No.~\grantnum{GS100000001}{mmmmmmm}.  Any opinions, findings, and
 % conclusions or recommendations expressed in this material are those
 % of the author and do not necessarily reflect the views of the
 % National Science Foundation.
%\end{acks}

%% Bibliography
\bibliographystyle{abbrvnat}
\bibliography{katya2,CorecUniformProof}
%% With biblatex:
% \printbibliography

%% Appendix
% \clearpage
% \appendix
% \input{content/Appendices/notation}
% %\input{content/Appendices/motivate_guard}
% % \input{content/Appendices/ex}
% \input{content/Appendices/proofs}
\end{document}